\documentclass[11pt]{article}
\usepackage[a4paper,margin=1in]{geometry}
\usepackage{amsmath,amssymb,amsthm,mathtools}
\usepackage{enumitem}
\usepackage{microtype}
\usepackage{hyperref}
\hypersetup{colorlinks=true,linkcolor=blue,citecolor=blue,urlcolor=blue}
\usepackage{braket}
\usepackage{comment}
\usepackage{xcolor}
\usepackage{authblk}
\usepackage{tikz}
\usetikzlibrary{fit,positioning}

\newtheorem{theorem}{Theorem}[section]
\newtheorem{proposition}[theorem]{Proposition}
\newtheorem{lemma}[theorem]{Lemma}
\newtheorem{corollary}[theorem]{Corollary}
\newtheorem{definition}[theorem]{Definition}

\newcommand{\cH}{\mathcal H}

\newcommand{\cav}{\operatorname{cav}}
\newcommand{\eps}{\varepsilon}

\title{The $I_{3322}$ Bell inequality requires infinite dimensions}
\author{Andrea Coladangelo\thanks{Email: \texttt{coladan@cs.washington.edu}}}
\affil{Paul G. Allen School of Computer Science and Engineering, University of Washington}
\date{}

\begin{document}
\maketitle

\begin{abstract}
The $I_{3322}$ inequality is one of the simplest bipartite Bell inequalities, with three possible questions and two possible answers per party. Yet, despite its simplicity, it has long been conjectured by Pál and Vértesi (Physical Review A, 2010) to possess the following quirk: no finite-dimensional quantum strategy can attain its maximal violation, but an infinite-dimensional strategy can. A Bell correlation, with five questions and three answers per party, that provably possesses the same property has since been discovered by Coladangelo and Stark (Nature Communications, 2020). However, proving that the same property holds for $I_{3322}$, which lives in the simplest Bell scenario in which such a phenomenon could occur, has remained elusive. Here, we provide a proof of this conjecture. 

An approximate proof, containing all the main ideas, was produced by GPT 5.5 Pro after various rounds of interaction. The presentation of this proof was revised substantially by the author for completeness, correctness, and clarity.
\end{abstract}

\medskip
\medskip
\medskip
\medskip
\begin{center}
\setlength{\fboxsep}{8pt}
\fbox{%
  \parbox{0.90\textwidth}{%
    \small
    \textbf{AI usage statement.}
   An approximate proof (containing all of the main ideas) was produced by GPT 5.5 Pro after some rounds of interaction with the author. That first approximate proof was much shorter, and arguably quite difficult to parse by a human. The presentation in this paper was the result of substantial revisions by the author for completeness, correctness, and clarity, with help from GPT 5.6 Sol.
  }%
}
\end{center}

\newpage
\tableofcontents

\section{Introduction}
Bell inequalities have fascinated quantum information scientists for decades. They show that measurements performed on spatially separated, but entangled, quantum systems can produce correlations that cannot be reproduced by any classical local model. Beyond the fundamental result that such a separation exists, one can ask a quantitative question: \emph{how large of a quantum system is needed to achieve the maximal violation of a given Bell inequality?} For most familiar examples, very small amounts of entanglement suffice.  For example, a single EPR pair suffices to attain the optimal quantum value of the CHSH inequality, while two EPR pairs suffice to win the magic-square game with probability one. More generally, the optimal strategies for many of the best-known Bell inequalities and nonlocal games live in low-dimensional Hilbert spaces.

A remarkably simple inequality appears to behave very differently.  In 1981, Froissart studied the space of bipartite Bell scenarios in which each party has three possible measurements, each with two possible outcomes~\cite{froissart1981constructive}. Besides liftings of the CHSH inequality, Froissart discovered one genuinely new facet of the local correlation polytope: the inequality now known as \(I_{3322}\). Nearly three decades later, Pál and Vértesi uncovered what seems to be a striking feature of its quantum violation~\cite{PV10}. They constructed strategies of increasing local dimension whose values continue to improve with the dimension, together with a candidate infinite-dimensional limiting strategy, and found compelling numerical agreement with upper bounds obtained from the Navascués--Pironio--Acín hierarchy. Based on this evidence, they conjectured that \emph{no finite-dimensional strategy attains the optimal quantum value of \(I_{3322}\)}.

The phenomenon suggested by Pál and Vértesi was mysterious at the time: it was not even known whether a finite Bell scenario could distinguish finite- from infinite-dimensional entanglement.  Considerable progress has since been made on closely related questions.  Slofstra showed that the set of quantum correlations is not closed~\cite{slofstra2019set}, with subsequent work discovering the same phenomenon in substantially smaller Bell scenarios~\cite{dykema2019non, coladangelo2020two, beigi2021separation}. In these examples, the optimal value can only be achieved in the limit of finite-dimensional strategies, but not by any fixed finite- or even infinite-dimensional strategy.  A different phenomenon was established by Coladangelo and Stark~\cite{coladangelo2020inherently}, who constructed a finite Bell correlation that is attained exactly by an infinite-dimensional strategy but cannot be attained in any finite dimension. The latter Bell correlation uses five questions and three answers per party. This was later simplified by Beigi to just four questions and two answers per party~\cite{beigi2021separation}. Thus both kinds of infinite-dimensional behavior are now known to occur in finite Bell scenarios: an optimum may exist only as a limit, or it may be attained, but only by a genuinely infinite-dimensional strategy.

Despite this progress, the Pál--Vértesi conjecture for \(I_{3322}\) has remained open.  This case is particularly compelling because \(I_{3322}\) lives in the smallest binary-output Bell scenario in which such infinite-dimensional behavior could occur. Note that with only \emph{two} binary measurements for any one of the parties, the problem reduces, through the usual Jordan's lemma decomposition, to optimization over a state of Schmidt rank $2$. \emph{Three} measurements per party are therefore the first place where the infinite-dimensional behavior can arise.  In this sense, \(I_{3322}\) is the minimal candidate for a Bell inequality whose maximal quantum violation inherently requires infinite-dimensional entanglement.

\subsection{Our result}
In this work, we resolve the Pál--Vértesi conjecture. We show that the infinite-dimensional behavior discussed above occurs in the simplest possible binary-output Bell scenario. 
\begin{theorem}[Informal]
No finite-dimensional strategy achieves the optimal quantum value of the $I_{3322}$ inequality. On the other hand, there is an infinite-dimensional strategy that attains it exactly.
\end{theorem}

We give a detailed overview of how the result is proven in the technical overview (Section~\ref{sec:technical-overview}). Here, we give a brief preview. 

\paragraph{The $I_{3322}$ operator and the two-dimensional geometry of the P\'al--V\'ertesi construction.}

The $I_{3322}$ inequality is one of the simplest bipartite Bell
inequalities beyond CHSH: each party has three possible binary measurements.
Let Alice and Bob's measurements be specified by the projectors
$A_1,A_2,A_3$ and $B_1,B_2,B_3$ respectively (think of these as the projectors onto the zero outcome for each possible measurement). Then, the $I_{3322}$ Bell operator is
\begin{align*}
I_{3322}={}&-A_2-B_1-2B_2
 +A_1B_1+A_1B_2+A_2B_1+A_2B_2 \\
&-A_1B_3+A_2B_3-A_3B_1+A_3B_2.
\end{align*}
where we are omitting the tensor product symbol between Alice's and Bob's operators for convenience.
The classical value of $I_{3322}$ can straightforwardly be seen to be zero.  

Our proof starts from the P\'al--V\'ertesi construction, which provides a family of finite-dimensional quantum strategies whose $I_{3322}$ values increase with the dimension. In the appropriate Schmidt basis, these strategies have the following block structure, depicted in Figure~\ref{fig:pv-blocks}.

The first two measurements of Alice are block diagonal with a $1\times 1$ block at the first basis vector and then $2\times 2$ blocks on the pairs $(2,3),(4,5),\dots$, while Bob's first two measurements have the same $2\times 2$ blocks but shifted by one basis vector, acting on $(1,2),(3,4),\dots$ and ending with a $1\times 1$ block. The third measurements $A_3,B_3$ use the same two interlaced block
decompositions, but with the roles switched: $A_3$ follows the block
pattern of $B_1,B_2$, while $B_3$ follows that of $A_1,A_2$.  On each
nontrivial $2\times2$ block, the corresponding projector for $A_3$ and $B_3$ is simply
$\ket{+}\bra{+}$.

\begin{figure}[h]
\centering
\begin{tikzpicture}[x=0.95cm,y=0.9cm]
  \foreach \x/\lab in {0/1,1/2,2/3,3/4,4/5,5/6,6/{\cdots},7/n}{
    \node[circle,draw,inner sep=1.2pt] (v\x) at (\x,0) {\scriptsize $\lab$};
  }

  \node[draw=blue!70!black,rounded corners,minimum width=0.45cm,minimum height=0.45cm] at (0,1.1) {};
  \node[draw=blue!70!black,rounded corners,minimum width=1.45cm,minimum height=0.45cm] at (1.5,1.1) {};
  \node[draw=blue!70!black,rounded corners,minimum width=1.45cm,minimum height=0.45cm] at (3.5,1.1) {};
  \node[draw=blue!70!black,rounded corners,minimum width=1.45cm,minimum height=0.45cm] at (5.5,1.1) {};
  \node[text=blue!70!black] at (3.3,1.7) {\scriptsize blocks for $A_1,A_2, B_3$};

  \node[draw=red!75!black,rounded corners,minimum width=1.45cm,minimum height=0.45cm] at (0.5,-1.1) {};
  \node[draw=red!75!black,rounded corners,minimum width=1.45cm,minimum height=0.45cm] at (2.5,-1.1) {};
  \node[draw=red!75!black,rounded corners,minimum width=1.45cm,minimum height=0.45cm] at (4.5,-1.1) {};
  \node[draw=red!75!black,rounded corners,minimum width=0.45cm,minimum height=0.45cm] at (7,-1.1) {};
  \node[text=red!75!black] at (3.7,-1.7) {\scriptsize blocks for $B_1,B_2, A_3$};

  \foreach \x in {0,1,2,3,4,5,7}{
    \draw[densely dotted,gray!60] (\x,0.25) -- (\x,0.88);
    \draw[densely dotted,gray!60] (\x,-0.25) -- (\x,-0.88);
  }
\end{tikzpicture}
\caption{Interlaced block structure in the P\'al--V\'ertesi strategy. The middle row represents the basis vectors.
}
\label{fig:pv-blocks}
\end{figure}

In the P\'al--V\'ertesi construction, the operators are specified entirely by a sequence of parameters $c_1,\dots,c_n\in[-1,1]$,
which we refer to as a \emph{path}. If $\ket{\psi}=\sum_{i=1}^n \lambda_i \ket{i}\ket{i}$ is the Schmidt decomposition of the joint state, then the Bell value can be written compactly as
$$
\langle \psi| I_{3322}|\psi\rangle
=
\langle \lambda, M(c)\lambda\rangle,
$$
where $\lambda=(\lambda_1,\dots,\lambda_n)^T$ and $M(c)$ is a symmetric tridiagonal matrix. The fact that $M(c)$ is tridiagonal follows directly from the
interlacing block structure of the measurements. Moreover, its $i$-th diagonal
entry depends only on the neighboring labels $c_{i-1},c_i$, while each off-diagonal entry depends only on one $c_i$. Thus, for a fixed path $c$, the optimal joint state $\ket{\psi}$ is
obtained by choosing $\lambda$ to be a normalized top eigenvector of
$M(c)$. The optimization
implicit in the P\'al--V\'ertesi construction therefore reduces to
maximizing the top eigenvalue of this tridiagonal matrix over paths $c$.  We define $q_*$ to be the supremum of these top
eigenvalues over all dimensions. P\'al and V\'ertesi observed numerically that the value achieved by their construction appears to increase as the path length increases, suggesting that the limiting value $q_*$ may not be attained at any finite length. Our proof starts by understanding this path optimization in more detail.

\paragraph{A Bellman-type ``certificate'', and the inherent two-dimensional geometry of $I_{3322}$.} 
The first milestone in our proof is to establish that $q_*$ is in fact an upper bound on the value of \emph{arbitrary} strategies (not just the particular ones from P\'al--V\'ertesi). The key object in our analysis is a \emph{certificate}: a certain concave scalar function $H$ satisfying two families of inequalities, which helps upper bound the value of arbitrary strategies. This certificate is rooted in ideas from Bellman's principle of optimality in
dynamic programming \cite{bellman1952theory}. This was later adapted in harmonic analysis, where Bellman functions have been used to
derive operator bounds from inequalities involving only a few
scalar variables
\cite{nazarov1997hunt, nazarov2001bellman}. In an optimization problem involving many
successive choices, the Bellman approach replaces the full history of those choices by an auxiliary function that records only the information
needed to make the next choice.  If that function satisfies the right one-step relation, then extending the
sequence by one choice yields an inequality involving only the function values before and after the extension. 

Our certificate follows from this principle. Since $q_*$ is defined by optimizing over paths of arbitrary length, we study what happens when one more label is appended to a path. Thanks to
the tridiagonal structure of the P\'al--V\'ertesi matrices, all the information from the preceding path that is relevant to this extension can be summarized by a certain scalar quantity, which follows a simple one-step update rule. For each possible path endpoint $v$, we then consider all paths ending at $v$ and keep the
smallest value of this scalar that can arise. This defines a function $H(v)$. From the one-step update rule, we then obtain an inequality relating the value of $H$ at any two points, $H(u)$ and $H(v)$, by considering a path ending at $u$, and appending the new label $v$.

What is surprising is that this function $H$, which is extracted from the
specific P\'al--V\'ertesi family, is useful to upper bound the value of \emph{arbitrary} quantum strategies. One reason for this is that there is actually an inherent
two-dimensional block structure in the $I_{3322}$ operator itself. In fact, the latter can be rewritten as: 
\begin{equation}
        I_{3322}
        = Z_AZ_B+\frac12Z_A-\frac12Z_B-I +\frac12X_AW_B+\frac12W_AX_B\,,
\label{eq:overview-I3322-rewrite}
\end{equation}
where $Z_A, X_A$ \emph{anticommute}, and are defined in terms of $A_1$ and $A_2$, and $W_A = 2A_3 - I$. Similarly on Bob's side. Crucially, the anticommutation induces a 2-dimensional block structure shared by the $Z_A$ and $X_A$ operators. What is fortunate is that, after decomposing a general strategy into such blocks, its value can be expressed in
terms of the same scalar quantities that appear in the
P\'al--V\'ertesi tridiagonal matrices. In this way, the certificate $H$ 
ends up being useful to bound the value of
general strategies.

The same certificate then also controls what happens when the upper bound $q_*$ is achieved: any strategy attaining $q_*$ must saturate every inequality used in the upper bound.  In finite
dimensions, we show that these equality conditions essentially force the strategy to decompose into finite components with the \emph{same}
nearest-neighbor structure underlying the P\'al--V\'ertesi construction. While this is not a ``rigidity'' statement in the formal self-testing sense, it does establish that the analysis of arbitrary strategies reduces to analyzing strategies from the P\'al--V\'ertesi family. We then prove that no finite-dimensional strategy from the latter family (i.e.\ corresponding to some finite path $c$) can attain $q_*$, ruling
out finite-dimensional attainment. 

On the other hand, showing that there exists an infinite-dimensional strategy attaining $q_*$ is not automatic, but it boils down to showing that one can extract a well-defined limiting strategy from sequences of finite-dimensional P\'al--V\'ertesi strategies.

\subsection{Concurrent work}
As we were completing this paper, we learned first of a Lean-verified proof by Douglas~\cite{douglas2026i3322}, who kindly reached out to us about his proof made public on Zenodo on August~5. We then coordinated arXiv postings of our concurrent works~\cite{douglas26arxiv}. A few days before our two papers appeared on arXiv, Pauwels~\cite{pauwels26} posted another independent and concurrent proof (of finite-dimensional non-attainment, though not infinite-dimensional attainment). We subsequently learned of another proof by Mghirbi~\cite{mghirbi2026i3322}, made public on Zenodo on August 10 (though not Lean-verified).

\subsection*{Acknowledgements}
The author is thankful for support from the Google Research Scholar program. The author is also thankful to Seth Douglas for reaching out about his concurrent result, and for kindly coordinating our arXiv posts (even though the author eventually missed the agreed upon arXiv deadline..).

\section{Technical overview}
\label{sec:technical-overview}

The main result of this paper is that the quantum value of \(I_{3322}\)
is a number \(q_*\) that can be attained on separable
infinite-dimensional Hilbert spaces, but not in any finite dimension.  More precisely, if
\(\omega_d(I_{3322})\) denotes the largest value achievable with local
dimension at most \(d\), then
\[
    \omega_d(I_{3322})<q_*
        \quad\text{for every finite }d \,.
\]
On the other hand, there is an infinite-dimensional strategy attaining $q_*$ on
\(\ell^2(\mathbb Z)\otimes\ell^2(\mathbb Z)\). Going forward, when we use the term infinite-dimensional strategy, we refer to a strategy on the separable infinite-dimensional Hilbert spaces \(\ell^2(\mathbb Z)\otimes\ell^2(\mathbb Z)\).


\subsection{The \(I_{3322}\) operator and its hidden two-dimensional geometry}

The inequality \(I_{3322}\) is one of the simplest bipartite Bell
inequalities beyond CHSH: each party has three possible binary measurements.
Let Alice and Bob's measurements be specified by the projectors
\(A_1,A_2,A_3\) and \(B_1,B_2,B_3\) respectively (think of these as the projectors onto the zero outcome for each possible measurement). Then, the $I_{3322}$ Bell operator is
\begin{align*}
I_{3322}={}&-A_2-B_1-2B_2
 +A_1B_1+A_1B_2+A_2B_1+A_2B_2 \\
&-A_1B_3+A_2B_3-A_3B_1+A_3B_2.
\end{align*}
For convenience, here we omit writing the tensor product symbol between Alice's and Bob's operators when no ambiguity can arise.
The classical value of $I_{3322}$ can straightforwardly be seen to be zero.  

The following rewriting of the $I_{3322}$ operator will reveal some crucial structure. Introduce
\begin{equation}
\label{eq:relations}
        Z_A=A_1+A_2-I,
        \qquad
        X_A=A_2-A_1,
\end{equation}
and define \(Z_B,X_B\) analogously.  Also write
\[
        W_A=2A_3-I,
        \qquad
        W_B=2B_3-I.
\]
The operators \(W_A,W_B\) are reflections.  More importantly, using the
fact that $A_1$ and $A_2$ are projectors, one can straightforwardly check that
\begin{equation}
        Z_AX_A+X_AZ_A=0,
        \qquad
        Z_A^2+X_A^2=I,
\label{eq:overview-anticommutation}
\end{equation}
and the same relations hold for \(Z_B,X_B\).  In terms of the new operators, the
Bell operator then becomes
\begin{equation}
\begin{aligned}
        I_{3322}
        ={}&Z_AZ_B+\frac12Z_A-\frac12Z_B-I \\
        &+\frac12X_AW_B+\frac12W_AX_B.
\end{aligned}
\label{eq:overview-I3322-rewrite}
\end{equation}

The anticommutation in \eqref{eq:overview-anticommutation} is the source
of the following important \emph{two-dimensional block} geometry. Ignoring exceptional one-dimensional
blocks for now, a pair \(Z,X\) satisfying these relations decomposes
into two-dimensional blocks of the form
\begin{equation}
\label{eq:blocks}
        Z=
        \begin{pmatrix}t&0\\0&-t\end{pmatrix},
        \qquad
        X=
        \begin{pmatrix}
        0&\sqrt{1-t^2}\\
        \sqrt{1-t^2}&0
        \end{pmatrix},
        \qquad 0<|t|<1.
\end{equation}
The latter is not difficult to see, and follows, for example, from Jordan's lemma along with a direct analysis of $2 \times 2$ matrices.

In a basis in which \(Z\) is diagonal, we call the eigenvalue of \(Z\)
associated with a basis vector its \emph{\(Z\)-label}.  Thus, in each
two-dimensional block, the two basis vectors have labels \(t\) and
\(-t\), and \(X\) maps one to the other with coefficient
\(\sqrt{1-t^2}\).  The exceptional one-dimensional blocks have label
\(\pm1\) and \(X=0\), or label \(0\) and \(X=\pm I\).

We introduce two scalar functions that will play an important role in
the analysis:
\begin{equation}
        a(t):=\frac12\sqrt{1-t^2},
        \qquad
        d(u,v):=uv+\frac{u-v}{2}-1.
\label{eq:overview-a-d}
\end{equation}
For a joint basis vector \(e_\alpha\otimes f_\beta\), let \(u\) be the
\(Z_A\)-label of Alice's basis vector \(e_\alpha\), and let \(v\) be the
\(Z_B\)-label of Bob's basis vector \(f_\beta\).  Then \(d(u,v)\) is
exactly the contribution of the first line of
\eqref{eq:overview-I3322-rewrite} on this joint basis vector, and \(a(t)\) is one half of the coefficient with which \(X_A\) or \(X_B\)
``connects'' two basis vectors whose \(Z\)-labels are \(t\) and \(-t\).

\subsection{Jacobi paths and their canonical strategies}
\label{sec:canonical-strategies}
We are ready to introduce a key object in the proof: a ``Jacobi path'' and its associated Jacobi matrix\footnote{The terminology of Jacobi matrices is not new: these are real symmetric tridiagonal matrices,
usually with positive entries on the adjacent off-diagonals, though we will
also allow these entries to vanish. Here, we introduce the notion of a ``Jacobi path'', a sequence $c = (c_0, \ldots, c_m)$ that defines a corresponding Jacobi matrix in a certain specific way.}. 
A finite Jacobi path of length \(m\) is a sequence
\[
        c=(c_0,c_1,\ldots,c_m),
        \qquad
        c_0,c_m\in\{\pm1\},
        \qquad
        c_1,\ldots,c_{m-1}\in[-1,1].
\]
Its Jacobi matrix is the \(m\times m\) real symmetric tridiagonal
matrix
\begin{equation}
        (J_m(c))_{ii}=d(c_{i-1},c_i),
        \qquad
        (J_m(c))_{i,i+1}=(J_m(c))_{i+1,i}=a(c_i).
\label{eq:overview-finite-Jacobi}
\end{equation}
The coordinates \(1,\ldots,m\) can be thought of as forming a nearest-neighbor chain, where
coordinate $i$ is coupled to coordinate $i+1$ via an edge of weight $a(c_i)$. So, the endpoint condition \(c_0,c_m\in\{\pm1\}\) can be thought of as ensuring that there are no ``couplings'' beyond coordinates \(1\) and \(m\), because \(a(c_0)=a(c_m)=0\).

\paragraph{The canonical strategy.} Every Jacobi path \(c=(c_0,\ldots,c_m)\) defines a canonical
\(I_{3322}\) strategy on
\(\mathcal H_A=\mathcal H_B=\mathbb C^m\). For the endpoint choice \(c_0=1\) and \(c_m=-1\), this construction is the same as the one considered by P\'al and V\'ertesi, rewritten in our
\(Z,X,W\) notation (note that P\'al and V\'ertesi do not use the terminology of Jacobi paths). Here, we make the minor formal extension of allowing all
four choices
\[
        c_0,c_m\in\{\pm1\}.
\]

Let
\(e_1,\ldots,e_m\) be the standard basis and set
\[
        s_i:= 2a(c_i) = \sqrt{1-c_i^2}.
\]
\begin{enumerate}
\item 
Define the diagonal operators \(Z_A,Z_B\) by
\[
        Z_Ae_i=
        \begin{cases}
        c_{i-1}e_i,&i\text{ odd},\\
        -c_i e_i,&i\text{ even},
        \end{cases}
        \qquad
        Z_Be_i=
        \begin{cases}
        c_i e_i,&i\text{ odd},\\
        -c_{i-1}e_i,&i\text{ even}.
        \end{cases}
\]
\item For each even \(i<m\), define
\[
        X_Ae_i=s_i e_{i+1},
        \qquad
        X_Ae_{i+1}=s_i e_i,
\]
and let \(X_A\) vanish on any basis vector not belonging to one of these
pairs.  Define \(X_B\) in the same way on the odd pairs:
\[
        X_Be_i=s_i e_{i+1},
        \qquad
        X_Be_{i+1}=s_i e_i
        \qquad(i<m\text{ odd}),
\]
again setting it equal to zero on any unpaired basis vector. One can see that these operators in fact satisfy the relations in \eqref{eq:overview-anticommutation}, and one can back out projectors $A_1, A_2$ and $B_1,B_2$ from them, by inverting~\eqref{eq:relations}.

\item Finally, let \(W_A\) swap \(e_i\) and \(e_{i+1}\) for every odd \(i<m\),
and let \(W_B\) perform the same swap for every even \(i<m\); both
operators act as the identity on any remaining basis vectors. Then, define $A_3 = \frac{I+W_A}{2}$ and $B_3 = \frac{I+W_B}{2}$.
\end{enumerate}




Let
\(\phi\) be a normalized
eigenvector corresponding to the largest eigenvalue of \(J_m(c)\). We can take $\phi$ to have all non-negative entries (since the off-diagonal entries
\(a(c_i)\) are nonnegative). We use its
coordinates $\phi_i$ as Schmidt coefficients to define Alice and Bob's joint state
\[
        \ket{\psi_c}=\sum_{i=1}^m\phi_i e_i\otimes e_i.
\]
For simplicity, here and in the rest of the paper, we usually denote vectors without using kets.

The diagonal terms in \eqref{eq:overview-I3322-rewrite} contribute
\(\sum_i d(c_{i-1},c_i)\phi_i^2\), while the non-diagonal terms
contribute
\(2\sum_{i=1}^{m-1}a(c_i)\phi_i\phi_{i+1}\).  Therefore,
\begin{equation}
        \langle\psi_c | I_{3322} | \psi_c\rangle
        =\langle\phi,J_m(c)\phi\rangle
        =\lambda_{\max}\left(J_m(c)\right).
\label{eq:overview-path-value}
\end{equation}

We remark that the operators involved in the canonical strategy above are all block-diagonal with $2\times 2$ blocks, just as in the optimal strategy for the correlation in \cite{coladangelo2020inherently} (on five questions and three answers). One interesting difference is that the $2\times 2$ blocks are ``interlacing'' \emph{between} Alice and Bob. Interlacing is present in  \cite{coladangelo2020inherently} as well, but only among operators on the same side, while corresponding operators on opposite sides have a matching block-decomposition.

\paragraph{The critical value $q_*$.} We can now define the supremum of the value in \eqref{eq:overview-path-value} over all finite Jacobi paths. Let 
\begin{equation}
\begin{aligned}
        Q_m
        &:={}
        \max_{\substack{c_0,c_m\in\{\pm1\}\\
                        c_1,\ldots,c_{m-1}\in[-1,1]}}
        \lambda_{\max}\left(J_m(c)\right),\quad \text{and}\\
        q_*&:=\sup_{m\geq1}Q_m.
\end{aligned}
\label{eq:overview-qstar}
\end{equation}
P\'al and V\'ertesi numerically observed a sequence of
finite-dimensional Jacobi path strategies whose values appeared to converge to \(0.250875384514\). They conjectured that: $q_*$ as defined above is in fact this limiting value, and no finite-dimensional strategy
attains this limiting value.

At this point, note that \(q_*\) is defined only as the supremum of the value that can be achieved using canonical strategies based on finite paths. Our proof will establish three facts:
\begin{itemize}
\item $q_*$ is an upper bound on the value of any arbitrary finite or infinite-dimensional strategy. 
\item No arbitrary finite-dimensional strategy attains $q_*$
\item There is an infinite-dimensional strategy attaining $q_*$. 
\end{itemize}


\paragraph{No finite path attains $q_*$}
A useful intermediate step towards our final goal is to first establish the restricted result that no strategy based on a finite Jacobi path can attain $q_*$. We do this by showing that, as long as $Q_m>0$, we have
\begin{equation}
\label{eq:8}
       Q_{m+1}>Q_m,
\end{equation}
which implies \(Q_m<q_*\) for every \(m\) (otherwise there would be an $m$ for which $Q_m$ exceeds the supremum). Note that one need not worry about the possibility that $q_* = 0$ because an explicit finite Jacobi path strategy achieving value strictly greater than $\frac14$ is known (see Lemma~\ref{lem:inf-quarter-witness}).

We will not describe how to prove \eqref{eq:8} in detail in this overview, but only
describe the main idea.  Suppose \(Q_m>0\), and choose a length-\(m\)
path \(c\) attaining \(Q_m\), whose existence follows from compactness.
If \(J_m(c)\) is reducible, choose an irreducible block with largest
eigenvalue \(Q_m\). The irreducibility ensures that the off-diagonal entries are strictly positive, so
its top eigenvector can be chosen strictly positive.  Write
\[
        b=(c_r,c_{r+1},\ldots,c_s)
\]
for the portion of the path corresponding to this block, so that
the block is simply the Jacobi matrix \(J_{s-r}(b)\).  Because the block begins and ends where
the neighboring ``couplings'' vanish, it must be that
\[
        c_r,c_s\in\{\pm1\}.
\]
We now lengthen the original path by inserting the label \(c_r(1-\delta)\) immediately to the right of \(c_r\): let
\[
        \widetilde c^{\,\delta}
        =
        (c_0,\ldots,c_r,c_r(1-\delta),
         c_{r+1},\ldots,c_m).
\]
When \(\delta=0\), the old block remains a direct summand of the enlarged matrix, and nothing changes. For small \(\delta>0\) though, the new coordinate acquires some positive coupling to the first coordinate of the old block. It turns out that, for small enough $\delta>0$, simply optimizing this new coordinate and leaving the previous top eigenvector unchanged on the old coordinates achieves a higher Rayleigh quotient, and thus the new Jacobi matrix has a strictly larger top eigenvalue (in this step it is important that the previous top eigenvector has a strictly positive first coordinate).

This gives $Q_{m+1} > Q_m$. Thus, finite Jacobi path strategies approximate \(q_*\), by definition of supremum, but do not attain it.


\subsection{The scalar certificate: a $q_*$ upper bound on arbitrary strategies}
We now turn to the first main challenge: upper bounding the value of an arbitrary strategy, finite or infinite-dimensional, by $q_*$. 

The central object in our proof is a concave function
\[
        H:[-1,1]\longrightarrow[0,\infty)
\]
satisfying
\begin{align}
        H(v)+d(u,v)+H(-u)&\leq q_*,
        &&u,v\in[-1,1],
\label{eq:overview-certificate-scalar}\\
        H(t)H(-t)&\geq a(t)^2,
        &&t\in[-1,1].
\label{eq:overview-certificate-product}
\end{align}
We refer to the latter as a \emph{scalar certificate}. Concavity will become important to
analyze the structure of strategies attaining $q_*$, but, for now, the two inequalities above already imply the
desired upper bound for arbitrary strategies. Before constructing \(H\), let us assume we have such an $H$, and prove this claim.

In this overview, we restrict our attention to finite-dimensional strategies. The extension to infinite-dimensional ones follows by the fact that the set of finite and infinite-dimensional quantum correlations have the same closure. So, let us consider an arbitrary finite-dimensional quantum strategy. By
purification and Naimark dilation, it suffices to consider a pure state and projective binary measurements. Apply the block decomposition from
\eqref{eq:blocks} on both sides.  Let \(z_\alpha\) and
\(w_\beta\) denote the corresponding \(Z_A\)- and \(Z_B\)-labels, and
expand
\[
        \ket{\psi}
        =\sum_{\alpha,\beta}
        \psi_{\alpha\beta}e_\alpha\otimes f_\beta.
\]
Let
\[
        p_{\alpha\beta}:=|\psi_{\alpha\beta}|^2,
        \qquad
        r_\alpha:=\sum_\beta p_{\alpha\beta},
        \qquad
        \ell_\beta:=\sum_\alpha p_{\alpha\beta}.
\]
The diagonal part of \eqref{eq:overview-I3322-rewrite} is exactly
\begin{align}
        \langle \psi| Z_AZ_B+\frac12Z_A-\frac12Z_B-I |\psi\rangle &= \sum_{\alpha,\beta} p_{\alpha\beta} (z_{\alpha}w_{\beta} + \frac{z_{\alpha} - w_{\beta}}{2} - 1) \nonumber \\
        &= \sum_{\alpha,\beta}
        p_{\alpha\beta}\,d(z_\alpha,w_\beta).
\label{eq:overview-arbitrary-diagonal}
\end{align}

It remains to bound the two terms containing \(X_A\) or \(X_B\).
For each Alice index \(\alpha\), define the unnormalized Bob-side
vector
\[
        R_\alpha:=\sum_\beta\psi_{\alpha\beta}f_\beta.
\]
Then
\[
        \ket{\psi}=\sum_\alpha e_\alpha\otimes R_\alpha,
        \qquad
        \|R_\alpha\|^2=r_\alpha.
\]
In this overview, for simplicity, we restrict our attention to two-dimensional blocks. Consider a two-dimensional Alice block
\(\{\alpha,\bar\alpha\}\) with labels $z_\alpha$ and $z_{\bar{\alpha}} = - z_\alpha$. On this block,
\[
        X_Ae_\alpha=2a(z)e_{\bar\alpha},
        \qquad
        X_Ae_{\bar\alpha}=2a(z)e_\alpha.
\]
Its contribution to the term
\(\frac12\langle X_A\otimes W_B\rangle\) is therefore
\[
\begin{aligned}
&\frac12\,2a(z)
 \left(
 \langle R_{\bar\alpha},W_BR_\alpha\rangle
 +\langle R_\alpha,W_BR_{\bar\alpha}\rangle
 \right)\\
&\qquad=
2a(z)\operatorname{Re}
\langle R_{\bar\alpha},W_BR_\alpha\rangle\\
&\qquad\leq
2a(z)\|R_\alpha\|\,\|R_{\bar\alpha}\|
=
2a(z)\sqrt{r_\alpha r_{\bar\alpha}},
\end{aligned}
\]
where we used \(\|W_B\|=1\).

The product inequality for \(H\), along with AM-GM, gives
\begin{equation}
\begin{aligned}
        2a(z)\sqrt{r_\alpha r_{\bar\alpha}}
        &\leq
        2\sqrt{H(z)H(-z)r_\alpha r_{\bar\alpha}}\\
        &\leq
        H(-z)r_\alpha+H(z)r_{\bar\alpha}.
\end{aligned}
\label{eq:overview-block-young}
\end{equation}
Summing over Alice's blocks gives
\[
        \frac12\langle X_AW_B\rangle
        \leq\sum_\alpha H(-z_\alpha)r_\alpha.
\]
Analogously, on Bob's side, we get $\frac12\langle W_AX_B\rangle \leq\sum_\beta H(w_\beta)\ell_\beta$.

Combining these inequalities with
\eqref{eq:overview-arbitrary-diagonal} yields
\begin{align}
        \langle I_{3322}\rangle
        &\leq
        \sum_{\alpha,\beta}p_{\alpha\beta}
        \bigl[
        d(z_\alpha,w_\beta)
        +H(-z_\alpha)+H(w_\beta)
        \bigr]
        \notag\\
        &\leq q_*.
\label{eq:overview-universal-bound}
\end{align}
The upper bound above can also be extended, with some work, to separable infinite-dimensional strategies (we refer the reader to Appendix~\ref{sec:app}).

\subsection{Where the certificate comes from}

In the previous subsection, we showed what the certificate $H$ buys us. Here, we discuss how one arrives at such a certificate. The information supplied by the P\'al--V\'ertesi construction concerns
an entire Jacobi path at once: by definition of \(q_*\),
\[
        \lambda_{\max}(J_m(c))\leq q_*
\]
for every finite Jacobi path \(c\). Thus we know how to bound the
matrix obtained after all the labels in a path have been assembled.
The previous subsection, however, requires a single function \(H\)
for which
\eqref{eq:overview-certificate-scalar} and
\eqref{eq:overview-certificate-product} hold separately for every
choice of \(u,v\), and \(t\).  These inequalities can then be applied
to the individual two-dimensional blocks of an arbitrary strategy and
averaged.  We now explain how to construct \(H\) from these bounds on the Jacobi matrices. The key is the tridiagonal form of these matrices.

Fix \(q>q_*\), and consider a finite path
\[
        c=(c_0,c_1,\ldots,c_m),
        \qquad c_0=1,
        \qquad c_1,\ldots,c_m\in[-1,1].
\]
Note that here the final path label need not lie in \(\{\pm1\}\), so this path is, strictly speaking, not necessarily a Jacobi path (which requires endpoints to be in $\{\pm 1\}$). Nevertheless, we
use the same formula as in \eqref{eq:overview-finite-Jacobi} to define
the \(m\times m\) tridiagonal matrix \(J_m(c)\). Set
\[
        A_m(q,c):=qI_m-J_m(c).
\]
Now, observe that, for each \(i\leq m\), the matrix
\(J_i(c_0,\ldots,c_i)\) has largest eigenvalue at most \(q_*\).
One way to see this is as follows: appending the endpoint label \(1\) to the path, makes \(J_i(c_0,\ldots,c_i)\) the upper-left
submatrix of the Jacobi matrix of a genuine Jacobi path, whose top eigenvalues is bounded by $q_*$. Then, the claim follows by Cauchy's interlacing theorem. Therefore,
\begin{equation}
        A_i(q,c) :=qI_i-J_i(c_0,\ldots,c_i)
        \succeq (q-q_*)I_i\succ0.
\label{eq:overview-prefix-positive}
\end{equation}
As we will see in a moment, working first with \(q>q_*\) is useful precisely because the
gap \(q-q_*\) is strictly positive (and prevents the procedure below from
encountering a zero denominator). At the end, we will be able to convert bounds involving $q$, to bounds involving $q_*$ by taking an appropriate limit \(q\downarrow q_*\).

Because \(A_m(q,c)\) is positive definite and tridiagonal, it possesses an ``LDL factorization'', i.e.\ it can be
written as
\[
        A_m(q,c)=L_mD_mL_m^{\mathsf T},
        \qquad
        D_m=\operatorname{diag}(p_1,\ldots,p_m),
\]
where \(L_m\) is lower bidiagonal with ones on its diagonal and each
\(p_i\) is positive.  The numbers \(p_i\) are called the
\emph{pivots} of \(A_m(q,c)\).  Comparing the entries on the two sides
of this factorization shows that they can be computed successively:
\begin{align}
        p_1&=q-d(c_0,c_1),
\label{eq:overview-first-pivot}\\
        p_i&=q-d(c_{i-1},c_i)
        -\frac{a(c_{i-1})^2}{p_{i-1}},
        \qquad 2\leq i\leq m.
\label{eq:overview-pivot-recursion}
\end{align}
Crucially, note that, when each successive coordinate is added, the entire preceding part of
the matrix influences the new pivot only through the single number
\(p_{i-1}\).  The term \(a(c_{i-1})^2/p_{i-1}\) accounts for the
coupling between the new coordinate and the preceding one.  This
recursion is often called the scalar Riccati recursion and is standard
in the analysis of symmetric tridiagonal matrices
\cite{BarthMartinWilkinson1967}.

Moreover, positive definiteness in \eqref{eq:overview-prefix-positive} ensures
that every pivot is well defined and positive.  In fact, the pivots are
uniformly bounded away from zero.  To see this, consider the
factorization of the \(i\times i\) prefix $A_i(q,c)$. It is not difficult to see that the bottom-right entry of
its inverse is
\[
        \bigl(A_i(q,c)^{-1}\bigr)_{ii}=\frac1{p_i}.
\]
On the other hand, \eqref{eq:overview-prefix-positive} implies
\(A_i(q,c)^{-1}\preceq(q-q_*)^{-1}I_i\).  Hence
\begin{equation}
        p_i\geq q-q_*>0.
\label{eq:overview-pivot-lower-bound}
\end{equation}

Now, as mentioned, the important feature of
\eqref{eq:overview-pivot-recursion} is that it ``compresses'' the dependence on the past.  If
a label path currently ends at \(u=c_i\) with final pivot
\(r=p_i\), then appending a new label \(v\) produces the new final
pivot
\[
        q-d(u,v)-\frac{a(u)^2}{r}.
\]
Thus, all of the earlier labels in the path affect the next step only through the current label \(u\) and the single number \(r\).

We now optimize over the part of the sequence that has already been
``traversed'', with the aim of replacing the entire preceding sequence by a single function of its final label. For \(v\in[-1,1]\), define
\begin{equation}
        h_q(v)
        :=
        \inf_{\substack{m\geq1,\ c_0=1,\ c_m=v\\
                        c_1,\ldots,c_{m-1}\in[-1,1]}}
        p_m(c),
\label{eq:overview-hq-definition}
\end{equation}
where \(p_m(c)\) denotes the final pivot produced by
\eqref{eq:overview-first-pivot}--\eqref{eq:overview-pivot-recursion}.
In other words, \(h_q(v)\) is the smallest final pivot that can be
approached by a finite path beginning at \(1\) and ending at
\(v\).  Before deriving the recursion satisfied by \(h_q\), we record that it
is finite and uniformly bounded away from zero.  The lower bound follows
because every final pivot is at least \(q-q_*\).  For the upper bound,
we use the one-step sequence \((1,v)\), whose final pivot is
\(q-d(1,v)\).  Therefore
\begin{equation}
        q-q_*\leq h_q(v)\leq q-d(1,v)
        =q+\frac{1-v}{2}\leq q+1.
\label{eq:overview-hq-bounds}
\end{equation}

Now, fix \(u,v\in[-1,1]\), choose sequences ending at \(u\) whose final
pivots tend to \(h_q(u)\), and append \(v\) to each of them.  The
recursion shows that the new final pivots then tend to
\[
        q-d(u,v)-\frac{a(u)^2}{h_q(u)}.
\]
These extended sequences are among all the sequences ending at \(v\),
over which \(h_q(v)\) takes its infimum.  Therefore, $h_q(v)$ is at most the above expression, which, upon rearranging, gives
\begin{equation}
        h_q(v)+d(u,v)+\frac{a(u)^2}{h_q(u)}\leq q.
\label{eq:overview-Bellman}
\end{equation}
The latter inequality is already close to the
certificate inequality \eqref{eq:overview-certificate-scalar} that we seek. The first two terms already have the desired
form, and the only mismatch is that the last term is
\(a(u)^2/h_q(u)\), whereas the certificate requires $h_q(-u)$. We therefore wish to
modify \(h_q\) into a function \(g_q\) satisfying
\[
        g_q(-u)=\frac{a(u)^2}{g_q(u)}.
\]
Equivalently,
\[
        g_q(u)g_q(-u)=a(u)^2,
\]
which is exactly the product inequality \eqref{eq:overview-certificate-product} required of the certificate,
in fact with equality. The remaining task is therefore to enforce this equality while retaining the upper bound we just proved. It turns out that this modification can be done, by leveraging simple symmetries of the functions $a$ and $d$. 

Eventually, we need our certificate $H$ to be concave. This property of the certificate is needed later to make the equality conditions sufficiently rigid to rule out finite-dimensional
attainment. So, we replace $g_q$
by its least concave majorant \(H_q\), which preserves both certificate inequalities (one of them for free, and the other with some work). Finally, by choosing \(q_k\downarrow q_*\) and
passing to a convergent subsequence of \(H_{q_k}\), we obtain the
required certificate \(H\) at \(q_*\).  The details of these
steps are given in Section~\ref{sec:certificate-proof}.

\medskip
Zooming out a bit, what we have done to arrive at the certificate is rooted in Bellman's principle of optimality in dynamic programming \cite{bellman1952theory}. A long sequence of choices is
summarized by a value function of the current endpoint, and extending
the sequence by one step gives a relation for that value function.  In
our setting, the endpoint label is the argument of the value function,
the final pivot is the quantity being minimized, and
\eqref{eq:overview-pivot-recursion} supplies the one-step update.

The overall procedure---compressing the Jacobi-path optimization into
the one-variable function \(h_q\), and then using \(h_q\) to construct
the certificate \(H\)---is directly connected to the Bellman-function
method used in harmonic analysis to prove operator inequalities~\cite{nazarov1997hunt, nazarov2001bellman}. In that method, one
encodes a complicated optimization in an auxiliary scalar function and
derives inequalities that control the individual pieces of an operator
decomposition.  Combining those inequalities then yields a bound for
the complete operator. This is exactly the role played by our certificate \(H\).  Its product
inequality \eqref{eq:overview-certificate-product} controls each two-dimensional \(X\)-coupling, while its other inequality \eqref{eq:overview-certificate-scalar} combines that bound with the corresponding diagonal contribution \(d(u,v)\). Averaging over all blocks then bounds
the value of the entire Bell expression. Thus, \(h_q\) can be thought of as a value function in the literal
dynamic-programming sense, while, in the terminology of harmonic
analysis, \(H\) would be called a Bellman function. 

\subsection{The structure of strategies attaining $q_*$: a reduction to paths (and cycles)}

We will now focus our attention on a hypothetical strategy that attains $q_*$, and show that it essentially reduces to a direct sum of Jacobi path (or cycle) components, which, as we already argued, cannot attain $q_*$.

Recall that the bound
\eqref{eq:overview-universal-bound} on the value of arbitrary strategies was obtained through a sequence of
inequalities.  If a strategy attains \(q_*\),
then every inequality must be saturated. More precisely, call a pair \((u,v)\) of Alice and Bob labels \emph{active} if it occurs
with positive total ``mass''
\[
        m_{u,v}
        :=
        \sum_{z_\alpha=u,\,w_\beta=v}p_{\alpha\beta}>0\,.
\]
Every active pair satisfies the scalar equality
\begin{equation}
        H(v)+d(u,v)+H(-u)=q_*\,,
\label{eq:overview-active-pair}
\end{equation}
which comes from saturating \eqref{eq:overview-universal-bound}. Moreover, consider a two-dimensional Alice block \(\{\alpha,\bar\alpha\}\), whose
two \(Z_A\)-labels are \(z_\alpha=u\) and
\(z_{\bar\alpha}=-u\), where
\(u\in(-1,1)\setminus\{0\}\), and whose marginal masses are
\(r_\alpha\) and \(r_{\bar\alpha}\). Equality in \eqref{eq:overview-block-young} then implies (since $a(u)>0$) that if either mass is positive, both are positive and
\begin{equation}
        H(u)H(-u)=a(u)^2.
\label{eq:overview-active-product}
\end{equation}
The analogous statement holds for every two-dimensional Bob block.

At this point, there is a crucial property of $H$ that leads to the degree-2 graph structure that we introduce shortly. This property is what ultimately causes any hypothetical strategy attaining $q_*$ to be decomposable into Jacobi path and cycle components.

Indeed, we show (in Lemma~\ref{lem:single-valued}) that, among pairs satisfying \eqref{eq:overview-active-pair}, any Alice label \(u \in (-1,1)\) at
which the product equality \eqref{eq:overview-active-product} holds can occur with only one \emph{unique} Bob label
\(v\).  Symmetrically, the same holds for Bob. This means that from
an active pair \((u,v)\), there can be at most one active pair of the
form \((-u,v')\) and at most one of the form \((u',-v)\). The graph
introduced next records precisely these two possible types of
continuation. This property relies crucially on the concavity of the certificate $H$.

We encode the active pairs in a finite weighted graph \(\Gamma\).  Its
vertices are the pairs
\[
        \nu=(u,v)
\]
with \(m_\nu:=m_{u,v}>0\).  Two active vertices \((u,v)\) and
\((-u,v')\) are joined by an Alice edge whenever \(0<|u|<1\), and this
edge has weight \(a(u)\).  Similarly, two active vertices \((u,v)\)
and \((u',-v)\) are joined by a Bob edge whenever \(0<|v|<1\), with
weight \(a(v)\).

The uniqueness condition above ensures that every vertex has at most one edge of each type. So, its degree
is at most two, and whenever two edges meet at a vertex, one
is an Alice edge and the other is a Bob edge.  Since the strategy is
finite-dimensional, \(\Gamma\) is finite, so every connected component
is an alternating \emph{path}, possibly consisting of a single vertex, or an
even alternating \emph{cycle}. 

To connect this graph with the value of the strategy, define a real
symmetric matrix \(M_\Gamma\), with rows and columns indexed by the
vertices of \(\Gamma\).  For simplicity, we ignore
one-dimensional blocks, and hence zero labels, in this overview. The gist of the proof does not change significantly when accounting for them. The
diagonal entries are
\begin{equation}
        (M_\Gamma)_{\nu\nu}
        =
        d(u_\nu,v_\nu),
\label{eq:overview-component-diagonal}
\end{equation}
For distinct vertices \(\nu\) and \(\mu\), the entry
\((M_\Gamma)_{\nu\mu}\) is the weight of the edge joining them, with the
two weights added in the case where \(\nu\) and \(\mu\) are
joined by both an Alice edge and a Bob edge. Thus, $M_{\Gamma}$ has the form of an \emph{adjacency} matrix.

Now, set
\[
        \xi_\nu:=\sqrt{m_\nu}.
\]
Since the active vertex masses partition the state weight,
\[
        \|\xi\|^2=\sum_\nu m_\nu=1.
\]
By \eqref{eq:overview-arbitrary-diagonal}, grouping the diagonal terms
according to their active pairs gives
\[
\sum_{\alpha,\beta}
p_{\alpha\beta}d(z_\alpha,w_\beta)
=
\sum_\nu d(u_\nu,v_\nu)m_\nu
=
\sum_\nu d(u_\nu,v_\nu)\xi_\nu^2.
\]
For the off-diagonal terms, the block calculation leading to
\eqref{eq:overview-block-young} shows that a two-dimensional Alice block
\(\{\alpha,\bar\alpha\}\) contributes at most
\[
        2a(z_\alpha)\sqrt{r_\alpha r_{\bar\alpha}},
\]
where $r_\alpha=\sum_\beta p_{\alpha\beta}$, and $
        r_{\bar\alpha}=\sum_\beta p_{\bar\alpha\beta}$.  The analogous estimate holds for Bob's two-dimensional blocks.  After grouping all blocks
corresponding to the same edge \(e=\{\nu,\mu\}\) and one application
of Cauchy--Schwarz, one can see that their total contribution is at most
\[
        2a_e\sqrt{m_\nu m_\mu}
        =
        2a_e\xi_\nu\xi_\mu,
\]
where \(a_e\) is the weight of \(e\).  Thus, we have
\[
\begin{aligned}
        q_*
        &\leq
        \sum_\nu d(u_\nu,v_\nu)\xi_\nu^2
        +
        2\sum_{e=\{\nu,\mu\}}a_e\xi_\nu\xi_\mu\\
        &=\langle\xi,M_\Gamma\xi\rangle.
\end{aligned}
\]
where the first inequality is because we are assuming a hypothetical strategy attaining $q_*$. It is not difficult to see that same calculation carried out in
\eqref{eq:overview-block-young}--\eqref{eq:overview-universal-bound},
now applied to an arbitrary real vector \(x\), gives
\[
        \langle x,M_\Gamma x\rangle\leq q_*\|x\|^2.
\]
Hence \(M_\Gamma\leq q_*I\), and, in fact, $\xi$ is the top eigenvector of $M_{\Gamma}$ with eigenvalue $q_*$.

Now, since $M_{\Gamma}$ has the form of an adjacency matrix, it decomposes as a direct sum over its connected components. Then, a connected component $\Gamma_0$ has the restriction $\xi|_{\Gamma_0}$ as a top eigenvector with eigenvalue $q_*$. Recall that each connected component of $\Gamma$ is either a path or a cycle, with alternating Alice and Bob edges. In the rest of this subsection, we will focus on path components for simplicity. Cycle components are handles with just slightly more work. We will argue that for a path component $\Gamma_0$, the matrix $M_{\Gamma_0}$ is in fact the Jacobi matrix of a certain Jacobi path. As argued in Section~\ref{sec:canonical-strategies}, any such matrix has top eigenvalue strictly less than $q_*$, which leads to the desired contradiction.

Order the vertices of $\Gamma_0$ along the path as
\[
        \nu_1,\ldots,\nu_m,
        \qquad
        \nu_i=(u_i,v_i),
\]
and let \(e_i\) be the edge joining \(\nu_i\) to \(\nu_{i+1}\).
Suppose that \(e_1\) is a Bob edge (the other case is similar). Set
\[
        c_0:=u_1,
        \qquad
        c_1:=v_1.
\]
As we move along the path, a Bob edge changes the sign of the Bob
coordinate, whereas an Alice edge changes the sign of the Alice
coordinate. Up to a sign, we define each subsequent $c_i$ to be the newly encountered coordinate. Precisely, we define the next labels recursively as
\[
        c_{i+1}:=
        \begin{cases}
        -u_{i+1},&i\ \text{odd},\\
        v_{i+1},&i\ \text{even},
        \end{cases}
        \qquad 1\leq i<m.
\]

By construction, for \(1\leq i<m\), the label \(c_i\) determines the weight of edge 
edge \(e_i\) in $\Gamma$: we have $\text{weight}(e_i)=a(c_i)$ and $0<|c_i|<1$. 
At either endpoint of the path, one edge type is missing, and this corresponds to $c_0, c_m \in \{\pm1\}$. 

Finally, using the identities $ a(-t)=a(t)$ and $d(-v,-u)=d(u,v)$, we get precisely that
\[
        (M_{\Gamma_0})_{ii}
        =
        d(c_{i-1},c_i),
        \qquad
        (M_{\Gamma_0})_{i,i+1}
        =
        a(c_i).
\]
Thus,
\[
        M_{\Gamma_0}=J_m(c)
\]
is exactly the Jacobi matrix of the path
\(c=(c_0,\ldots,c_m)\). As we argued in Section~\ref{sec:canonical-strategies}, its top eigenvalue is strictly less than
\(q_*\), which yields the desired contradiction.

\subsection{An infinite-dimensional strategy attaining $q_*$}
\label{sec:infinite-overview}
The definition of \(q_*\), along with the canonical strategies from Section~\ref{sec:canonical-strategies}, provides a sequence of finite-dimensional strategies whose values approach
\(q_*\). This approximating sequence alone does not automatically give an infinite-dimensional strategy attaining $q_*$ (and there are indeed correlations that are only attained in the limit, but not by any fixed infinite-dimensional strategy~\cite{slofstra2019set, dykema2019non}). In our case, an infinite-dimensional strategy attaining $q_*$ does exist, but we will need to prove the existence of a well-defined limit of the previously mentioned sequence. 

Recall that we defined
\[
        a(t)=\frac12\sqrt{1-t^2},
        \qquad
        d(u,v)=uv+\frac{u-v}{2}-1 \,,
\]
and that we denote by
\(J_n(c)\) the Jacobi matrix associated with a finite path \(c=(c_0,\ldots,c_n)\). We give an outline of the proof, which proceeds in four steps:
\begin{enumerate}
\item First, we choose finite Jacobi paths \(c^{(k)}\) of length $n_k$ whose largest eigenvalues \(q_k\) converge to \(q_*\).  For each \(k\), choose a non-negative eigenvector
\(\phi^{(k)}\) satisfying
\[
        J_{n_k}(c^{(k)})\phi^{(k)}=q_k\phi^{(k)}.
\]
Rescale \(\phi^{(k)}\) so that its largest coordinate is \(1\). Pad the \(k\)-th path with \(k\) repeated \(+1\) labels on the left
and \(k\) repeated \(-1\) labels on the right.  Since
\[
        a(\pm1)=0,
        \qquad
        d(1,1)=d(-1,-1)=0,
\]
the padded Jacobi matrix is simply
\[
        0_k\oplus J_{n_k}(c^{(k)})\oplus0_k.
\]
Thus the padding does not change the original eigenvalue, and we extend
\(\phi^{(k)}\) by zero on the added coordinates. Relabel the coordinates relative to the chosen maximal coordinate:
assign it index \(0\), and use negative and positive indices for the
coordinates to its left and right, respectively.  Thus
\[
        \phi_0^{(k)}=1.
\]
Going forward we denote by $c^{(k)}$ and $\phi^{(k)}$ the modified path and eigenvector. The padding ensures that there are at least \(k\) coordinates on each
side of index \(0\).  Consequently, every fixed index
\(i\in\mathbb Z\) occurs in the reindexed path for all sufficiently
large \(k\). Then, since the path labels lie in \([-1,1]\), and the rescaled eigenvector coordinates lie in \([0,1]\) we can find convergent subsequences
\[
        c_i^{(k)}\rightarrow c_i,
        \qquad
        \phi_i^{(k)}\rightarrow\phi_i
        \qquad(i\in\mathbb Z).
\]
The limiting labels \(c=(c_i)_{i\in\mathbb Z}\) naturally define the infinite-dimensional Jacobi operator \(J_\infty(c)\) by
\begin{equation}
        \bigl(J_\infty(c)x\bigr)_i
        :=
        a(c_{i-1})x_{i-1}
        +d(c_{i-1},c_i)x_i
        +a(c_i)x_{i+1},
        \qquad i\in\mathbb Z.
\label{eq:overview-infinite-Jacobi}
\end{equation}
Now, for any fixed \(i\), the index \(i\) is an interior point of the finite
path for all sufficiently large \(k\).  We may therefore pass to the
limit in the eigenvalue equation at that coordinate.  Since
\(q_k\to q_*\), this gives
\[
        J_\infty(c)\phi=q_*\phi 
\]
coordinatewise.
The normalization also passes to the limit:
\[
        0\leq\phi_i\leq1,
        \qquad
        \phi_0=1.
\]
Thus, \(\phi\) is bounded and nonzero.  However, the finite eigenvectors
were normalized by their largest coordinate, rather than by their
\(\ell^2\)-norm.  The pointwise limit may therefore fail to satisfy
\[
        \sum_{i\in\mathbb Z}\phi_i^2<\infty.
\]
Hence \(\phi\) cannot yet be used to back out a valid quantum state, as we do for a canonical strategy in Section~\ref{sec:canonical-strategies} (in particular $\phi$ cannot be used as the sequence of Schmidt coefficients of a state). 

One might naturally wonder: why not normalize the $\phi^{(k)}$ by their $\ell^2$ norm in the first place? The issue is that the vector's $\ell^2$ mass would spread over a larger and larger number of coordinates, so every coordinate would approach zero, and the limiting $\phi$ would actually be the zero vector.

\item Second, we restrict to the connected set $I$ of indices containing \(0\)
on which $\phi$ is positive, i.e.\ $\phi_i >0$ for $i \in I$. The certificate properties force all the
local inequalities used in the upper bound to be equalities on this
component. With some work, these equalities, together with the concavity of the
certificate, imply that the path labels are non-increasing from left to right along \(I\), i.e.\ $c_i\ge c_{i+1}$ for all $i$.

\item Third, we upgrade \(\phi\) to a genuine
\(\ell^2\)-eigenvector.  The property we need is the existence of
finitely supported vectors \(x^{(k)}\), supported on finite intervals
of \(I\), such that
\[
        \frac{\langle x^{(k)},J_\infty(c)x^{(k)}\rangle}
             {\|x^{(k)}\|_2^2}
        \longrightarrow q_*.
\]
These vectors are obtained by leaving \(\phi\)
unchanged on increasingly large intervals containing \(0\) and then
gradually scaling its coordinates to zero over longer boundary regions. Their support intervals give finite monotone paths from
\(+1\) to \(-1\) whose largest eigenvalues approach \(q_*\). As we will see in the main text, monotonicity is the key to ensuring that the normalized top eigenvectors of the Jacobi matrices corresponding to these monotone paths
retain at least one coordinate that is uniformly bounded away from zero. After centering at such a
coordinate and passing to a convergent subsequence, we obtain a nonzero
\(\tilde{\phi} \in \ell^2(\mathbb Z)\) satisfying
\[
        J_\infty(c) \tilde{\phi}=q_*\tilde{\phi}.
\]
We finally normalize \(\tilde{\phi}\). We refer to Section~\ref{sec:infinite-dimensional-attainment} for many more details.

\item Finally, we use the infinite-dimensional analogue of the
construction of a canonical strategy from a Jacobi path from Section~\ref{sec:canonical-strategies} to turn $\tilde{\phi}$ and the limiting path $c$ into a strategy achieving~$q_*$. 
\end{enumerate}

\section{Preliminaries}

\subsection{Binary projective strategies}

A finite-dimensional projective strategy consists of finite-dimensional Hilbert spaces $\cH_A,\cH_B$, a state $\ket{\psi}\in \cH_A\otimes\cH_B$, and projectors $A_i\in B(\cH_A)$, $B_j\in B(\cH_B)$.  The expectation notation means
\[
        \langle A_iB_j\rangle
        =\langle\psi,(A_i\otimes B_j)\psi\rangle,
\]
and similarly for one-party terms.  As usual, for binary finite-dimensional strategies it is enough to consider projective measurements and pure states: general effects can be dilated and mixed states purified without changing the achievable value.

\section{The $I_{3322}$ inequality}

Let $A_i, B_i$, for $i \in \{1,2,3\}$ be projectors on Alice and Bob's side respectively. The $I_{3322}$ Bell operator is defined to be
\begin{align}
\label{eq:I3322}
I_{3322} = &-A_2 - B_1 - 2B_2 + A_1B_1 + A_1B_2 + A_2B_1 + A_2B_2 \nonumber\\
&- A_1B_3 + A_2B_3 - A_3B_1 + A_3B_2 
\end{align}
Let $\ket{\psi} \in \mathcal{H}_A \otimes \mathcal{H}_B$.
It is known~\cite{froissart1981constructive} that for product states $\ket{\psi}$
$$ \langle \psi |I_{3322} |\psi \rangle \leq 0 \,.$$

\subsection{A useful equivalent form for $I_{3322}$}

Here, we rewrite the $I_{3322}$ Bell expression in an equivalent form. This rewriting is crucial in highlighting some of the underlying structure. Let $A_i, B_i$, for $i \in \{1,2,3\}$ be the projectors appering in \eqref{eq:I3322}. Define
\[
        Z_A=A_1+A_2-I,
        \qquad
        X_A=A_2-A_1,
\]
\[
        Z_B=B_1+B_2-I,
        \qquad
        X_B=B_2-B_1,
\]
and
\[
        W_A=2A_3-I,
        \qquad
        W_B=2B_3-I.
\]
Then $W_A,W_B$ are reflections, and a direct calculation (using $A_i^2 = A_i$ for all $i$) gives
\[
        Z_AX_A+X_AZ_A=0,
        \qquad
        Z_A^2+X_A^2=I,
\]
and the same identities on Bob's side.

\begin{lemma}[$I_{3322}$ rewrite]
\label{lem:rewrite}
The Bell operator $I_{3322}$ from \eqref{eq:I3322} can be rewritten as
\begin{equation}
\label{eq:rewrite}
        I_{3322}
    =Z_AZ_B+\frac12Z_A-\frac12Z_B
        +\frac12X_AW_B+\frac12W_AX_B-I .
\end{equation}
\end{lemma}

\begin{proof}
Substitute
\[
        A_1=\frac12(I+Z_A-X_A),\qquad
        A_2=\frac12(I+Z_A+X_A),
\]
\[
        B_1=\frac12(I+Z_B-X_B),\qquad
        B_2=\frac12(I+Z_B+X_B),
\]
\[
        A_3=\frac12(I+W_A),\qquad B_3=\frac12(I+W_B)
\]
into \eqref{eq:I3322} and collect terms.  The constant term is $-I$, the $Z$ terms are $\frac12Z_A-\frac12Z_B$, the diagonal two-party term is $Z_AZ_B$, and the remaining two-party terms are $\frac12X_AW_B+\frac12W_AX_B$.
\end{proof}

\section{From Jacobi paths to quantum strategies}
\label{sec:path-to-strategy}

In this section, we start by introducing the notion of a ``Jacobi path'' and its associated Jacobi matrix, which will be crucial in the rest of this proof. We will then show:
\begin{itemize}
\item Every Jacobi path corresponds to some quantum strategy for I3322 whose value is precisely the maximum eigenvalue of the Jacobi matrix. We show this in this section. This correspondence was already shown by Pal and Vertesi, but we include it here for self-containedness. Here we slightly extend the notion of path to allow for arbitrary endpoints in $\{\pm 1\}$ (not just starting at $+1$ and ending at $-1$), as we describe below.
\item Every finite-dimensional \emph{optimal} strategy for I3322 can be decomposed into a direct sum of Jacobi paths (and cycles). We will show this in Section~\ref{sec:structure-optimal}.
\end{itemize}

\begin{definition}[Jacobi path and its Jacobi matrix]
\label{def:signed-endpoint-jacobi-path}
Let $m\ge 1$. A Jacobi path of length $m$ is a sequence $c = (c_0, c_1, \ldots, c_m)$ such that
\[
        c_0, c_m \in \{\pm 1\},
        \qquad
        c_1,\ldots,c_{m-1}\in[-1,1], \,.
\]
We define its Jacobi matrix $J_m(c)$ to be the $m \times m$ real symmetric tridiagonal matrix such that
\[
        \bigl(J_m(c)\bigr)_{ii}=d(c_{i-1},c_i),
\]
\[
        \bigl(J_m(c)\bigr)_{i,i+1}
        =
        \bigl(J_m(c)\bigr)_{i+1,i}
        =a(c_i).
\]
where
\begin{equation}
\label{eq:a-d}
        a(t)=\frac12\sqrt{1-t^2},
        \qquad
        d(u,v)=uv+\frac{u-v}{2}-1.
\end{equation}
\end{definition}

\begin{definition}[Quantum strategy corresponding to a Jacobi path]
\label{def:strategy-for-path}
Let $c = (c_0, c_1, \ldots, c_m)$ be a Jacobi path. Its corresponding quantum strategy is the finite-dimensional strategy for $I_{3322}$ with $\mathcal H_A=\mathcal H_B=\mathbb C^m$ defined as follows.

Let $s_i=\sqrt{1-c_i^2}=2a(c_i)$.
Let $e_1,\ldots,e_m$ be the standard basis of $\mathbb C^m$.  We first define Hermitian operators $Z_A,X_A,Z_B,X_B,W_A,W_B$:
\begin{itemize}
\item Define $Z_A$ and $Z_B$ to be diagonal:  
\[
        Z_A e_i=
        \begin{cases}
        c_{i-1}e_i, & i\text{ odd},\\
        -c_i e_i, & i\text{ even},
        \end{cases}
\]
and
\[
        Z_B e_i=
        \begin{cases}
        c_i e_i, & i\text{ odd},\\
        -c_{i-1} e_i, & i\text{ even}.
        \end{cases}
\]
\item Define $X_A$ on the ``even edges'' by
\[
        X_A e_i=s_i e_{i+1},
        \qquad
        X_A e_{i+1}=s_i e_i
        \qquad (i\text{ even},\ 1\le i<m).
\]
Then, set $X_A e_1 =0$, and, if $m$ is even, $X_A e_m =0$ (i.e.\ $X_A e_j = 0$ when $j$ does not belong to any of the even edges). 

Define $X_B$ similarly on the ``odd edges'':
\[
        X_B e_i=s_i e_{i+1},
        \qquad
        X_B e_{i+1}=s_i e_i
        \qquad (i\text{ odd},\ 1\le i<m),
\]
and set $X_B e_m=0$ if $m$ is odd.

\item Finally, define $W_A$ to be the reflection which swaps the two endpoints of every odd edge and acts as the identity on the remaining standard basis vectors:
\[
        W_A e_i=e_{i+1},
        \qquad
        W_A e_{i+1}=e_i
        \qquad (i\text{ odd},\ 1\le i<m),
\]
and $W_A e_j=e_j$ otherwise.  

Define $W_B$ similarly on the even edges:
\[
        W_B e_i=e_{i+1},
        \qquad
        W_B e_{i+1}=e_i
        \qquad (i\text{ even},\ 1\le i<m),
\]
and $W_B e_j=e_j$ otherwise.
\end{itemize}
Then, the strategy uses \[
        A_1=\frac{I+Z_A-X_A}{2},
        \qquad
        A_2=\frac{I+Z_A+X_A}{2},
\]
\[
        B_1=\frac{I+Z_B-X_B}{2},
        \qquad
        B_2=\frac{I+Z_B+X_B}{2},
\]
\[
        A_3=\frac{I+W_A}{2},
        \qquad
        B_3=\frac{I+W_B}{2}.
\]
Since the Jacobi matrix $J_m(c)$ has nonnegative off-diagonal entries, its largest eigenvalue admits a normalized eigenvector $\lambda=(\lambda_1,\ldots,\lambda_m)$ with $\lambda_i\ge0$ (indeed, replacing any real vector $x$ by $(|x_1|,\ldots,|x_m|)$ does not decrease $\frac{\langle x, J_m(c) x \rangle }{\|x\|^2}$, and leaves the norm unchanged). Choose such a nonnegative top eigenvector $\lambda$ with norm $1$, and define the state
\[
        \ket{\psi}=\sum_{i=1}^m \lambda_i e_i\otimes e_i.
\]
in $\mathbb C^m\otimes\mathbb C^m$. This is the state shared by Alice and Bob.
\end{definition}

Note that all of the $A_i$ and $B_i$ are block-diagonal with $2\times 2$ blocks in the standard basis. 

\begin{proposition}[Value of a Jacobi path strategy]
\label{prop:path-to-strategy}
Let $c = (c_0, c_1, \ldots, c_m)$ be a Jacobi path, and let $J_m(c)$ be its Jacobi matrix. The quantum strategy corresponding to Jacobi path $c$ (as in Definition~\ref{def:strategy-for-path}) achieves value $\lambda_{\max}\left(J_m(c)\right)$.
\end{proposition}

\begin{proof}
On each $2\times 2$ block corresponding to an even edge $(e_i,e_{i+1})$, the operators $Z_A, X_A$ have the form
\[Z_A=\begin{pmatrix}-c_i&0\\0&c_i\end{pmatrix},
        \qquad
        X_A=\begin{pmatrix}0&s_i\\s_i&0\end{pmatrix}.
\]
Since these anti-commute, and $c_i^2 + c_i^2 =1$, we have
\[
        Z_AX_A+X_AZ_A=0,
        \qquad
        Z_A^2+X_A^2=I.
\]
On $1$-dimensional blocks, $X = 0$ and $Z=\pm I$, so the same relations still hold. Thus, these relations hold everywhere. Similarly,
\[
        Z_BX_B+X_BZ_B=0,
        \qquad
        Z_B^2+X_B^2=I.
\]
Also $W_A$ and $W_B$ are self-adjoint unitaries.  From the above identities, one can straightforwardly check that the corresponding $A_i$ and $B_i$ (as defined in Definition~\ref{def:strategy-for-path}) are indeed projectors.

For any operators $R,S$ on $\mathbb C^m$, letting $R_{ij}$ and $S_{ij}$ be their entries in the standard basis, we have
\[
        \langle\psi| (R\otimes S) |\psi\rangle
        =
        \sum_{i,j=1}^m \lambda_i\lambda_j R_{ij}S_{ij}.
\]
Now, recall from Lemma~\ref{lem:rewrite} that the $I_{3322}$ Bell operator can be rewritten as
\[
        I_{3322}
        =Z_AZ_B+\frac12Z_A-\frac12Z_B
        +\frac12X_AW_B+\frac12W_AX_B-I,
\]
Then, the value of the strategy on state $\ket{\psi}$ can be written as 
\[
        \sum_{i,j=1}^m \lambda_i\lambda_j M_{ij},
\]
where
\[
        M_{ij}
        =(Z_A)_{ij}(Z_B)_{ij}
        +\frac12 (Z_A)_{ii}\delta_{ij}
        -\frac12 (Z_B)_{ii}\delta_{ij}
        -\delta_{ij}
        +\frac12 (X_A)_{ij}(W_B)_{ij}
        +\frac12 (W_A)_{ij}(X_B)_{ij}.
\]
We will now show that, for the choice of $Z_A, Z_B, X_A, X_B, W_A, W_B$ in our chosen quantum strategy (i.e.\ the strategy corresponding to Jacobi path $c$), the operator $M$ is actually equal to $J_m(c)$.

For the diagonal entries, if $i$ is odd then
\[
        (Z_A)_{ii}=c_{i-1},
        \qquad
        (Z_B)_{ii}=c_i,
\]
and hence
\[
        M_{ii}=c_{i-1}c_i+\frac{c_{i-1}-c_i}{2}-1=d(c_{i-1},c_i).
\]
If $i$ is even then
\[
        (Z_A)_{ii}=-c_i,
        \qquad
        (Z_B)_{ii}=-c_{i-1},
\]
and the same computation gives
\[
        M_{ii}=c_{i-1}c_i+\frac{c_{i-1}-c_i}{2}-1=d(c_{i-1},c_i).
\]

For an off-diagonal pair $(i,i+1)$, if $i$ is even then
\[
        (X_A)_{i,i+1}=s_i,
        \qquad
        (W_B)_{i,i+1}=1,
\]
while $(X_B)_{i,i+1}=0$.  Therefore
\[
        M_{i,i+1}=\frac{s_i}{2}=a(c_i).
\]
If $i$ is odd then
\[
        (W_A)_{i,i+1}=1,
        \qquad
        (X_B)_{i,i+1}=s_i,
\]
while $(X_A)_{i,i+1}=0$, and again
\[
        M_{i,i+1}=\frac{s_i}{2}=a(c_i).
\]
All other off-diagonal entries vanish. Thus, $M=J_m(c)$.
So the value of the strategy is
\[
        \sum_{i,j=1}^m \lambda_i\lambda_j M_{ij}
        =
        \langle \lambda, J_m(c) \lambda\rangle
        =
        \lambda_{\max}\left(J_m(c)\right).
\]
This proves the proposition.
\end{proof}

\begin{definition}[The critical value $q_*$]
\label{def:q*}
For $n \in \mathbb{N}$, define
\[
        Q_m=\sup_{\substack{c_1,\ldots,c_{m-1}\in[-1,1]\\ c_0,c_m \in \{\pm 1\}}} \lambda_{\max}\left(J_m(c)\right),
        \qquad
        q_*:=\sup_{m\ge1}Q_m.
\]
\end{definition}
We will show in Section~\ref{sec:upper-bound} that $q_*$ upper bounds the value of any quantum strategy (even infinite-dimensional). Moreover, we will show in Section~\ref{sec:infinite-dimensional-attainment} that there is an infinite-dimensional strategy achieving $q_*$. So, as we will see, $q^*$ is in fact the optimal value for $I_{3322}$. We remark that, even though we define $q_*$ as the supremum over paths with arbitrary endpoints in $\{\pm 1\}$, this supremum is equal to the one over paths that start at $+1$ and end at $-1$, which are those considered by Pál and Vértesi. We show the latter in Lemma~\ref{lem:inf-open-reduction}. So, indeed $q_*$ is the same optimal value conjectured by Pál and Vértesi.

\section{Finite Jacobi paths cannot attain $q^*$}
\label{sec:paths-do-not-attain-q*}
In this section, we show that \emph{finite} Jacobi paths cannot attain $q_*$. The main idea is that, given any finite Jacobi path $c$ there is a simple way to increase its length by 1 to obtain a new path $c'$ such that $\lambda_{\max}\left(J_n(c')\right)>\lambda_{\max}\left(J_n(c)\right)$, i.e.\ the new Jacobi matrix has a \emph{strictly} larger top eigenvalue (and hence the corresponding $I_{3322}$ strategy attains a \emph{strictly} larger value). This immediately yields that no finite Jacobi path can exactly attain $q_*$, since $q_*$ is defined as the supremum attained by any finite path (Definition~\ref{def:q*}).

Note that this structural result is only the first step towards showing that no finite dimensional strategy can attain $q^*$: it only handles quantum strategies that are derived from a Jacobi path. The crux later in Section~\ref{sec:structure-optimal} will be to show that any \emph{arbitrary} optimal quantum strategy achieving $q^*$ must decompose into a direct sum of Jacobi path strategies.

We start with the following lemma.

\begin{lemma}[A coarse upper bound]
\label{lem:rowsum}
For all $n \in \mathbb{N}$, for all Jacobi paths $c = (c_0, \ldots, c_n)$,
\[
        \lambda_{\max}\left(J_n(c) \right)<\frac13.
\]
\end{lemma}

\begin{proof}
The off-diagonal entries of $J_n(c)$ are nonnegative, so, by Gershgorin's circle theorem, the largest eigenvalue is bounded by the largest row sum:
\[
\lambda_{\max}\left(J_n(c)\right)
\le
\max_i\{d(c_{i-1},c_i)+a(c_{i-1})+a(c_i)\}.
\]
A direct two-variable maximization gives
\[
        \max_{u,v\in[-1,1]}
        \{d(u,v)+a(u)+a(v)\}
        =\frac{\sqrt5-1}{4}<\frac13.
\]
\end{proof}

For convenience, we switch notation from $m$ to $n$.
\begin{lemma}[Strict increase with each dimension]
\label{lem:strict-increase}
Assume \(Q_n>0\). Then \(Q_{n+1}>Q_n\).
\end{lemma}

\begin{proof}
For fixed \(n\), the set of Jacobi paths is compact: this is the set of $(c_0, c_1, \ldots, c_n)$, where $c_0, c_n \in \{\pm 1\}$, and $c_1, \ldots, c_{n-1} \in [-1,1]$. Moreover, the map
\[
        c\longmapsto \lambda_{\max}J_n(c)
\]
is continuous. So, the supremum defining \(Q_n\) is attained. Let
\[
        c_0, c_1,\ldots,c_{n-1}, c_n
\]
be a maximizing path, where $c_0, c_n \in \{\pm1\}$.

Now, write $J_n(c)$ as a direct sum of irreducible blocks: $$J_n(c) = \bigoplus_{\ell \in [n]} J_n^{(\ell)}(c) \,,$$
where the $J_n^{\ell}(c)$ are the irreducible blocks
(this decomposition always exists -- though it could be trivial if $\ell = 1$ -- and is unique, i.e.\ the blocks $J_n(c)$ are uniquely defined, if we do not allow for ``reordering'' of the standard basis vectors).

Since $J_n(c)$ is tridiagonal with off-diagonal entries $\left( J_n(c) \right)_{j,j+1} = \left( J_n(c) \right)_{j+1,j} = a(c_j) $, each irreducible block ends at an index $j$ such that $a(c_j) = 0$ (since this is where the off-diagonal entry between indices $j$ and $j+1$ vanishes). The latter happens if and only if $c_j = \pm 1$.

Let $J_n^{\ell}(c)$ be a block whose largest eigenvalue is $Q_n$. Suppose this block is
supported on the standard basis vectors
\[
        e_p, e_{p+1},\ldots,e_q.
\]
Consider the Jacobi path $b = (b_0, b_1, \ldots, b_m)$, where
\[
        b_0=c_{p-1},\quad b_1=c_p,\quad b_2=c_{p+1},\ldots,\quad
        b_m=c_q,
\]
and \(m=q-p+1\). Note that this is indeed a Jacobi path since $c_{p-1}, c_q \in \{\pm 1\}$, by the argument above (and this is also true in the edge cases -- if the block is the first or the last -- since $c_0, c_n \in \{\pm 1\}$). Now, notice that the block $J_n^{\ell}(c)$ is precisely equal to the $m \times m$ Jacobi matrix $J_m(b)$. The latter has diagonal entries
\[
\left(J_m(b)\right)_{i,i} = d(b_{i-1},b_i)
\]
and off-diagonal entries
\[
        \left(J_m(b)\right)_{i,i+1}=\left(J_m(b)\right)_{i+1,i}=a(b_i).
\]
By construction, the block is irreducible, which implies
\[
        a(b_i)>0
        \qquad \textnormal{for all } \, 1\le i<m \,.
\]
Let \(x\) be the normalized top eigenvector of $J_m(b)$:
\[
        J_m(b)x=Q_nx,
        \qquad
        \|x\|_2=1.
\]
Now, for a large enough $\mu>0$, the matrix $J_m(b)+\mu I$ has \emph{all} entries strictly positive, since, as mentioned above, the off-diagonal entries of $J_m(b)$ are strictly positive. $J_m(b)+\mu I$ is also still symmetric, and thus has real eigenvalues. The Perron-Frobenius theorem says that a matrix with all positive entries has a top eigenvector with all positive entries. Let $x$ be such an eigenvector. Then, $x$ is also a top eigenvector of $J_m(b)$. Since, by construction, the top eigenvalue of $J_m(b)$ is $Q_n$, we have 
\[
        J_m(b)x=Q_nx,
        \qquad
        \|x\|_2=1\,.
\]
Moreover, in particular, we have $x_1>0$.

Now fix \(0\le\delta<1\). Define a new Jacobi path
\[
        \widetilde c^{\,\delta}_0,\ldots,
        \widetilde c^{\,\delta}_{n+1}
\]
by inserting the value \(c_{p-1}(1-\delta)\) between \(c_{p-1}\) and
\(c_p\), so that:
\[
  \begin{cases}      \widetilde c^{\,\delta}_j=c_j
        \qquad\qquad (0\le j\le p-1), \\
     \widetilde c^{\,\delta}_p=c_{p-1}(1-\delta),\\
        \widetilde c^{\,\delta}_j=c_{j-1}
        \qquad \,\,\,\,\,\, (p+1\le j\le n+1).
\end{cases}
\]
This is still a valid Jacobi path since the endpoints are still $\pm1$, but it is longer by $1$ compared to the original $c$. Let $J_{n+1}(\tilde{c}^{\delta})$ be its Jacobi matrix.

The off-diagonal entry
between the inserted ``vertex'' and the first ``vertex'' of the old block is
\[
        J_{n+1}(\tilde{c}^{\delta})_{p,p+1} = J_{n+1}(\tilde{c}^{\delta})_{p+1,p} = a(\sigma(1-\delta)) =: w_\delta \,.
\]
The diagonal entry of the inserted vertex is
\begin{equation}
\label{eq:23}
        J_{n+1}(\tilde{c}^{\delta})_{p,p} =
        d(\sigma,\sigma(1-\delta)) = 
        -\left(1-\frac{\sigma}{2}\right)\delta =: \alpha_\delta \,,
\end{equation}
and the change in the first diagonal entry of the old block is
\begin{equation}
\label{eq:24}
        J_{n+1}(\tilde{c}^{\delta})_{p+1,p+1} - J_{n}(c)_{p,p} = d(\sigma(1-\delta),b_1)-d(\sigma,b_1)
        =
        -\sigma\delta\left(b_1+\frac12\right) =: \beta_\delta
\end{equation}

Now, consider the following vector supported only on the inserted vertex and on the
previous block:
\[
        y_\delta=(0^{p-1}, t,x_1,\ldots,x_m, 0^{n-q}),
\]
where $t \in \mathbb{R}$, and $x$ is the eigenvector of the old block defined earlier. Then,
       $\|y_\delta\|^2=1+t^2$, and
\[
\langle y_\delta,J_{n+1}(\widetilde c^{\,\delta})y_\delta\rangle
        = Q_n+\beta_\delta x_1^2
        +2w_\delta t x_1 + \alpha_\delta t^2\,.
\]
So, 
\[\frac{\langle y_\delta,J_{n+1}(\widetilde c^{\,\delta})y_\delta\rangle}
     {\|y_\delta\|_2^2}
-Q_n
=
\frac{
(\alpha_\delta-Q_n)t^2
+2w_\delta t x_1
+\beta_\delta x_1^2
}{1+t^2}.
\]
Choose
\[
        t=\frac{w_\delta x_1}{Q_n-\alpha_\delta}.
\]
This is well-defined for all sufficiently small \(\delta\), because
\(Q_n>0\) and \(\alpha_\delta\le0\). With this choice,
\begin{equation}
\label{eq:rayleigh}
\frac{\langle y_\delta,J_{n+1}(\widetilde c^{\,\delta})y_\delta\rangle}
     {\|y_\delta\|_2^2}
-Q_n
=
\frac{x_1^2}{1+t^2}
\left(
        \frac{w_\delta^2}{Q_n-\alpha_\delta}
        +\beta_\delta
\right).
\end{equation}
Now, $w_\delta^2
        =
        a(\sigma(1-\delta))^2
        =
        \frac{2\delta-\delta^2}{4}$, and so, using \eqref{eq:23} and \eqref{eq:24}, we have
\[
\lim_{\delta\downarrow0}
\frac1\delta
\left(
        \frac{w_\delta^2}{Q_n-\alpha_\delta}
        +\beta_\delta
\right)
=
\frac{1}{2Q_n}
-\sigma\left(b_1+\frac12\right).
\]
Since \(b_1\in[-1,1]\) and \(\sigma\in\{\pm1\}\),
\[
        \sigma\left(b_1+\frac12\right)\le \frac32.
\]
By Lemma~\ref{lem:rowsum}, \(Q_n<1/3\). Hence,
\[
        \frac{1}{2Q_n}
        -\sigma\left(b_1+\frac12\right)
        \ge
        \frac{1}{2Q_n}-\frac32
        >0.
\]
Thus, for all sufficiently small \(\delta>0\),
\[
        \frac{w_\delta^2}{Q_n-\alpha_\delta}
        +\beta_\delta>0.
\]
Substituting this back into \eqref{eq:rayleigh}, and since \(x_1>0\), we have
\[\frac{\langle y_\delta,J_{n+1}(\widetilde c^{\,\delta})y_\delta\rangle}
     {\|y_\delta\|_2^2} - Q_n > 0 \,.
\]
Thus,
\[
        \lambda_{\max}\left(J_{n+1}(\widetilde c^{\,\delta})\right)>Q_n.
\]
Taking the supremum over all possible Jacobi paths gives
\[
        Q_{n+1}>Q_n.
\]
\end{proof}

\begin{corollary}[No Jacobi path attains $q_*$]
\label{cor:no-finite-path}
For all $n \in \mathbb{N}$,
\[
        Q_n<q_*.
\]
Equivalently, no finite signed-endpoint Jacobi path, with any endpoint
signs, attains \(q_*\).
\end{corollary}

\begin{proof}
Recall that \[
Q_n=\sup_{\substack{c_1,\ldots,c_{n-1}\in[-1,1]\\ c_0,c_n \in \{\pm 1\}}} \lambda_{\max}\left(J_n(c)\right),
        \qquad
        q_*:=\sup_{n\ge1}Q_n.
\]
Since $q_*>1/4$, there exists some \(n\) with
\(Q_n>0\).  If \(Q_n=q_*\) for some finite \(n\), then Lemma~\ref{lem:strict-increase}
would give
\[
        Q_{n+1}>Q_n=q_*,
\]
contradicting the definition
\[
        q_*=\sup_{n\ge1}Q_n.
\]
\end{proof}
\section{A certificate upper-bounding the value of arbitrary strategies}
\label{sec:upper-bound}

In this section, we prove an upper bound on the optimal value of arbitrary strategies for $I_{3322}$. The proof of this upper bound crucially goes through finding a ``certificate'' $H$, which will be a scalar function satisfying certain inequalities. This certificate will be especially important later when trying to understand the structure of optimal strategies, as these will have to saturate the inequalities involving $H$. 

We will start by stating the properties of the desired certificate $H$, and proving the upper-bound on the value of arbitrary strategies assuming the existence of such an $H$. We will then spend the rest of the section proving the existence of $H$.

\subsection{The desired certificate}
Recall that we defined, for $t,u,v \in [-1,1]$,
\begin{equation}
\label{eq:10}
        a(t)=\frac12\sqrt{1-t^2},
        \qquad
        d(u,v)=uv+\frac{u-v}{2}-1.
\end{equation}
The following is the desired certificate.
\begin{theorem}[Scalar certificate]
\label{thm:scalar-certificate}
There exists a concave function $H:[-1,1]\to[0,\infty)$ such that
\begin{equation}
\label{eq:Hstar-cert}
        H(v)+d(u,v)+H(-u)\le q_*
        \qquad\forall u,v\in[-1,1],
\end{equation}
and
\begin{equation}
\label{eq:Hstar-product}
        H(t)H(-t)\ge a(t)^2
        \qquad\forall t\in[-1,1].
\end{equation}
\end{theorem}
\noindent We defer the proof to Section~\ref{sec:certificate-proof}.

\subsection{Upper bounding the value of arbitrary strategies assuming existence of the certificate}
\label{sec:upper-bound-subsec}

\begin{theorem}[$q_*$ upper bound]
\label{thm:universal-upper}
Every finite-dimensional strategy for $I_{3322}$ has value at most $q_*$. The same bound holds for separable infinite-dimensional strategies.
\end{theorem}

\begin{proof}[Proof (assuming Theorem~\ref{thm:scalar-certificate})]
We will first assume a finite-dimensional strategy. As we discuss at the end, the proof will immediately extend to separable infinite-dimensional strategies. Without loss of generality (via purification and Naimark dilation), it is enough
to consider a pure state and projective binary measurements.  Let the strategy be specified by the state
\[
        \ket{\psi}\in \mathcal H_A\otimes\mathcal H_B
\]
and projectors $A_i, B_i$ for $i \in \{1,2,3\}$ on $\mathcal H_A$ and $\mathcal H_B$. Define
\[
        Z_A=A_1+A_2-I,\qquad X_A=A_2-A_1,
\]
\[
        Z_B=B_1+B_2-I,\qquad X_B=B_2-B_1,
\]
and
\[
        W_A=2A_3-I,\qquad W_B=2B_3-I.
\]
Then, all of these operators are Hermitian, and one can verify directly that
\[
        Z_AX_A+X_AZ_A=0,\qquad Z_A^2+X_A^2=I,
\]
and similarly for \(Z_B,X_B\). The above identities imply the following useful structure.
\begin{lemma}
\label{lem:structure}
Let $Z,X$ be self-adjoint operators on a finite-dimensional Hilbert space satisfying $ZX+XZ=0$ and $Z^2+X^2=I$.  Then there is an orthonormal basis $\{e_{\alpha}\}$ in which $Z$ and $X$ are simultaneously block-diagonal with $2\times 2$ (or $1 \times 1$) blocks such that the following holds. 

Denote by $\bar{\alpha}$ the index paired up with $\alpha$ to form a $2\times 2$ block. For any pair $(e_{\alpha}, e_{\bar{\alpha}})$ corresponding to a $2 \times 2$ block, there exists $z_\alpha\in (-1,1)$ with $z_{\alpha} \neq 0$ such that
\[
        Z e_\alpha=z_\alpha e_\alpha,
        \qquad
Z e_{\bar{\alpha}}=-z_\alpha e_{\bar{\alpha}},
\]
and
\[
        X e_\alpha=\sqrt{1-z_\alpha^2}\,e_{\bar\alpha}.
\]
For any $e_{\alpha}$ corresponding to a $1\times 1$ block, either: 
\begin{itemize}
\item there exists $z_{\alpha} \in \{\pm1\}$ such that $Z e_{\alpha} = z_{\alpha} e_{\alpha}$ and $X e_\alpha=0$, or
\item $Z e_{\alpha} = 0$ and $Xe_{\alpha} = \varepsilon_{\alpha} e_{\alpha}$ for some $\varepsilon_{\alpha} \in \{\pm 1\}$.
\end{itemize}
\end{lemma}
The proof of this lemma is straightforward and follows from an application of Jordan's lemma combined with the two identities in the lemma statement (we include this proof for completeness in Appendix~\ref{sec:app}). 

Now, back to the proof of Theorem~\ref{thm:universal-upper}. We invoke Lemma~\ref{lem:structure} for the pair $X_A, Z_A$ defined at the beginning of the proof. This tells us that there is
an orthonormal basis $\{e_\alpha\}$ for Alice such that $Z_A$ and $X_A$ are simultaneously block-diagonal with $2\times 2$ or $1\times 1$ blocks. Moreover, for each pair $(e_{\alpha}, e_{\bar{\alpha}})$ corresponding to a $2$-dimensional block, there is a $z_\alpha\in (-1,1)$, such that
\begin{equation}
\label{eq:structure-1}
        Z_A e_\alpha=z_\alpha e_\alpha,
        \qquad
Z_A e_{\bar{\alpha}}=-z_\alpha e_{\bar{\alpha}},
\end{equation}
and
\begin{equation}
\label{eq:structure-2}
        X_A e_\alpha=\sqrt{1-z_\alpha^2}\,e_{\bar\alpha}.
\end{equation}
From the above, it also follows that $z_{\bar{\alpha}} = - z_{\alpha}$. And for any $e_{\alpha}$ corresponding to a $1\times 1$ block, either:
\begin{itemize}
\item there exists $z_{\alpha}\in \{\pm 1\}$ such that $Z e_{\alpha} = z_{\alpha} e_{\alpha}$ and $X e_\alpha=0$, or
\item $Z e_{\alpha} = 0$ and $Xe_{\alpha} = \varepsilon_\alpha e_{\alpha}$ for some $\varepsilon_\alpha \in \{\pm 1\}$.
\end{itemize}


Similarly, Bob has an analogous orthonormal
basis $\{f_\beta\}$ and corresponding values $w_\beta$. Now, we can write the state $\ket{\psi}$ in these two bases: 
\[
        \ket{\psi}=\sum_{\alpha,\beta}\psi_{\alpha\beta}
        e_\alpha\otimes f_\beta.
\]
for some coefficients $\psi_{\alpha \beta}$.
Define
\[
        p_{\alpha\beta}=|\psi_{\alpha\beta}|^2,
        \qquad
        r_\alpha=\sum_\beta p_{\alpha\beta},
        \qquad
        \ell_\beta=\sum_\alpha p_{\alpha\beta}.
\]
In words, $r_\alpha$ is the marginal mass on Alice's $\alpha$-th index, and
$\ell_\beta$ is the marginal mass on Bob's $\beta$-th index.

For each Alice index $\alpha$, define the Bob-side vector
\[
        R_\alpha:=\sum_\beta \psi_{\alpha\beta}f_\beta\in\mathcal H_B\,.
\]
Then, $\|R_\alpha\|^2=r_\alpha$.
Similarly, for each Bob index \(\beta\), define the Alice-side vector
\[
        C_\beta:=\sum_\alpha \psi_{\alpha\beta}e_\alpha\in\mathcal H_A,
\]
and $\|C_\beta\|^2=\ell_\beta$. Now, recall that the $I_{3322}$ Bell operator is
\[
        I_{3322}
        =
        Z_A\otimes Z_B+\frac12 Z_A\otimes I
        -\frac12 I\otimes Z_B
        +\frac12 X_A\otimes W_B
        +\frac12 W_A\otimes X_B
        -I\otimes I.
\]
The ``diagonal'' part of the Bell expression is
\begin{align}
&\bra{\psi}
\left(
Z_A\otimes Z_B+\frac12 Z_A\otimes I
-\frac12 I\otimes Z_B-I\otimes I
\right)\ket{\psi} \nonumber \\
&\qquad =
\sum_{\alpha,\beta}
p_{\alpha\beta}
\left(
z_\alpha w_\beta+\frac12 z_\alpha-\frac12 w_\beta-1
\right) \nonumber \\
&\qquad =
\sum_{\alpha,\beta}p_{\alpha\beta}d(z_\alpha,w_\beta). \label{eq:diagonal-bound}
\end{align}
Next, we bound the term containing \(X_AW_B\) (the $W_AX_B$ term is analogous). Define the operator
\[
        K_A
        :=
        \operatorname{Tr}_A\!\left[
        (X_A\otimes I)|\psi\rangle\langle\psi|
        \right]
\]
on $\mathcal{H}_B$.
Then it is straightforward to verify that $K_A$ is Hermitian (since $X_A$ is Hermitian), and by properties of the partial trace,
\begin{equation}
\label{eq:off-diagonal-term}
        \langle\psi | X_A\otimes W_B | \psi\rangle
        =
        \operatorname{Tr}(W_BK_A).
\end{equation}
Let us compute \(K_A\) in terms of the vectors \(R_\alpha\).  Note first, by the definition of $R_{\alpha}$, that
\[
        \ket{\psi}=\sum_\alpha e_\alpha\otimes R_\alpha \,.
\]
So,
\[
\begin{aligned}
        K_A
        &=
        \operatorname{Tr}_A
        \left[
        \sum_{\alpha,\gamma}
        X_A|e_\alpha\rangle\langle e_\gamma|
        \otimes
        |R_\alpha\rangle\langle R_\gamma|
        \right] \\
        &=
        \sum_{\alpha,\gamma}
        \langle e_\gamma,X_Ae_\alpha\rangle\,
        |R_\alpha\rangle\langle R_\gamma|.
\end{aligned}
\]
From the structural identities \eqref{eq:structure-1} and \eqref{eq:structure-2} for the 2-dimensional blocks (and the special cases of 1-dimensional blocks), the above becomes
$$ K_A = \sum_{\alpha} \langle e_{\bar{\alpha}}| X_A |e_{\alpha} \rangle |R_\alpha\rangle\langle R_{\bar{\alpha}}| \,, $$
with the convention that we define $\bar{\alpha} = \alpha$ for 1-dimensional blocks.

Now, for a 2-dimensional block $b=\{\alpha,\bar\alpha\}$, we have $\langle e_{\bar{\alpha}}| X_A |e_{\alpha} \rangle = 2 a(z_{\alpha}) = 2a(z_{\bar{\alpha}})$, where recall that we defined $a(t) = \frac12 \sqrt{1-t^2}$. So, the contribution to $K_A$ from a block $b=\{\alpha,\bar\alpha\}$ is
\begin{equation}
\label{eq:2d-block}
        K_A(b)
        =
        2a(z_{\alpha})
        \left(
        |R_\alpha\rangle\langle R_{\bar\alpha}|
        +
        |R_{\bar\alpha}\rangle\langle R_\alpha|
        \right).
\end{equation}
For a one-dimensional block \(b=\{\alpha\}\), where \(z_\alpha=0\)
and \(X_Ae_\alpha=\varepsilon_\alpha e_\alpha\), the block contribution is
\[
        K_A(b)
        =
        \varepsilon_\alpha |R_\alpha\rangle\langle R_\alpha|.
\]
For a one-dimensional block \(b=\{\alpha\}\) where \(z_\alpha=\pm1\), the block
contribution is $K_A(b)= 0$.  We can then write
\[
        K_A=\sum_{b} K_A(b),
\]
where the sum is over all the blocks.

Since \(W_B\) is Hermitian and unitary, \(\|W_B\|_{\mathrm{op}}=1\).
By trace-norm duality,
\[
        \operatorname{Tr}(W_BK_A)\le \|K_A\|_1.
\]
Substituting this back into \eqref{eq:off-diagonal-term} gives
\begin{equation}
\label{eq:22}
\langle\psi | X_A\otimes W_B | \psi\rangle \leq \| K_A\|_1 \,.
\end{equation}
By a triangle inequality, we have
\begin{equation}
\label{eq:23}
        \|K_A\|_1
        \leq
        \sum_{b}\|K_A(b)\|_1.
\end{equation}
Now, for a rank-one operator \(|u\rangle\langle v|\), we have
\[
        \bigl\||u\rangle\langle v|\bigr\|_1
        =
        \|u\|\,\|v\|.
\]
Hence, for a two-dimensional Alice block $b = \{\alpha, \bar{\alpha}\}$, we have, from \eqref{eq:2d-block} (and a triangle inequality),
\begin{align}
        \|K_A(b)\|_1
        &\le
        2a(z_{\alpha})
        \left(
        \bigl\||R_\alpha\rangle\langle R_{\bar\alpha}|\bigr\|_1
        +\bigl\||R_{\bar\alpha}\rangle\langle R_\alpha|\bigr\|_1
        \right) \nonumber\\
        &=4a(z_{\alpha})\|R_\alpha\|\,\|R_{\bar\alpha}\| \nonumber \\
        &= 4a(z_{\alpha})\sqrt{r_\alpha r_{\bar\alpha}}. \label{eq:83}
\end{align}

Now, let $H:[-1,1]\to[0,\infty)$ be the certificate supplied by Theorem~\ref{thm:scalar-certificate}. We remark that Theorem~\ref{thm:scalar-certificate} guarantees existence of a \emph{concave} function $H$ satisfying the desired inequalities. However, concavity is not actually needed for the current proof of Theorem~\ref{thm:universal-upper}, i.e.\ of the $q_*$ upper bound on the value of any strategy. Instead, concavity will be crucial later in Section~\ref{sec:no-finite-dimensional-attainment} when we show that no finite-dimensional strategy can attain $q_*$.

Theorem~\ref{thm:scalar-certificate} gives us that $H$ satisfies
\begin{equation*}
        H(v)+d(u,v)+H(-u)\le q_*
        \qquad\forall u,v\in[-1,1],
\end{equation*}
and
\begin{equation*}
        H(t)H(-t)\ge a(t)^2
        \qquad\forall t\in[-1,1].
\end{equation*}
Then, for a two-dimensional block $b = \{\alpha, \bar{\alpha}\}$, we have
\begin{equation}
\label{eq:34}
        K_A(b) \leq 4a(z_{\alpha})\sqrt{r_\alpha r_{\bar\alpha}}
        \le 4\sqrt{H(z_{\alpha})H(-z_{\alpha})r_\alpha r_{\bar\alpha}} \le
        2H(-z_{\alpha})r_\alpha+2H(z_{\alpha})r_{\bar\alpha}.
\end{equation}
where the last inequality is by AM-GM.

For a one-dimensional block $b = \{\alpha\}$ with $z_{\alpha} = 0$,
\[
        \|K_A(b)\|_1
        =
        r_\alpha,
\] 
Since $H(0)^2\ge a(0)^2=1/4$, and so
\(H(0)\ge1/2\), we have
\begin{equation}
\label{eq:211}
        \|K_A(b)\|_1
       = r_\alpha\le 2H(0)r_\alpha.
\end{equation}
Finally, one-dimensional blocks with $z_{\alpha} = \pm 1$ contribute zero. 

So, summing over all the blocks, and combining with Equations \eqref{eq:22} and \eqref{eq:23}, gives
\begin{equation}
\label{eq:off-diagonal-bound1}
       \langle\psi | X_A\otimes W_B | \psi\rangle  \le
        2 \sum_\alpha H(-z_\alpha)r_\alpha.
\end{equation}
where the sum on the RHS now ranges over all possible ``labels'' $\alpha$ (so it includes both labels involved in a two-dimensional block).
The analysis is identical on Bob's side, and we similarly get
\begin{equation}
\label{eq:off-diagonal-bound2}
        \langle\psi | W_A\otimes X_B | \psi\rangle
        \le
        2 \sum_\beta H(w_\beta)\ell_\beta.
\end{equation}
Note that we choose the upper bound $2 \sum_{\beta} H(w_{\beta}) \ell_{\beta}$ instead of $2 \sum_{\beta} H(-w_{\beta}) \ell_{\beta}$ for Bob. Both upper bounds are valid (we have the freedom to choose how to split factors at the AM-GM step above), but the stated one is what is compatible with our certificate's guarantees in order to achieve the final bound.
Combining the bounds from \eqref{eq:off-diagonal-bound1} and \eqref{eq:off-diagonal-bound2} with the bound for the ``diagonal'' part from \eqref{eq:diagonal-bound}, we have
\begin{align}
\langle \psi | I_{3322} |\psi \rangle
&\le
\sum_{\alpha,\beta}p_{\alpha\beta}d(z_\alpha,w_\beta)
+
\sum_\alpha H(-z_\alpha)r_\alpha
+
\sum_\beta H(w_\beta)\ell_\beta \nonumber\\
&=
\sum_{\alpha,\beta}
p_{\alpha\beta}
\bigl[
d(z_\alpha,w_\beta)+H(-z_\alpha)+H(w_\beta)
\bigr]. \label{eq:35}
\end{align}
Since
        $H(v)+d(u,v)+H(-u)\le q_*$ for all $u,v\in[-1,1]$,
each summand is at most $q_*$. Since $\sum_{\alpha,\beta}p_{\alpha\beta}=1$, we get
$$\langle \psi | I_{3322} |\psi \rangle
\leq  q_* \,,$$
as desired. 

The proof above assumed a finite-dimensional
strategy. However, the same $q_*$ bound extends to separable infinite-dimensional strategies. This is because the set of correlations achievable with infinite-dimensional strategies lies in the closure of the set of correlations achievable with finite-dimensional strategies~\cite{scholz2008tsirelson}, and the \(I_{3322}\) value is
a linear function of the correlation. Hence, the bound
\(I_{3322}\leq q_*\) passes immediately to the limit. This completes the proof of Theorem~\ref{thm:universal-upper}.
\end{proof}

\subsection{Proving existence of the certificate}
\label{sec:certificate-proof}
In this section, we prove Theorem~\ref{thm:scalar-certificate}, i.e.\ that the desired certificate $H$ exists, which in turn concludes the proof of the $q_*$ upper bound from Theorem~\ref{thm:universal-upper}. 

\subsubsection{Riccati pivots}
\label{sec:riccati-pivots}
For convenience, we recall once more the definition of $q_* = \sup_{n\ge1}Q_n$, where $$Q_n=\sup_{\substack{c_1,\ldots,c_{n-1}\in[-1,1]\\ c_0,c_n \in \{\pm 1\}}} \lambda_{\max}\left(J_n(c)\right)\,,$$
and $J_n(c)$ is the Jacobi matrix corresponding to the Jacobi path $c = (c_0, \ldots, c_n)$ (from Definition~\ref{def:signed-endpoint-jacobi-path}).
Recall also the definitions of $a(t)$ and $d(u,v)$ from Equation \eqref{eq:10}.

For an overview of the main ideas we refer the reader to the technical overview.

\begin{lemma}[Riccati pivots]
\label{lem:pivots}
Fix \(q>q_*\). Let $m \in \mathbb{N}$. Let $c = (c_0, \ldots, c_m)$ be a ``path'' with $c_0 = 1$ and $c_1, \ldots, c_m \in [-1,1]$ (here we don't strictly require the right endpoint to be $\pm 1$ as in Definition~\ref{def:signed-endpoint-jacobi-path}). Define
\[
        p_1=q-d(c_0,c_1),
\]
and, recursively,
\[
        p_i
        =
        q-d(c_{i-1},c_i)
        -
        \frac{a(c_{i-1})^2}{p_{i-1}},
        \qquad 2\le i\le m.
\]
Then every \(p_i\) is well-defined and satisfies
\[
        p_i\ge q-q_*>0.
\]
We refer to the $p_i$ as the ``Riccati pivots'' corresponding to the path $c$.
\end{lemma}

\begin{proof}
For \(1\le r\le m\), we define \(J_r(c_0,\ldots,c_r)\) to be the
\(r\times r\) ``Jacobi matrix'' associated with the truncated path
$(c_0,\ldots,c_r)$ (we use quotations marks because $(c_0,\ldots,c_r)$ may not technically be a Jacobi path, since the right endpoint may not be $\pm1 $). More precisely, let
\[
     \bigl(J_r(c_0,\ldots,c_r)\bigr)_{ii}
        =
        d(c_{i-1},c_i),
\]
and
\[
        \bigl(J_r(c_0,\ldots,c_r)\bigr)_{i,i+1}
        =
        \bigl(J_r(c_0,\ldots,c_r)\bigr)_{i+1,i}
        =
        a(c_i).
\]
Set
\[
        A_r=qI_r-J_r(c_0,\ldots,c_r).
\]
Then \(A_r\) is tridiagonal and symmetric with diagonal entries
\[
        \alpha_i=q-d(c_{i-1},c_i),
        \qquad 1\le i\le r,
\]
and off-diagonal entries
\[
        \beta_i=-a(c_i),
        \qquad 1\le i<r.
\]

We first show that
$$\lambda_{\max} \left( J_r(c_0,\ldots,c_r) \right)\le q_*\,.$$
If \(c_r= \pm 1\), then \(J_r(c_0,\ldots,c_r)\) is the Jacobi matrix of a bona fide Jacobi path as in Definition~\ref{def:signed-endpoint-jacobi-path}, and, by definition, $\lambda_{\max} \left( J_r(c_0,\ldots,c_r) \right) \le q_*$.
If $c_r \in (-1,1)$, append the value $c_{r+1} = 1$ to the path. Then
\[
        c_0=1,\quad c_1,\ldots,c_r,\quad c_{r+1}=1
\]
is a bona fide Jacobi path and
\(J_r(c_0,\ldots,c_r)\) is a principal submatrix of the bona fide
\((r+1)\times(r+1)\) Jacobi matrix $J_{r+1}(c_0, \ldots, c_{r+1})$.  By the Cauchy interlacing theorem, and since a Jacobi matrix is Hermitian, the largest eigenvalue of the principal submatrix is upper bounded by the largest eigenvalue of the full Jacobi matrix, and thus by \(q_*\). Therefore, we have
\[
        A_r=qI_r-J_r(c_0,\ldots,c_r)
        \succeq (q-q_*)I_r
        \succ0.
\]
In particular, \(A_r\) is strictly positive definite for every \(r\).

We now compute the \(LDL^T\) factorization of \(A_r\).  Since \(A_r\)
is positive definite and tridiagonal, it has a factorization
\[
        A_r=L_rD_rL_r^T,
\]
where \(L_r\) is lower bidiagonal with all ones on the diagonal, and
\[
        D_r=\operatorname{diag}(p^{(r)}_1,\ldots,p^{(r)}_r)
\]
has strictly positive diagonal entries.  Write
\[
        (L_r)_{i+1,i}=\ell_i.
\]
Multiplying out \(L_rD_rL_r^T\), its first diagonal entry is $p^{(r)}_1$. So, we have
\[
        p^{(r)}_1=\alpha_1=q-d(c_0,c_1).
\]
where recall that the $\alpha_i$ are the diagonal entries of $A_r$. Considering the $(i,i+1)$-entry gives the identity
\[
        \beta_i=p^{(r)}_i\ell_i,
\]
so
\[
        \ell_i=\frac{\beta_i}{p^{(r)}_i}
        =
        -\frac{a(c_i)}{p^{(r)}_i}.
\]
Finally, considering the \((i+1,i+1)\)-entry gives the identity
\[
        \alpha_{i+1}
        =
        p^{(r)}_{i+1}
        +
        \bigl(p^{(r)}_i\bigr)\ell_i^2.
\]
Substituting the formula for \(\ell_i\), we obtain
\[
        p^{(r)}_{i+1}
        =
        \alpha_{i+1}
        -
        \frac{\beta_i^2}{p^{(r)}_i}
        =
        q-d(c_i,c_{i+1})
        -
        \frac{a(c_i)^2}{p^{(r)}_i}.
\]
The latter is well-defined since $p_i^{(r)}$ is strictly positive. Thus the diagonal entries of \(D_r\) are exactly the Riccati pivots
defined above:
\begin{equation}
\label{eq:pivots}
        p^{(r)}_i=p_i,
        \qquad 1\le i\le r.
\end{equation}

It remains to prove the lower bound \(p_i\ge q-q_*\).  Fix \(i\).  Apply
the preceding discussion to the truncated matrix
\[
        A_i=qI_i-J_i(c_0,\ldots,c_i).
\]
The final diagonal pivot in the \(LDL^T\) factorization of \(A_i\) is
\(p_i\).  Since
\[
        A_i\succeq (q-q_*)I_i,
\]
we have
\[
        A_i^{-1}\preceq (q-q_*)^{-1}I_i.
\]
Since $L_i$ is lower bidiagonal with ones on the diagonal, one can verify that the last column of \(L_i^{-1}\) is
the last standard basis vector. Therefore, from
\[
        A_i^{-1}=L_i^{-T}D_i^{-1}L_i^{-1},
\]
we get
\[
(A_i^{-1})_{ii}=\frac1{p_i^{(i)}} = \frac{1}{p_i} \,,
\]
where the last equality is from Equation \eqref{eq:pivots} (which we have shown holds for  all $r$).
Hence
\[
        \frac1{p_i}
        =
        (A_i^{-1})_{ii}
        \le
        \frac1{q-q_*}.
\]
Thus
\[
        p_i\ge q-q_*>0.
\]
Since \(i\) was arbitrary, the above holds for all $i$.
\end{proof}

\begin{lemma}[Bellman supersolution]
\label{lem:bellman}
For every $q>q_*$ there is a function $h_q:[-1,1]\to(0,\infty)$ such that
\begin{equation}
\label{eq:bellman}
        h_q(v)+d(u,v)+\frac{a(u)^2}{h_q(u)}\le q
        \qquad\forall u,v\in[-1,1].
\end{equation}
Moreover, for all $v \in [-1,1]$,
\[
        q-q_*\le h_q(v)\le q-d(1,v)=q+\frac{1-v}{2}.
\]
\end{lemma}

\begin{proof}
For $v\in[-1,1]$, define $h_q(v)$ to be the infimum of the final Riccati pivot $p_m$ over all $m\in \mathbb{N}$ and all paths $c = (c_0,c_1,\ldots,c_m)$ with $c_0 = 1$ and $c_m =v$. Lemma~\ref{lem:pivots} gives $h_q(v)\ge q-q_*>0$. The path $(1,v)$ gives the upper bound.

Now, fix $u,v \in [-1,1]$. Choose a sequence of paths $\{c^{(k)}\}$ with right endpoint $u$ whose final pivots tend to $h_q(u)$ (such a sequence exists by definition of $h_q(u)$ as the infimum). Appending the value $v$ to each of the paths in the sequence produces a new sequence of paths $\{c'^{(k)}\}$ ending at $v$ whose final pivots are
\[
        q-d(u,v)-\frac{a(u)^2}{p^{(k)}},
\]
where $p^{(k)}$ is the final pivot for the old path $c^{(k)}$. By definition,
$$h_q(v) \leq q - d(u,v) -\frac{a(u)^2}{p^{(k)}} \,.$$
Rearranging and taking the limit as $k \rightarrow \infty$ (under which $p^{(k)} \rightarrow h_q(u)$), gives \eqref{eq:bellman}.
\end{proof}

\subsubsection{Symmetrizations}
\label{sec:product}

For convenience, we recall here that $a(t) = \frac12\sqrt{1-t^2}$, and $d(u,v) = uv + \frac{u-v}{2} -1$ for $t, u,v \in [-1,1]$.
\begin{lemma}
\label{lem:product-sat}
For every \(q>q_*\) there is a bounded function
\(h_0:[-1,1]\to[0,\infty)\) such that
\begin{equation}
\label{eq:product-sat}
        h_0(t)h_0(-t)=a(t)^2
        \qquad\forall t\in[-1,1],
\end{equation}
and
\begin{equation}
\label{eq:product-cert}
        h_0(v)+d(u,v)+h_0(-u)\le q
        \qquad\forall u,v\in[-1,1].
\end{equation}
\end{lemma}

\begin{proof}
Let \(h_q\) be the Bellman supersolution from
Lemma~\ref{lem:bellman}. Thus,
\begin{equation}
\label{eq:bellman}
        h_q(v)+d(u,v)+\frac{a(u)^2}{h_q(u)}\le q
        \qquad\forall u,v\in[-1,1],
\end{equation}
and, for all $t \in [-1,1]$,
\[
        q-q_*\le h_q(t)\le q+\frac{1-t}{2}\le q+1.
\]
In particular, \(h_q(t)>0\) for every \(t\in[-1,1]\).

Now, note that on the open interval \((-1,1)\), we have \(a(t)>0\).  Therefore there is
a real-valued function \(p:(-1,1)\to\mathbb R\) defined by
\[
        h_q(t)=a(t)e^{p(t)},
        \qquad |t|<1.
\]
Equivalently,
\[
        p(t)=\log\frac{h_q(t)}{a(t)}.
\]
Substituting this into \eqref{eq:bellman}, for
\(u,v\in(-1,1)\), gives
\[
        d(u,v)+a(v)e^{p(v)}+a(u)e^{-p(u)}\le q.
\]
For any real-valued function \(r:(-1,1)\to\mathbb R\), define
\[
        F(r)
        :=
        \sup_{u,v\in(-1,1)}
        \left\{
        d(u,v)+a(v)e^{r(v)}+a(u)e^{-r(u)}
        \right\}.
\]
The preceding inequality equivalently says that
\[
        F(p)\le q.
\]
We now show that \(F\) is convex.  Let $r_0,r_1: (-1,1) \rightarrow \mathbb{R}$, and let $\theta\in[0,1]$.  Let
\[
        r_\theta=\theta r_0+(1-\theta)r_1.
\]
For fixed $u,v \in (-1,1)$, convexity of the exponential function gives
\[
        e^{r_\theta(v)}
        \le
        \theta e^{r_0(v)}+(1-\theta)e^{r_1(v)} \,,
\]
and
\[
        e^{-r_\theta(u)}
        =
        e^{\theta(-r_0(u))+(1-\theta)(-r_1(u))}
        \le
        \theta e^{-r_0(u)}+(1-\theta)e^{-r_1(u)}.
\]
So, since $a(v), a(u) >0$,
\[
\begin{aligned}
&d(u,v)+a(v)e^{r_\theta(v)}+a(u)e^{-r_\theta(u)}\\
&\qquad\le
\theta\left[
d(u,v)+a(v)e^{r_0(v)}+a(u)e^{-r_0(u)}
\right]\\
&\qquad\quad+
(1-\theta)\left[
d(u,v)+a(v)e^{r_1(v)}+a(u)e^{-r_1(u)}
\right].
\end{aligned}
\]
Taking the supremum over $u,v\in(-1,1)$ gives
\[
        F(r_\theta)\le \theta F(r_0)+(1-\theta)F(r_1).
\]
Thus, $F$ is convex.

Next we record a symmetry of \(F\).  For a function $r: (-1,1) \rightarrow \mathbb{R}$, define its
``reflected sign-reversal'' \(r^\sharp\) by
\[
        r^\sharp(t):=-r(-t),
        \qquad t \in (-1,1).
\]
We claim that
\[
        F(r^\sharp)=F(r).
\]
Indeed,
\[
\begin{aligned}
F(r^\sharp)
&=
\sup_{u,v\in(-1,1)}
\left\{
d(u,v)+a(v)e^{r^\sharp(v)}+a(u)e^{-r^\sharp(u)}
\right\}\\
&=
\sup_{u,v\in(-1,1)}
\left\{
d(u,v)+a(v)e^{-r(-v)}+a(u)e^{r(-u)}
\right\}.
\end{aligned}
\]
Now make the change of variables
\[
        U=-v,\qquad V=-u.
\]
Since \(u,v\in(-1,1)\), also \(U,V\in(-1,1)\).  Using
\[
        a(-t)=a(t)
\]
and
\[
        d(-V,-U)=d(U,V),
\]
we obtain
\[
\begin{aligned}
F(r^\sharp)
&=
\sup_{U,V\in(-1,1)}
\left\{
d(U,V)+a(U)e^{-r(U)}+a(V)e^{r(V)}
\right\}\\
&=
\sup_{U,V\in(-1,1)}
\left\{
d(U,V)+a(V)e^{r(V)}+a(U)e^{-r(U)}
\right\}\\
&=
F(r).
\end{aligned}
\]

Now, define the odd part of \(p\) by
\[
        p_{odd}(t):=\frac{p(t)-p(-t)}2.
\]
We can write this as
\[
        p_{odd}=\frac{p+p^\sharp}{2}.
\]
By convexity of \(F\) and by the symmetry \(F(p^\sharp)=F(p)\),
\[
        F(p_{odd})
        =
        F\left(\frac{p+p^\sharp}{2}\right)
        \le
        \frac12F(p)+\frac12F(p^\sharp)
        =
        F(p)
        \le q.
\]

Now define $h_0: [-1,1] \rightarrow \mathbb{R}$ on the open interval by
\[
        h_0(t):=a(t)e^{p_{odd}(t)},
        \qquad t \in (-1,1),
\]
and set
\[
        h_0(1)=h_0(-1)=0.
\]
By construction, \(p_{odd}\) is odd, so
\[
        p_{odd}(-t)=-p_{odd}(t).
\]
Therefore, for $t\in (-1,1)$,
\[
\begin{aligned}
        h_0(t)h_0(-t)
        &=
        a(t)e^{p_{odd}(t)}\,a(-t)e^{p_{odd}(-t)}\\
        &=
        a(t)^2 e^{p_{odd}(t)-p_{odd}(t)}
        =
        a(t)^2.
\end{aligned}
\]
At \(t=\pm1\), both sides are zero, because \(a(\pm1)=0\) and
\(h_0(\pm1)=0\).  Hence
\[
        h_0(t)h_0(-t)=a(t)^2
        \qquad\forall t\in[-1,1].
\]

Now, the inequality \(F(p_{odd})\le q\) says that for all \(u,v\in(-1,1)\),
\[
        d(u,v)+a(v)e^{p_{odd}(v)}+a(u)e^{-p_{odd}(u)}\le q.
\]
Using the definition of \(h_0\) and the oddness of \(p_{odd}\), this becomes
\[
        d(u,v)+h_0(v)+h_0(-u)\le q
        \qquad\forall u,v\in(-1,1).
\]
Thus \eqref{eq:product-cert} holds in the interior.

It remains to justify the endpoint values and boundedness. We start with boundedness. From
\[
        h_q(t)=a(t)e^{p(t)}
\]
and the definition of \(p_{odd}\), we can write \(h_0\) directly in terms of
\(h_q\):
\[
\begin{aligned}
        h_0(t)
        &=
        a(t)\exp\left(\frac{p(t)-p(-t)}2\right)\\
        &=
        a(t)
        \sqrt{
        \frac{e^{p(t)}}{e^{p(-t)}}
        }\\
        &=
        a(t)
        \sqrt{
        \frac{h_q(t)/a(t)}{h_q(-t)/a(-t)}
        }\\
        &=
        a(t)\sqrt{\frac{h_q(t)}{h_q(-t)}}.
\end{aligned}
\]
Since
\[
        q-q_*\le h_q(s)\le q+1
        \qquad\forall s\in[-1,1],
\]
we get, for $t\in (-1,1)$,
\[
        0\le h_0(t)
        \le
        a(t)\sqrt{\frac{q+1}{q-q_*}}
        \le
        \frac12\sqrt{\frac{q+1}{q-q_*}}.
\]
Thus \(h_0\) is bounded.  


Finally, we are left with extending \eqref{eq:product-cert} from
\((-1,1)^2\) to \([-1,1]^2\). We start by observing that, since
\[
        0\le h_0(t)
        \le
        a(t)\sqrt{\frac{q+1}{q-q_*}},
\]
we have
\[
        h_0(t)\to0
        \qquad\text{as }t\to\pm1.
\]
This ``agrees'' with the definition \(h_0(\pm1)=0\).

Let \(u,v\in[-1,1]\).  Choose sequences
\(u_j,v_j\in(-1,1)\) as follows: if \(u\in(-1,1)\), take
\(u_j=u\) for all \(j\); if \(u=\pm1\), take any sequence
\(u_j\in(-1,1)\) with \(u_j\to u\).  Similarly, if
\(v\in(-1,1)\), take \(v_j=v\) for all \(j\); if \(v=\pm1\), take
\(v_j\in(-1,1)\) with \(v_j\to v\).

For every \(j\), the interior inequality gives
\[
        d(u_j,v_j)+h_0(v_j)+h_0(-u_j)\le q.
\]
As \(j\to\infty\), the term \(d(u_j,v_j)\) converges to \(d(u,v)\) by
continuity of \(d\).  If \(v\in(-1,1)\), then \(v_j=v\), so
\(h_0(v_j)=h_0(v)\) for all \(j\); if \(v=\pm1\), then
\(h_0(v_j)\to0=h_0(v)\) by the endpoint bound.  Likewise, if
\(u\in(-1,1)\), then \(h_0(-u_j)=h_0(-u)\) for all \(j\); if
\(u=\pm1\), then \(-u_j\to -u\in\{\pm1\}\), so
\(h_0(-u_j)\to0=h_0(-u)\).

Therefore, passing to the limit gives
\[
        d(u,v)+h_0(v)+h_0(-u)\le q.
\]
Thus \eqref{eq:product-cert} holds for all \(u,v\in[-1,1]\).
\end{proof}



\subsubsection{Concavification}
In the previous subsections, \ref{sec:riccati-pivots} and \ref{sec:product}, we proved the existence of a certificate function $h_0$ that is ``almost'' the desired $H$ from Theorem~\ref{thm:scalar-certificate}. There are two missing ingredients still, which we will obtain in this subsection:
\begin{itemize}
\item[(i)] The upper bound in \eqref{eq:bellman} is some $q>q_*$ rather than $q_*$ itself (and $h_0$ is allowed to depend on $q$).
\item[(ii)] The function $h_0$ is not necessarily concave.
\end{itemize}
Item (i) is essential to establish the desired $q_*$ upper bound on the value of any strategy. Item (ii) is not strictly necessary for the latter. Instead, it will be crucial in Section~\ref{sec:no-finite-dimensional-attainment} when proving that any arbitrary finite-dimensional strategy cannot attain $q_*$. Moreover, it turns out that our proof of the existence of the certificate $H$ with the correct $q_*$ upper bound does in fact go through a step in which we ``concavify'' $h_0$, as we describe next.


For a bounded function $f$ on $[-1,1]$, let $\cav(f)$ denote its least concave majorant. Equivalently,
\[
        \cav(f)(x)=\sup\left\{\sum_i\theta_i f(x_i):
        \theta_i\ge0,
        \sum_i\theta_i=1,
        \sum_i\theta_i x_i=x\right\}.
\]
By Caratheodory's theorem, in one dimension it is enough to use at most two points in the decomposition.

\begin{lemma}[Concavification preserves the certificate]
\label{lem:concavification}
Let $h_0:[-1,1]\to[0,\infty)$ be a bounded function that satisfies \eqref{eq:product-sat} and \eqref{eq:product-cert} for some $q>q_*$, and set $H=\cav(h_0)$.  Then $H$ is concave, non-negative, bounded, and satisfies
\begin{equation}
\label{eq:H-product}
        H(t)H(-t)\ge a(t)^2
        \qquad\forall t\in[-1,1],
\end{equation}
and
\begin{equation}
\label{eq:H-cert-q}
        H(v)+d(u,v)+H(-u)\le q
        \qquad\forall u,v\in[-1,1].
\end{equation}
\end{lemma}

\begin{proof}
Boundedness and non-negativity of $H$ clearly follow from boundedness and non-negativity of $h_0$. Moreover, since $H\ge h_0$, \eqref{eq:H-product} follows from \eqref{eq:product-sat}. 

We are left with proving \eqref{eq:H-cert-q}. We know $h_0$ satisfies \eqref{eq:product-sat}, which, setting $x=v$ and $y=-u$, we can rewrite as
\begin{equation}
\label{eq:30}
        h_0(x)+h_0(y)+d(-y,x)\le q. 
\end{equation}
Fix $x,y$ and $\eps>0$.  Choose convex decompositions
\[
        x=\sum_i\theta_i x_i,
        \qquad
        y=\sum_j\eta_j y_j,
\]
such that
\[
        H(x)\le\sum_i\theta_i h_0(x_i)+\eps,
        \qquad
        H(y)\le\sum_j\eta_j h_0(y_j)+\eps.
\]
We can now apply inequality \eqref{eq:30} to every pair $(x_i,y_j)$ and average with weights $\theta_i\eta_j$.  Since $d(-y,x)$ is affine separately in $x$ and $y$, we get
\[
        \sum_i\theta_i h_0(x_i)+\sum_j\eta_j h_0(y_j)+d(-y,x)\le q.
\]
Thus $H(x)+H(y)+d(-y,x)\le q+2\eps$.  Taking the limit $\eps\downarrow0$ gives \eqref{eq:H-cert-q}.
\end{proof}

We now have a \emph{concave} function $H$ that satisfies the inequalities \eqref{eq:H-product} and \eqref{eq:H-cert-q}. The only upgrade left to perform is to improve the upper bound in \eqref{eq:H-cert-q} from some $q>q_*$ to $q_*$ itself. We describe this next.

\begin{lemma}[Compactness for uniformly bounded concave functions]
\label{lem:concave-compactness}
Let $\left(G_k:[-1,1]\to[0,M]\right)$ be a sequence of concave functions. Then
there is a subsequence $(\tilde{G}_k)$, and a concave function
\(G^\circ:(-1,1)\to[0,M]\) such that
\[
        \tilde{G}_k\to G^\circ
\]
locally uniformly on \((-1,1)\).  Moreover \(G^\circ\) has finite
one-sided endpoint limits
\[
        G_*(-1):=\lim_{x\downarrow -1}G^\circ(x),
        \qquad
        G_*(1):=\lim_{x\uparrow 1}G^\circ(x),
\]
and the function \(G_*:[-1,1]\to[0,M]\) obtained by setting
\(G_*=G^\circ\) on \((-1,1)\) and using these endpoint values is concave
on \([-1,1]\).
\end{lemma}

\begin{proof}
Let
\[
        K_\delta=[-1+\delta,1-\delta]\subset(-1,1),
        \qquad 0<\delta<1.
\]
We first prove that each $G_k$ is Lipschitz on $K_{\delta}$ (with a Lipschitz constant independent of $k$).
Fix $k$. For \(x<y\) in \(K_\delta\), define the ``secant slope''
\[
        S_k(x,y):=\frac{G_k(y)-G_k(x)}{y-x}.
\]
For a concave function, secant slopes decrease as the interval moves to
the right: if \(r<s<t\), then
\[
        S_k(r,s)\ge S_k(r,t)\ge S_k(s,t).
\]
Since \(-1<x<y<1\), we can apply the above to get
\[
        S_k(-1,y)\ge S_k(x,y)\ge S_k(x,1).
\]
Since $G_k$ takes values in $[0,M]$, we have
\[
        S_k(-1,y)
        =
        \frac{G_k(y)-G_k(-1)}{y+1}
        \le
        \frac{M}{y+1}
        \le
        \frac{M}{\delta},
\]
and
\[
        S_k(x,1)
        =
        \frac{G_k(1)-G_k(x)}{1-x}
        \ge
        -\frac{M}{1-x}
        \ge
        -\frac{M}{\delta}.
\]
Therefore
\[
        |S_k(x,y)|\le \frac{M}{\delta}
        \qquad\forall x<y\text{ in }K_\delta.
\]
Equivalently, the functions \(G_k\) are Lipschitz on
\(K_\delta\), with Lipschitz constant \(M/\delta\). Note that this holds for any $k$ and $\delta>0$.

By Arzelà--Ascoli, for each fixed compact interval
\[
        K_\delta=[-1+\delta,1-\delta]\subset(-1,1),
\]
the sequence \((G_k)\) has a subsequence converging uniformly on
\(K_\delta\).  We need a single subsequence that converges uniformly on
every compact subinterval of \((-1,1)\).  We obtain it by a diagonal
argument.

For \(m\ge2\), set
\[
        K_m:=\left[-1+\frac1m,\,1-\frac1m\right].
\]
Then
\[
        K_2\subset K_3\subset K_4\subset\cdots,
        \qquad
        \bigcup_{m\ge2}K_m=(-1,1).
\]
Moreover, every compact set \(K\subset(-1,1)\) is contained in \(K_m\)
for some sufficiently large \(m\).

We now construct nested subsequences.  Starting from the original sequence
\((G_k)\), Arzelà--Ascoli gives a subsequence, call it
\[
        (G^{(2)}_j)_{j\ge1},
\]
which converges uniformly on \(K_2\). Note, crucially, that we have already shown that each function in $(G^{(2)}_j)_{j\ge1}$ is Lipschitz on $K_\delta$ for any $\delta >0$, and so in particular on $K_3$. Thus, from this subsequence, apply
Arzelà--Ascoli again on \(K_3\) to obtain a further subsequence $(G^{(3)}_j)_{j\ge1}$ that converges uniformly on \(K_3\). Formally, by induction, we obtain that for each
\(m\ge2\) there is a subsequence
\[
        (G^{(m)}_j)_{j\ge1}
\]
that converges uniformly on $K_m$. Moreover, for all $m \geq 2$, $(G^{(m+1)}_j)_{j\ge1}$ is a subsequence of $(G^{(m)}_j)_{j\ge1}$.

Let \(G^{[m]}\) (with domain $K_m$) denote the uniform limit on \(K_m\). Crucially these limits, for different $m$'s, are
compatible on overlaps. Indeed, since $(G^{(m+1)}_j)_j$ is a subsequence of $(G^{(m)}_j)_j$,
and since \((G^{(m)}_j)_j\) converges uniformly on \(K_m\) to
\(G^{[m]}\), the subsequence \((G^{(m+1)}_j)_j\) also converges uniformly
on \(K_m\) to \(G^{[m]}\).  But it converges uniformly on \(K_{m+1}\), and
hence also on \(K_m\), to \(G^{[m+1]}\).  By uniqueness of uniform limits,
\[
        G^{[m+1]}=G^{[m]}
        \qquad\text{on }K_m.
\]
Therefore the functions \(G^{[m]}\) piece together to define a single
function
\[
        G^\circ:(-1,1)\to[0,M],
\]
by setting
\[
        G^\circ(x)=G^{[m]}(x)
        \quad\text{whenever }x\in K_m.
\]
This is well-defined because the limiting functions agree on overlaps. The fact that $G^\circ$ takes values in $[0,M]$ follows from the fact that each $G_k$ takes values in $[0,M]$, and thus all of the limiting functions also do.

Now choose the diagonal subsequence
\[
        \widetilde G_j:=G^{(j)}_j,
        \qquad j\ge2.
\]
Fix \(m\ge2\).  For every \(j\ge m\), the term \(\widetilde G_j\) belongs
to the subsequence \((G^{(m)}_r)_r\), because the subsequences are nested.
Since \((G^{(m)}_r)_r\) converges uniformly on \(K_m\) to \(G^{[m]}\),
the diagonal subsequence \((\widetilde G_j)_j\) also converges uniformly
on \(K_m\) to \(G^{[m]}=G^\circ|_{K_m}\).

Thus the diagonal subsequence converges uniformly on every \(K_m\).  Since
every compact subset of \((-1,1)\) is contained in some \(K_m\), this means
that
\[
        \widetilde G_j\to G^\circ
\]
locally uniformly on \((-1,1)\).\footnote{``Locally uniform'' convergence on $(-1,1)$ means exactly that $G_j\to G^\circ$ uniformly on every compact subset of $(-1,1)$.}

We now verify that the limit \(G^\circ\) is concave. Indeed, for \(x,y\in(-1,1)\) and
\(\theta\in[0,1]\), the concavity inequality for \(G_k\) gives
\[
        G_k(\theta x+(1-\theta)y)
        \ge
        \theta G_k(x)+(1-\theta)G_k(y).
\]
Letting \(k\to\infty\) gives the same inequality for \(G^\circ\).

It remains to define the endpoints.  Since \(G^\circ\) is finite and
concave on \((-1,1)\), it has one-sided limits at the endpoints, possibly
in \([-\infty,\infty]\). This is a standard fact that follows from the observation the fact that the ``secant slopes'' are monotone non-increasing (and thus have a limit). But \(0\le G^\circ\le M\), so these one-sided
limits are finite and belong to \([0,M]\).  Define
\[
        G_*(-1):=\lim_{x\downarrow -1}G^\circ(x),
        \qquad
        G_*(1):=\lim_{x\uparrow 1}G^\circ(x),
\]
and set \(G_*=G^\circ\) on \((-1,1)\).

Finally, \(G_*\) is concave on the closed interval \([-1,1]\).  This
follows by applying the concavity inequality for \(G^\circ\) at interior
points and then letting the relevant interior points tend to the endpoints.
For example, if one of \(x,y\) is an endpoint, approximate it from inside
\((-1,1)\), use concavity of \(G^\circ\), and pass to the one-sided limit.
\end{proof}




We now have the necessary ingredients to complete the proof of Theorem~\ref{thm:scalar-certificate}. Recall that we have been able to prove that, for any $q>q_*$, there exists an $H: [-1,1]\to[0,\infty)$ that is concave, bounded, and satisfies
\begin{equation}
\label{eq:32}
        H(t)H(-t)\ge a(t)^2
        \qquad\forall t\in[-1,1],
\end{equation}
and
\begin{equation}
\label{eq:31}
        H(v)+d(u,v)+H(-u)\le q
        \qquad\forall u,v\in[-1,1].
\end{equation}
All that is left is to improve the upper bound in \eqref{eq:31} to $q_*$ for an $H$ that is still concave.

\begin{proof}[Proof of Theorem~\ref{thm:scalar-certificate}]
Choose a decreasing sequence
\[
        q_k\downarrow q_*,
        \qquad q_k>q_*.
\]
For each \(k\), Lemma~\ref{lem:product-sat} gives a function $h_{0,k}: [-1,1]\to[0,\infty)$ satisfying the hypothesis of Lemma~\ref{lem:concavification}, and
the latter gives that its least concave majorant $H_k: [-1,1]\to[0,\infty)$ is concave, bounded, nonnegative, and satisfies
\begin{equation}
\label{eq:Hk-product}
        H_k(t)H_k(-t)\ge a(t)^2
        \qquad\forall t\in[-1,1]\,,
\end{equation}
and
\begin{equation}
\label{eq:Hk-cert}
        H_k(v)+d(u,v)+H_k(-u)\le q_k
        \qquad\forall u,v\in[-1,1]\,.
\end{equation}

We start by obtaining a bound on $H_k$ that is uniform across all $k$. In \eqref{eq:Hk-cert}, set
\[
        u=-v.
\]
Then
\[
        H_k(v)+d(-v,v)+H_k(v)\le q_k,
\]
so
\[
        2H_k(v)\le q_k-d(-v,v), \qquad\forall v\in[-1,1] \,.
\]
Now,
\[
        d(-v,v)
        =
        (-v)v+\frac{-v-v}{2}-1
        =
        -v^2-v-1.
\]
Therefore,
\[
        2H_k(v)
        \le
        q_k+v^2+v+1, \qquad\forall v\in[-1,1] \,.
\]
Since \(v\in[-1,1]\),
\[
        v^2+v+1\le 3.
\]
Also $q_k\le q_1$ by construction. So, we have
\[
        0\le H_k(v)\le \frac{q_1+3}{2}
        \qquad\forall v\in[-1,1],\ \forall k\,.
\]
Now, letting
\[
        M:=\frac{q_1+3}{2}.
\]
we have that the sequence $H_k$ is uniformly bounded in $[0,M]$, as desired.

By Lemma~\ref{lem:concave-compactness}, after passing to a subsequence,
there is a concave function
\[
        H^\circ:(-1,1)\to[0,M]
\]
such that
\[
        H_k\to H^\circ
\]
locally uniformly on \((-1,1)\).  Define \(H_*\) on the closed interval
\([-1,1]\) by
\[
        H_*(t)=H^\circ(t)
        \qquad (t\in(-1,1)),
\]
and by the one-sided endpoint limits
\[
        H_*(-1)=\lim_{t\downarrow -1}H^\circ(t),
        \qquad
        H_*(1)=\lim_{t\uparrow 1}H^\circ(t).
\]
By Lemma~\ref{lem:concave-compactness}, these endpoint limits are well-defined, and belong to $[0,M]$. Moreover, $H_*$ is concave on
\([-1,1]\), and
\[
        0\le H_*(t)\le M
        \qquad\forall t\in[-1,1].
\]
We now pass the two inequalities \eqref{eq:Hk-product} and \eqref{eq:Hk-cert} to the limit.

We start with \eqref{eq:Hk-product}. Let \(t\in(-1,1)\).  Since
\(H_k\to H_*\) locally uniformly on \((-1,1)\), we have
\[
        H_k(t)\to H_*(t),
        \qquad
        H_k(-t)\to H_*(-t)\,.
\]
Taking \(k\to\infty\) in
\[
        H_k(t)H_k(-t)\ge a(t)^2
\]
gives
\[
        H_*(t)H_*(-t)\ge a(t)^2
        \qquad(t\in(-1,1)).
\]
At the endpoints \(t=\pm1\), the right-hand side is zero:
\[
        a(\pm1)^2=0.
\]
Since \(H_*\ge0\), we also have
\[
        H_*(\pm1)H_*(\mp1)\ge0=a(\pm1)^2.
\]
Thus 
\[
        H_*(t)H_*(-t)\ge a(t)^2
        \qquad \forall t\in[-1,1].
\]

We now consider \eqref{eq:Hk-cert}. First let
\[
        u,v\in(-1,1).
\]
Then also \(-u\in(-1,1)\), and local uniform convergence gives
\[
        H_k(v)\to H_*(v),
        \qquad
        H_k(-u)\to H_*(-u).
\]
Since \(q_k\to q_*\), taking \(k\to\infty\) in
\[
        H_k(v)+d(u,v)+H_k(-u)\le q_k
\]
gives
\[
        H_*(v)+d(u,v)+H_*(-u)\le q_*.
\]
It remains to extend the inequality to boundary points.  Let
\[
        u,v\in[-1,1].
\]
Choose sequences \(u_j,v_j\in(-1,1)\) such that
\[
        u_j\to u,
        \qquad
        v_j\to v.
\]
If \(u\in(-1,1)\), we may take \(u_j=u\) for all \(j\); if \(u=\pm1\),
we choose \(u_j\) approaching \(u\) from the interior.  Similarly, if
\(v\in(-1,1)\), we may take \(v_j=v\) for all \(j\); if \(v=\pm1\), we
choose \(v_j\) approaching \(v\) from the interior. For each \(j\), the already-proved interior inequality gives
\[
        H_*(v_j)+d(u_j,v_j)+H_*(-u_j)\le q_*.
\]
Now let \(j\to\infty\).  The function \(d\) is continuous, so
\[
        d(u_j,v_j)\to d(u,v).
\]
Also, by the definition of the endpoint values of \(H_*\) as one-sided
limits, we have
\[
        H_*(v_j)\to H_*(v),
        \qquad
        H_*(-u_j)\to H_*(-u).
\]
Therefore
\[
        H_*(v)+d(u,v)+H_*(-u)\le q_* \qquad \forall u,v \in [-1,1] \,.
\]
$H_*$ is the desired certificate, and completes the proof of Theorem~\ref{thm:scalar-certificate}. In turn, this completes the proof of Theorem~\ref{thm:universal-upper}.
\end{proof}

\section{The structure of optimal strategies}
\label{sec:structure-optimal}

In this section, we analyze the structure of finite-dimensional strategies that attain the upper bound $q_*$. We show that any such strategy essentially decomposes into Jacobi path (and cycle) components, and that the value corresponding to these components (up to some small edge cases) can be understood in terms of the Jacobi matrices introduced in Definition~\ref{def:signed-endpoint-jacobi-path} and analyzed in Section~\ref{sec:paths-do-not-attain-q*}. Thanks to this decomposition, we will eventually arrive at a contradiction, based on the fact that such components can contribute only strictly less than $q_*$.

\subsection{Vanishing of every slack term}
\label{sec:zero-slack}
We start by recalling, for convenience, some of the notation we introduced in Section~\ref{sec:upper-bound-subsec} to describe an arbitrary finite-dimensional strategy, as well as some of the structure we derived.

Let the strategy be specified by the state
\[
        \ket{\psi}\in \mathcal H_A\otimes\mathcal H_B
\]
and projectors $A_i, B_i$ for $i \in \{1,2,3\}$ on $\mathcal H_A$ and $\mathcal H_B$. Define
\[
        Z_A=A_1+A_2-I,\qquad X_A=A_2-A_1,
\]
\[
        Z_B=B_1+B_2-I,\qquad X_B=B_2-B_1,
\]
and
\[
        W_A=2A_3-I,\qquad W_B=2B_3-I.
\]
Then, all of these operators are Hermitian, and one can verify directly that
\[
        Z_AX_A+X_AZ_A=0,\qquad Z_A^2+X_A^2=I,
\]
and similarly for \(Z_B,X_B\). The above identities imply the following useful structure by Lemma~\ref{lem:structure}. There is
an orthonormal basis $\{e_\alpha\}$ for Alice such that $Z_A$ and $X_A$ are simultaneously block-diagonal with $2\times 2$ or $1\times 1$ blocks. Moreover, for each pair $(e_{\alpha}, e_{\bar{\alpha}})$ corresponding to a $2$-dimensional block, there is a $z_\alpha\in (-1,1)\setminus\{0\}$ such that
\[
        Z_A e_\alpha=z_\alpha e_\alpha,
        \qquad
Z_A e_{\bar\alpha}=-z_\alpha e_{\bar\alpha},
\]
and
\[
        X_A e_\alpha=\sqrt{1-z_\alpha^2}\,e_{\bar\alpha}.
\]
From the above, it also follows that $z_{\bar{\alpha}} = - z_{\alpha}$. And for any $e_{\alpha}$ corresponding to a $1\times 1$ block, either:
\begin{itemize}
\item there exists $z_{\alpha}\in \{\pm 1\}$ such that $Z e_{\alpha} = z_{\alpha} e_{\alpha}$ and $X e_\alpha=0$, or
\item $Z e_{\alpha} = 0$ and $Xe_{\alpha} = \varepsilon_\alpha e_{\alpha}$ for some $\varepsilon_\alpha \in \{\pm 1\}$.
\end{itemize}

Similarly, Bob has an analogous orthonormal
basis $\{f_\beta\}$ and corresponding values $w_\beta$. Now, we can write the state $\ket{\psi}$ in these two bases:
\[
        \ket{\psi}=\sum_{\alpha,\beta}\psi_{\alpha\beta}
        e_\alpha\otimes f_\beta.
\]
for some coefficients $\psi_{\alpha \beta}$. Then, define 
\[
        p_{\alpha\beta}=|\psi_{\alpha\beta}|^2,
        \qquad
        r_\alpha=\sum_\beta p_{\alpha\beta},
        \qquad
        \ell_\beta=\sum_\alpha p_{\alpha\beta}.
\]

\begin{lemma}[Zero slack]
\label{lem:zero-slack}
Let $H:[-1,1]\to[0,\infty)$ be the concave certificate function guaranteed by Theorem~\ref{thm:scalar-certificate}. If a finite-dimensional strategy attains value $q_*$, then:
\begin{enumerate}[label=\textup{(\alph*)}]
\item For any pair $(\alpha, \beta)$ with $p_{\alpha\beta}>0$, 
\begin{equation}
\label{eq:pointwise-active}
        H(w_\beta)+d(z_\alpha,w_\beta)+H(-z_\alpha)=q_*;
\end{equation}
\item For any two-dimensional Alice block $\{\alpha, \bar{\alpha}\}$, 
\begin{equation}
\label{eq:chain}
2a(z_{\alpha})\sqrt{r_\alpha r_{\bar\alpha}}
        =2\sqrt{H(z_{\alpha})H(-z_{\alpha})r_\alpha r_{\bar\alpha}}  =
        H(-z_{\alpha})r_\alpha+H(z_{\alpha})r_{\bar\alpha}.
\end{equation}
As a consequence, if $r_{\alpha}, r_{\bar{\alpha}}>0$, we have
\begin{equation}
\label{eq:local-product-active}
        H(z_{\alpha})H(-z_{\alpha})=a(z_{\alpha})^2.
\end{equation}
Moreover, since $|z_\alpha| <1$, we have
$$r_{\alpha}>0 \iff r_{\bar\alpha}>0 \,.$$
The analogous statements hold on Bob's side.
\item For any one-dimensional Alice block $\{\alpha\}$ with $z_{\alpha} = 0$,
\begin{equation}
\label{eq:ralpha}
r_\alpha = 2H(0)r_\alpha \,.
\end{equation}
As a consequence, if $r_{\alpha}>0$, then $H(0) = \frac12$. 

Similarly, on Bob's side.
\end{enumerate}
\end{lemma}

\begin{proof}
The proof of Theorem~\ref{thm:universal-upper} tells us that, in order for the strategy to attain $q_*$, all summands in the final expression \eqref{eq:35} with $p_{\alpha\beta}>0$ must equal $q_*$. This gives precisely \eqref{eq:pointwise-active}. Moreover, the chain of inequalities in \eqref{eq:34} must be a chain of equalities, which gives precisely~\eqref{eq:chain}. Consequence \eqref{eq:local-product-active} is immediate. 

Now, if $|z_\alpha| < 1$, then $a(z_\alpha) = \frac12 \sqrt{1-z_\alpha^2}>0$. Then, the product inequality gives 
$$H(z_\alpha) H(-z_\alpha)\geq a(z_\alpha)^2 >0 \,, $$
which implies, since $H$ is non-negative, $H(z_\alpha)>0$ and $H(-z_\alpha)>0$. Then, suppose $r_\alpha >0$. If $r_{\bar\alpha} = 0$, then the LHS of \eqref{eq:chain} is zero, while the RHS is $H(-z_\alpha)r_\alpha >0$, which is a contradiction.  The converse direction is analogous.

Finally, for a one-dimensional Alice block $\{\alpha\}$, the proof of Theorem~\ref{thm:universal-upper} tells us that the inquality in \eqref{eq:211} must be an equality, which gives precisely \eqref{eq:ralpha}. The consequence when $r_\alpha >0$ is immediate.

The statements on Bob's side are proven analogously.

\end{proof}

\subsection{Consequences of concavity}

We first record two elementary consequences of concavity.  Let \(I\subset\mathbb R\)
be an interval and let \(G:I\to\mathbb R\) be a finite concave function.  For an
interior point \(x\in I^\circ\), the left and right derivatives
\[
        D_-G(x)
        :=
        \lim_{s\uparrow x}\frac{G(x)-G(s)}{x-s},
        \qquad
        D_+G(x)
        :=
        \lim_{s\downarrow x}\frac{G(s)-G(x)}{s-x}
\]
exist as finite real numbers and satisfy
\[
        D_-G(x)\ge D_+G(x).
\]
Indeed, for \(r<s<t\), concavity gives the monotonicity of secant slopes
\[
        \frac{G(s)-G(r)}{s-r}
        \ge
        \frac{G(t)-G(r)}{t-r}
        \ge
        \frac{G(t)-G(s)}{t-s}.
\]
This implies that the left and right secant slopes have one-sided limits,
and that the left derivative dominates the right derivative. We can see the latter as follows.

Let \(G:I\to\mathbb R\) be a finite concave function on an interval
\(I\), and let \(x\in I^\circ\).  For \(p<q\), write
\[
        S(p,q):=\frac{G(q)-G(p)}{q-p}
\]
for the secant slope of \(G\) on \([p,q]\).  Concavity implies that if
\(r<s<t\), then
\[
        S(r,s)\ge S(r,t)\ge S(s,t).
\]
Indeed, the inequality \(S(r,s)\ge S(r,t)\) is obtained by writing \(s\)
as a convex combination of \(r\) and \(t\) and applying concavity to
\(G(s)\); the inequality \(S(r,t)\ge S(s,t)\) follows from the same
concavity inequality after rearranging.

Now fix \(x\in I^\circ\).  For \(s<x\), define
\[
        L_x(s):=S(s,x)=\frac{G(x)-G(s)}{x-s}.
\]
If \(s_1<s_2<x\), then applying the preceding secant-slope inequality to
\(s_1<s_2<x\) gives
\[
        S(s_1,x)\ge S(s_2,x).
\]
Hence \(L_x(s)\) is monotone nonincreasing as \(s\uparrow x\).

Choose \(a,b\in I\) with \(a<x<b\).  For every \(s\in(a,x)\), the same
secant-slope inequalities give
\[
        S(x,b)\le S(s,x)\le S(a,x).
\]
Thus \(L_x\) is monotone and bounded on \((a,x)\).  Therefore the limit
\[
        D_-G(x):=\lim_{s\uparrow x}L_x(s)
        =
        \lim_{s\uparrow x}\frac{G(x)-G(s)}{x-s}
\]
exists as a finite real number.

Similarly, for \(t>x\), define
\[
        R_x(t):=S(x,t)=\frac{G(t)-G(x)}{t-x}.
\]
If \(x<t_1<t_2\), then
\[
        S(x,t_1)\ge S(x,t_2),
\]
so \(R_x(t)\) is monotone nonincreasing as a function of \(t\).  Moreover,
for \(t\in(x,b)\),
\[
        S(x,b)\le S(x,t)\le S(a,x).
\]
Thus \(R_x\) is monotone and bounded near \(x\), and therefore
\[
        D_+G(x):=\lim_{t\downarrow x}R_x(t)
        =
        \lim_{t\downarrow x}\frac{G(t)-G(x)}{t-x}
\]
exists as a finite real number.

Finally, for every \(s<x<t\), concavity gives
\[
        S(s,x)\ge S(x,t).
\]
Taking \(s\uparrow x\) and \(t\downarrow x\), we obtain
\[
        D_-G(x)\ge D_+G(x).
\]
Thus the one-sided derivatives of a finite concave function exist at every
interior point, and the left derivative dominates the right derivative.

\vspace{5mm}
We call
\[
        \Delta_G(x):=D_-G(x)-D_+G(x)\ge0
\]
the derivative jump of \(G\) at \(x\).  Thus \(G\) is differentiable at
\(x\) if and only if \(\Delta_G(x)=0\).

For the scalar certificate \(H\), we say that \emph{local product equality}
holds at an interior point \(t\in(-1,1)\) if
\[
        H(t)H(-t)=a(t)^2.
\]
This is a local condition at the pair of points \(t\) and \(-t\); it does
not assert product equality for all labels.

\begin{lemma}[Local product equality implies differentiability]
\label{lem:differentiability}
Let \(H:[-1,1]\to[0,\infty)\) be concave and suppose
\[
        H(t)H(-t)\ge a(t)^2
        \qquad\forall t\in[-1,1].
\]
If \(t_0\in(-1,1)\) satisfies local product equality,
\[
        H(t_0)H(-t_0)=a(t_0)^2,
\]
then \(H\) is differentiable at both \(t_0\) and \(-t_0\).
\end{lemma}

\begin{proof}
Since \(t_0\in(-1,1)\), we have
\[
        a(t_0)=\frac12\sqrt{1-t_0^2}>0.
\]
The equality
\[
        H(t_0)H(-t_0)=a(t_0)^2>0
\]
therefore implies
\[
        H(t_0)>0,
        \qquad
        H(-t_0)>0.
\]

Define
\[
        P(t):=H(t)H(-t)-a(t)^2.
\]
By the product inequality for \(H\),
\[
        P(t)\ge0
        \qquad\forall t\in[-1,1].
\]
By local product equality at \(t_0\),
\[
        P(t_0)=0.
\]
Thus \(P\) has a minimum at \(t_0\).

We now compute the one-sided derivatives of \(P\) at \(t_0\).  Since \(H\)
is concave, the one-sided derivatives of \(H\) exist at \(t_0\) and at
\(-t_0\).  Because \(H\) is continuous on \((-1,1)\), the usual product
rule for one-sided derivatives applies.  Also
\[
        a(t)^2=\frac{1-t^2}{4},
\]
so
\[
        \frac{d}{dt}\bigl[-a(t)^2\bigr]
        =
        \frac{t}{2}.
\]

First compute the left derivative of \(P\).  The left derivative of
\(t\mapsto H(t)\) at \(t_0\) is \(D_-H(t_0)\).  The left derivative of
\(t\mapsto H(-t)\) at \(t_0\) is
\[
        -D_+H(-t_0),
\]
because as \(t\uparrow t_0\), the argument \(-t\) decreases to
\(-t_0\) from the right.  Therefore
\[
\begin{aligned}
        P'_-(t_0)
        &=
        D_-H(t_0)H(-t_0)
        -
        H(t_0)D_+H(-t_0)
        +
        \frac{t_0}{2}.
\end{aligned}
\]
Similarly, the right derivative of \(t\mapsto H(-t)\) at \(t_0\) is
\[
        -D_-H(-t_0),
\]
because as \(t\downarrow t_0\) from the right, the argument \(-t\)
approaches \(-t_0\) from the left.  Hence
\[
\begin{aligned}
        P'_+(t_0)
        &=
        D_+H(t_0)H(-t_0)
        -
        H(t_0)D_-H(-t_0)
        +
        \frac{t_0}{2}.
\end{aligned}
\]
Subtracting these two identities gives
\[
\begin{aligned}
        P'_-(t_0)-P'_+(t_0)
        &=
        \bigl(D_-H(t_0)-D_+H(t_0)\bigr)H(-t_0)\\
        &\qquad
        +
        H(t_0)\bigl(D_-H(-t_0)-D_+H(-t_0)\bigr)\\
        &=
        \Delta_H(t_0)H(-t_0)+H(t_0)\Delta_H(-t_0).
\end{aligned}
\]
Both terms on the right-hand side are nonnegative.

On the other hand, \(P\) has a minimum at \(t_0\).  Therefore its left
derivative is nonpositive and its right derivative is nonnegative:
\[
        P'_-(t_0)\le0\le P'_+(t_0).
\]
Consequently
\[
        P'_-(t_0)-P'_+(t_0)\le0.
\]
Combining this with the previous nonnegative expression, we obtain
\[
        \Delta_H(t_0)H(-t_0)+H(t_0)\Delta_H(-t_0)=0.
\]
Since
\[
        H(t_0)>0,
        \qquad
        H(-t_0)>0,
\]
and since both derivative jumps are nonnegative, we must have
\[
        \Delta_H(t_0)=0,
        \qquad
        \Delta_H(-t_0)=0.
\]
Thus \(H\) is differentiable at both \(t_0\) and \(-t_0\).
\end{proof}

We shall also use the following terminology.  An \emph{affine majorant} of
a function \(G\) on an interval is an affine function
\[
        L(x)=mx+b
\]
such that
\[
        G(x)\le L(x)
        \qquad\text{for all }x
\]
in the interval.  If equality holds at \(x_0\), we say that \(L\) touches
\(G\) at \(x_0\).

If \(G\) is concave and differentiable at \(x_0\), and if an affine
majorant \(L(x)=mx+b\) touches \(G\) at \(x_0\), then
\[
        m=G'(x_0).
\]
Indeed, the function \(L-G\) is nonnegative and has a minimum equal to
zero at \(x_0\).  Therefore
\[
        D_-(L-G)(x_0)\le0\le D_+(L-G)(x_0).
\]
Since \(G\) is differentiable at \(x_0\),
\[
        D_-(L-G)(x_0)=D_+(L-G)(x_0)=m-G'(x_0).
\]
Hence \(m-G'(x_0)=0\), as claimed.

\begin{lemma}[Single-valuedness]
\label{lem:single-valued}
Let $H:[-1,1]\to[0,\infty)$ be the concave certificate from Theorem~\ref{thm:scalar-certificate}, and define its equality set
\[
        E_H
        :=
        \{(u,v)\in[-1,1]^2:
        H(v)+d(u,v)+H(-u)=q_*\}.
\]
Suppose $(u,v)\in E_H$. Then, the following holds:
\begin{itemize}
\item If $u \in (-1,1)$ and $H(u)H(-u)=a(u)^2$, then 
\begin{equation}
\label{eq:v-from-u}
        v=H'(-u)-\frac12.
\end{equation}
In particular, $u$ uniquely determines $v$.
\item If $v \in (-1,1)$ and $H(v)H(-v)=a(v)^2$, then 
\begin{equation}
\label{eq:u-from-v}
       u=\frac12 -H'(v) \,.
\end{equation}
In particular, $v$ uniquely determines $u$.
\end{itemize}
\end{lemma}

\begin{proof}
We first prove \eqref{eq:u-from-v}. Suppose the product equality holds at $v \in (-1,1)$. Then, Lemma~\ref{lem:differentiability}
implies that \(H\) is differentiable at \(v\).  Fix \(u\), and define the affine
function
\[
        L_u(x):=q_*-d(u,x)-H(-u).
\]
The scalar certificate says
\[
        H(x)+d(u,x)+H(-u)\le q_*
        \qquad\forall x\in[-1,1].
\]
Equivalently,
\[
        H(x)\le L_u(x)
        \qquad\forall x\in[-1,1].
\]
Thus \(L_u\) is an affine majorant of \(H\).  Since \((u,v)\in E_H\),
we have
\[
        H(v)+d(u,v)+H(-u)=q_*,
\]
or equivalently
\[
        H(v)=L_u(v).
\]
Therefore \(L_u\) touches \(H\) at \(v\).  Since \(H\) is differentiable
at \(v\), the slope of \(L_u\) must equal \(H'(v)\).

Now
\[
        d(u,x)=ux+\frac{u-x}{2}-1,
\]
so
\[
        \frac{d}{dx}d(u,x)=u-\frac12.
\]
Hence
\[
        L_u'(x)
        =
        -\frac{d}{dx}d(u,x)
        =
        \frac12-u.
\]
Since \(L_u'(x)=H'(v)\), we obtain
\[
        H'(v)=\frac12-u,
\]
which is \eqref{eq:u-from-v}.

We now prove \eqref{eq:v-from-u}. Suppose the product equality holds at $u$ and $u \in (-1,1)$. Then, Lemma~\ref{lem:differentiability}
implies that \(H\) is differentiable at \(-u\). Fix \(v\), and define
\[
        G(x):=H(-x).
\]
Because \(H\) is concave, \(G\) is also concave.  Since \(H\) is
differentiable at \(-u\), the function \(G\) is differentiable at \(u\),
with
\[
        G'(u)=-H'(-u).
\]

Define the affine function
\[
        M_v(x):=q_*-d(x,v)-H(v).
\]
The scalar certificate says
\[
        H(v)+d(x,v)+H(-x)\le q_*
        \qquad\forall x\in[-1,1].
\]
Equivalently,
\[
        G(x)=H(-x)\le M_v(x)
        \qquad\forall x\in[-1,1].
\]
Thus \(M_v\) is an affine majorant of \(G\).  Since \((u,v)\in E_H\),
we have
\[
        H(v)+d(u,v)+H(-u)=q_*,
\]
or
\[
        G(u)=M_v(u).
\]
Therefore \(M_v\) touches \(G\) at \(u\).  Since \(G\) is differentiable
at \(u\), the slope of \(M_v\) must equal \(G'(u)\).

Now
\[
        d(x,v)=xv+\frac{x-v}{2}-1,
\]
so
\[
        \frac{d}{dx}d(x,v)=v+\frac12.
\]
Hence
\[
        M_v'(x)=-(v+\tfrac12).
\]
Therefore
\[
        -H'(-u)=G'(u)=M_v'(u)=-(v+\tfrac12).
\]
Rearranging gives
\[
        v=H'(-u)-\frac12,
\]
which is \eqref{eq:v-from-u}.
\end{proof}

\subsection{The graph of a strategy attaining \(q_*\)}
In this subsection, we introduce a graph $\Gamma$ associated to a hypothetical strategy attaining $q_*$, whose vertices are pairs of labels corresponding to Alice and Bob blocks. We will then use the vanishing of slack terms and the uniqueness condition from the previous subsection to show that $\Gamma$ has degree two, and its connected components are paths and cycles. Finally, we will associate a symmetric matrix with $\Gamma$, and show that each path or cycle
component actually essentially has a form very similar to that of a Jacobi matrix, with top eigenvalue \(q_*\). 

Let $H$ be the concave certificate function from Theorem~\ref{thm:scalar-certificate}, and recall that we defined the \emph{equality set}
\[
        E_H
        :=
        \{(u,v)\in[-1,1]^2:
        H(v)+d(u,v)+H(-u)=q_*\}\,.
\]

A key object that will help us understand the structure of strategies attaining $q_*$ is the following weighted graph.
\begin{definition}[Graph $\Gamma$ of a strategy attaining $q_*$]
\label{def:graph}
For a strategy attaining $q_*$, we define its corresponding graph $\Gamma$ as follows (where we use the same notation to describe the strategy as introduced at the start of Section~\ref{sec:zero-slack}). 

A vertex is a pair
\[
        \nu=(u,v)\in E_H
\]
with strictly positive total ``mass''
\[
        m_\nu=\sum_{z_\alpha=u,\,w_\beta=v}p_{\alpha\beta}>0.
\]
We include an edge between vertices $(u,v)$ and $(-u,v')$ if $0<|u|<1$. We refer to the latter as an ``Alice edge'', and let its weight be $a(u)$. Similarly, we include an edge between $(u,v)$ and $(u',-v)$ if $0<|v|<1$. We refer to the latter as a ``Bob edge'', and let its weight be $a(v)$.    
\end{definition}
Now, for the rest of this section, assume (for the sake of contradiction) a finite-dimensional strategy exists that attains $q_*$, and let $\Gamma$ be its corresponding weighted graph as in Definition~\ref{def:graph}. Then, $\Gamma$ satisfies the following structural property. 
\begin{lemma}
\label{lem:path-or-cycle}
Each vertex of $\Gamma$ has at most one Alice edge and at most one Bob edge incident to it. As a consequence, every finite connected component of $\Gamma$ is either a path or a cycle, with alternating Alice and Bob edges.    
\end{lemma}
\begin{proof}
Let $(u,v) \in E_H$ be a vertex. By definition, $m_{\nu} >0$. This implies that there exists a pair $\alpha, \beta$ such that $u = z_{\alpha}$, $v = w_{\beta}$, and $p_{\alpha, \beta}>0$.

First, let us consider Alice edges, which can exist only if $|u|<1$ and $u\neq 0$. Then, it must be that $\alpha$ is part of 2-dimensional Alice block (since the $1$-dimensional blocks have $z_{\alpha} \in \{-1,0,1\}$). Moreover, $r_{\alpha}>0$ since $p_{\alpha \beta}>0$. Now, suppose there is an Alice edge $e$ from $(u,v)$ to $(-u,v')$. Then, similarly, there exists a pair $\alpha', \beta'$ such that $-u = z_{\alpha'}$, $v' = w_{\beta'}$, and $p_{\alpha' \beta'} > 0$, which implies $r_{\alpha'} >0$. Thus, we can invoke Lemma~\ref{lem:zero-slack} to deduce that
$$H(z_{\alpha})H(-z_{\alpha}) = a(z_{\alpha})^2 \,.$$
Combining this with the fact that $(u,v) \in E_H$, lets us apply Lemma~\ref{lem:single-valued} to obtain that $v=H'(-u)-\frac12$. Thus, $u$ uniquely determines $v$, and similarly $-u$ uniquely determines $v'$. Thus, $e$ is the unique Alice edge at $(u,v)$. Analogously, we can show that $(u,v)$ has at most one Bob edge. 

The final consequence of the lemma is immediate: in a finite graph where vertices have degree at most 2, any connected components is either a paths or a cycle. Since every vertex has at most one Alice edge and Bob edge, the edges along any such path or cycle must alternate.

\end{proof}

\begin{lemma}
\label{cor:no-dangling-scalar}
Let \(\nu=(u,v)\in\Gamma\).  If \(0<|u|<1\), then \(\nu\) has a
unique Alice-edge neighbour \((-u,v')\) in \(\Gamma\).  Analogously, if
\(0<|v|<1\), then \(\nu\) has a unique Bob-edge neighbour.
\end{lemma}

\begin{proof}
Since \(m_\nu>0\), there are indices \(\alpha,\beta\) such that
\[
        z_\alpha=u,\qquad w_\beta=v,\qquad p_{\alpha\beta}>0.
\]
Thus \(r_\alpha>0\).  Since \(0<|u|<1\), the index \(\alpha\)
belongs to a two-dimensional Alice block
\(\{\alpha,\bar\alpha\}\).  Lemma~\ref{lem:zero-slack}, part~\textup{(b)},
gives
\[
        r_{\bar\alpha}>0.
\]
Hence \(p_{\bar\alpha\beta'}>0\) for some \(\beta'\), and
\((-u,w_{\beta'})\) is a vertex of \(\Gamma\).  Since \(0<|u|<1\), the definition of $\Gamma$ (Definition~\ref{def:graph}) includes an Alice edge joining
these two vertices.

Since \(r_\alpha>0\) and \(r_{\bar\alpha}>0\), the product-equality
conclusion of Lemma~\ref{lem:zero-slack}, part~\textup{(b)}, gives
\[
        H(u)H(-u)=a(u)^2.
\]
Lemma~\ref{lem:single-valued} therefore implies that there is
at most one Bob label $v'$ that can be paired with the Alice label
\(-u\).  Hence the Alice-edge neighbour of \((u,v)\) is unique. The proof
for Bob-edge neighbours is analogous.
\end{proof}

\begin{lemma}
\label{lem:95}
Let $M_{\Gamma}$ be the real symmetric matrix with rows and columns labeled by vertices of $\Gamma$ defined as follows. The diagonal entries are 
$$ (M_{\Gamma})_{\nu \nu} = \begin{cases} d(u_{\nu}, v_{\nu}) \quad \quad \quad \,\,\,\,\textnormal{if } u_{\nu}, v_{\nu} \neq 0 \\
d(u_{\nu}, v_{\nu}) + \frac12 \quad \quad \textnormal{if } u_{\nu} = 0 \textnormal{ or } v_{\nu} = 0, \textnormal{ but not both} \\
d(u_{\nu}, v_{\nu}) + 1 \quad \quad \textnormal{ if } u_{\nu} = v_{\nu}=0
\end{cases}
$$
For two distinct vertices \(\nu,\mu\), the off-diagonal entry
\((M_\Gamma)_{\nu\mu}\) is the sum of the weights of all Alice or Bob edges
joining them (if the same two vertices are joined by one Alice
edge and one Bob edge, both weights are included). In other words,
\(M_\Gamma\) is the weighted adjacency matrix of \(\Gamma\), with
\(d(u_\nu,v_\nu)\), \(d(u_\nu,v_\nu)+\frac12\), or
\(d(u_\nu,v_\nu)+1\) on the diagonal.
Then,
$$M_\Gamma\le q_*I\,.$$ 
Moreover, the vector $\xi$ with entries $\xi_{\nu} = \sqrt{m_{\nu}}$ is an eigenvector of $M_{\Gamma}$ with eigenvalue~$q_*$.
\end{lemma}

\begin{proof}
We first prove $M_\Gamma\le q_*I$.

Let $V(\Gamma)$, $E_A(\Gamma)$, $E_B(\Gamma)$ be respectively the vertex set, Alice edge set, and Bob edge set of $\Gamma$. Let \(y=(y_\nu)_{\nu\in V(\Gamma)}\in\mathbb R^{V(\Gamma)}\),
and let \(x_\nu=|y_\nu|\).  Since all off-diagonal edge weights are
nonnegative,
\[
        \langle y,M_\Gamma y\rangle
        \le
        \langle x,M_\Gamma x\rangle.
\]
So it suffices to consider \(x_\nu\ge0\). 

Now, for an Alice edge
\[
        e=\{\nu,\mu\},
        \qquad
        \nu=(u,v),
        \qquad
        \mu=(-u,v'),
\]
its contribution to $\langle x,M_\Gamma x\rangle$ is $2a(u)x_\nu x_\mu$. Notice that the product inequality
\[
        H(u)H(-u)\ge a(u)^2
\]
(where $H:[-1,1] \rightarrow [0,\infty)$ is the certificate from Theorem~\ref{thm:scalar-certificate})
gives
\[
        2a(u)x_\nu x_\mu
        \le 2\sqrt{H(u)H(-u)} x_{\nu} x_{\mu}\leq
        H(-u)x_\nu^2+H(u)x_\mu^2.
\]
where the last inequality is by AM-GM.
Similarly, for a Bob edge
\[
        e=\{\nu,\mu\},
        \qquad
        \nu=(u,v),
        \qquad
        \mu=(u',-v),
\]
its contribution to
\(\langle x,M_\Gamma x\rangle\) is \(2a(v)x_\nu x_\mu\), and we have
\[
        2a(v)x_\nu x_\mu
        \le
        H(v)x_\nu^2+H(-v)x_\mu^2.
\]
Then, we have
\[
\begin{aligned}
\langle x,M_\Gamma x\rangle
&\le
\sum_{\nu\in V(\Gamma)}(M_\Gamma)_{\nu\nu}x_\nu^2\\
&\quad+
\sum_{e=\{\nu,\mu\}\in E_A(\Gamma)}
\bigl[H(-u_\nu)x_\nu^2+H(u_\nu)x_\mu^2\bigr]\\
&\quad+
\sum_{e=\{\nu,\mu\}\in E_B(\Gamma)}
\bigl[H(v_\nu)x_\nu^2+H(-v_\nu)x_\mu^2\bigr].
\end{aligned}
\]
Now, we can rewrite the latter in terms of sums over vertices, rather than edges. The key observation is that each vertex $\nu = (u,v)$ is involved in at most one Alice edge and one Bob edge (or less if one of $u$ or $v$ are $0$). Thus, the contribution from terms involving $x_{\nu}^2$ is at most
$$
\begin{cases}
d(u,v)x_{\nu}^2 + H(-u)x_{\nu}^2 + H(v)x_{\nu}^2 \quad\quad \textnormal{if } u, v \neq 0 \\
(d(u,v) + \frac12) x_{\nu}^2 + H(v)x_{\nu}^2 \quad\quad \textnormal{if } u = 0  \,\land \,v \neq 0 \\
(d(u,v) + \frac12) x_{\nu}^2 + H(-u)x_{\nu}^2 \quad\quad \textnormal{if } u \neq 0  \,\land \,v = 0 \\
(d(u,v) + 1)x_{\nu}^2  \quad\quad \textnormal{if } u = 0  \,\land \,v = 0 \,,
\end{cases}
$$
where here we used the fact that $H$ is non-negative, and that there is no Alice edge adjacent to a vertex that has $u=0$, and similarly no Bob edge with $v=0$.
Crucially, since $H(0)^2 \geq a(0)^2 = \frac14$, and so $H(0) \geq \frac12$, we can upper bound the expression in all four cases above simply by 
$$ d(u,v)x_{\nu}^2 + H(-u)x_{\nu}^2 + H(v)x_{\nu}^2 \,.$$
Thus, putting things together, we have
\[
\begin{aligned}
        \langle x,M_\Gamma x\rangle
        &\le
        \sum_{\nu\in V(\Gamma)}
        \bigl[
        d(u_\nu,v_\nu)+H(-u_\nu)+H(v_\nu)
        \bigr]x_\nu^2. \,,
\end{aligned}
\]
By the properties of the certificate,
\[
        H(v)+d(u,v)+H(-u)\le q_*
        \qquad\forall u,v\in[-1,1].
\]
Applying this to every term in the above sum, we get
\[
        \langle x,M_\Gamma x\rangle
        \le
        q_*\sum_{\nu}x_\nu^2
        =
        q_*\sum_{\nu}y_\nu^2.
\]
Thus,
\[
        \langle y,M_\Gamma y\rangle
        \le
        q_*\|y\|^2
\]
for every real vector \(y\), proving
\[
        M_\Gamma\le q_*I.
\]

We now turn to proving that, if a strategy attains $q_*$, then the vector $\xi$ with entries $\xi_{\nu} = \sqrt{m_{\nu}}$ is an eigenvector of $M_\Gamma$ with eigenvalue $q_*$.

Let \(V(\Gamma)\) denote the set of vertices of $\Gamma$. Recall that a vertex is a pair $\nu=(u_\nu,v_\nu)\in E_H$, and we defined its ``mass''
$$m_\nu=\sum_{z_\alpha=u_\nu,\,w_\beta=v_\nu}p_{\alpha\beta}
$$
and required $m_{\nu}>0$.

The first observation is that, since the strategy attains
\(q_*\), 
\begin{equation}    \label{eq:sum_m_nu}\sum_{\nu\in V(\Gamma)}m_\nu=1.
\end{equation}
The latter follows from Lemma~\ref{lem:zero-slack} and the fact that $\sum_{\alpha,\beta} p_{\alpha \beta} = 1$: to see this, note that, for any $p_{\alpha \beta}>0$, Lemma~\ref{lem:zero-slack} guarantees that $\nu = (z_{\alpha}, w_{\beta}) \in E_H$; moreover, $m_{\nu} \geq p_{\alpha \beta} >0$, and so $\nu = (z_{\alpha}, w_{\beta})$ is a legitimate vertex, and thus $p_{\alpha \beta}$ is included in the summation above.

In the proof of the
universal upper bound (Theorem~\ref{thm:universal-upper}), the Bell value is decomposed as $\text{val} = D+T_A+T_B$,
where
\[
        D
        =
        \sum_{\alpha,\beta}
        p_{\alpha\beta}d(z_\alpha,w_\beta)
\]
is the diagonal contribution, and
\begin{equation}
\label{eq:T_A}
        T_A
        =
        \frac12
        \langle \psi,X_A\otimes W_B\,\psi\rangle,
        \qquad
        T_B
        =
        \frac12
        \langle \psi,W_A\otimes X_B\,\psi\rangle.
\end{equation}
Aggregating contributions corresponding to the same vertex, we can rewrite the diagonal contribution as
\[
        D
        =
        \sum_{\nu\in V(\Gamma)}
        d(u_\nu,v_\nu)m_\nu .
\]
Now consider the Alice term \(T_A\). In the proof of Theorem~\ref{thm:universal-upper}, we showed that $T_A$ decomposes into a sum over Alice blocks. As we have seen a few times by now, there are three types of Alice blocks to consider (according to Lemma~\ref{lem:structure}):

\begin{itemize}
\item First, if $b = \{\alpha,\bar\alpha\}$ is a two-dimensional Alice block with $z_\alpha=u, z_{\bar\alpha}=-u$, and $u \neq 0, |u|<1$, then, its contribution to $T_A$ is, from Equation~\eqref{eq:83} at most
\[
        2a(u)\sqrt{r_\alpha r_{\bar\alpha}},
\]
where recall that $r_\alpha=\sum_\beta p_{\alpha\beta}$, and $r_{\bar\alpha}=\sum_\beta p_{\bar\alpha\beta}$. 

The contribution to $T_A$ from two-dimensional blocks can then be rewritten as a sum over Alice edges, because of the following. First, we can ignore blocks where either $r_\alpha$ or $r_{\bar{\alpha}}$ is zero, since they contribute zero. On the other hand, since both $r_{\alpha}, r_{\bar\alpha}>0$, there exist $\beta$ and $\beta'$ such that $p_{\alpha\beta}, p_{\bar\alpha \beta'}>0$. Then, since the strategy attains $q_*$, we can invoke Lemma~\ref{lem:zero-slack} to deduce that slacks vanish: in particular, $a(u)^2 = H(u) H(-u)$, and there must exist vertices $\nu = (u, v)$ and $\mu = (-u, v')$ in $\Gamma$, where $v = w_\beta$ and $v' = w_{\beta'}$. Moreover, we can invoke the single-valuedness lemma, Lemma~\ref{lem:single-valued}, at $u$ to deduce that each contribution $2a(u)\sqrt{r_{\alpha} r_{\bar{\alpha}}}$ corresponds to exactly one Alice edge \(e=\{\nu,\mu\}\), where \(\nu=(u,v)\) and \(\mu=(-u,v')\), with weight $a_e := a(u)$ and $u = z_{\alpha}$. 

So, letting \(\mathcal B_e\) be the set of blocks corresponding to Alice edge $e$, and $E_A(\Gamma)$ the set of Alice edges, we have that the total contribution to $T_A$ from two-dimensional blocks is
\[
\begin{aligned}
        \sum_{e \in E_A(\Gamma)}\sum_{\{\alpha, \bar{\alpha}\} \in \mathcal B_e}
        2a_e\sqrt{r_\alpha r_{\bar\alpha}}
        &\le \sum_{e \in E_A(\Gamma)}
        2a_e
        \sqrt{
        \left(\sum_{\{\alpha, \bar{\alpha}\} \in \mathcal B_e}r_\alpha\right)
        \left(\sum_{\{\alpha, \bar{\alpha}\} \in \mathcal B_e}r_{\bar\alpha}\right)
        } \\
        &=
        \sum_{e \in E_A(\Gamma)}2a_e\sqrt{m_\nu m_\mu}.
\end{aligned}
\]
where the identities $m_{\nu} = \sum_{\{\alpha, \bar{\alpha}\} \in \mathcal B_e}r_\alpha$ and \(m_{\mu} =\sum_{\{\alpha, \bar{\alpha}\} \in \mathcal B_e}r_{\bar\alpha}\) also follow by Lemma~\ref{lem:single-valued}.
\item Second, if \(b=\{\alpha\}\) is a one-dimensional zero block with $z_\alpha=0, 
        X_Ae_\alpha=\varepsilon_\alpha e_\alpha$, and $
        \varepsilon_\alpha\in\{\pm1\}$, its contribution to $T_A$ is, from Equation~\eqref{eq:211}, at most
\[
        \frac12 r_\alpha.
\]
Since the strategy attains $q_*$, Lemma~\ref{lem:zero-slack} gives $\frac14 = a(0)^2 = H(0)^2$. By Lemma~\ref{lem:single-valued}, each one-dimensional zero block
\(\{\alpha\}\) is associated with a single vertex
\(\nu=(0,v_\nu)\).  For each such vertex $\nu$, recall that we defined
\[
        m_\nu= \sum_{\alpha:\,\{\alpha\}\text{ is associated with }\nu}
        r_\alpha.
\]
Since the block \(\{\alpha\}\) contributes at most
\(\frac12r_\alpha\) to \(T_A\), summing over all one-dimensional zero
blocks we have that
\[
        \frac12\sum_{\nu:\,u_\nu=0}m_\nu
\]
is an upper bound on their total contribution to \(T_A\).

\item Third, if \(b=\{\alpha\}\) is a one-dimensional block with $z_\alpha=\pm1, \,
       X_Ae_\alpha=0$, its contribution to $T_A$ is zero. 
\end{itemize}
Thus, we have
\begin{equation}
        T_A
        \leq 2\sum_{e=\{\nu,\mu\}\in E_A(\Gamma)}
        a_e\sqrt{m_\nu m_\mu} +
\frac12 \sum_{\nu \in V(\Gamma): u_{\nu} = 0} m_\nu 
\end{equation}
The same argument on Bob's side gives
\[
        T_B
        \le
        2\sum_{e=\{\nu,\mu\}\in E_B(\Gamma)}
        a_e\sqrt{m_\nu m_\mu} + \frac12 \sum_{\mu \in V(\Gamma) : v_{\mu} = 0} m_{\mu}
\]
Therefore, the value $D + T_A + T_B$ of the original strategy is bounded above by
\begin{align*}
&\sum_{\nu\in V(\Gamma)}
d(u_\nu,v_\nu)m_\nu
        +
        2\sum_{e=\{\nu,\mu\}\in E(\Gamma)}
        a_e\sqrt{m_\nu m_\mu}
        +
        \frac12\sum_{\nu:\,u_\nu=0}m_\nu
        +
        \frac12\sum_{\mu:\,v_\mu=0}m_\mu
\\
&=
\sum_{\nu\in V(\Gamma)}
(M_\Gamma)_{\nu\nu}m_\nu
+
2\sum_{\nu, \mu \in V(\Gamma): \nu \neq \mu}
(M_\Gamma)_{\nu\mu}\sqrt{m_\nu m_\mu}
\\
&=
\langle\xi,M_\Gamma\xi\rangle.
\end{align*}
where the first equality follows by the definition of $M_{\Gamma}$, given in Lemma~\ref{lem:95}.

Now, let $\xi$ be the vector whose entries, indexed by vertices, are
\[
        \xi_\nu=\sqrt{m_\nu}.
\]
By Lemma~\ref{lem:95}, $M_\Gamma\le q_*I$. Thus, we have
$$ \text{val} = D + T_A + T_B \leq \langle \xi, M_\Gamma \xi \rangle \leq q_* \|\xi \|^2 = q_*$$
where the last equality is by \eqref{eq:sum_m_nu}. Since the strategy we are considering attains $q_*$, the inequality must be equalities, which means $\xi$ is an eigenvector of $M_{\Gamma}$ with eigenvalue $q_*$.
\end{proof}

Lemma~\ref{lem:95} readily gives the following corollary.
\begin{corollary}
Decompose $\Gamma$ into its connected components: $\Gamma = \bigsqcup_j \Gamma_j$, where the $\Gamma_j$ are the connected components. Then, $$M_{\Gamma} = \bigoplus_j M_{\Gamma_j}\,,$$
where the $M_{\Gamma_j}$ are defined analogously as $M_{\Gamma}$ but with respect to $\Gamma_j$. Let $\xi$ be the vector with entries $\xi_\nu = \sqrt{m_\nu}$ for $\nu \in \Gamma$. Then, $\xi$ decomposes as:
\[
        \xi=\bigoplus_j \xi^{(j)}.
\]
where each $\xi^{(j)}$ is an eigenvector of $M_{\Gamma_j}$ with eigenvalue $q_*$.
\end{corollary}

\begin{proof}
The fact that $$M_{\Gamma} = \bigoplus_j M_{\Gamma_j} $$
follows directly from the fact that $M_{\Gamma}$ has the structure of an adjacency matrix of $\Gamma$. Then, the fact that each $\xi^{(j)}$ is an eigenvector of $M_{\Gamma_j}$ with eigenvalue $q_*$ follows immediately from Lemma~\ref{lem:95}.
\end{proof}

\begin{lemma}
\label{lem:component-jacobi}
Let \(\Gamma_0\) be a finite connected component of the graph
\(\Gamma\) from Definition~\ref{def:graph}, and let
\(M_{\Gamma_0}\) be the submatrix of \(M_\Gamma\) obtained by
restricting its rows and columns to the vertices of \(\Gamma_0\).  Define
\[
        \rho_0(t):=\frac12\mathbf 1_{\{t=0\}},
        \qquad
        \widehat d(u,v):=d(u,v)+\rho_0(u)+\rho_0(v).
\]
Let $m$ denote number of vertices of $\Gamma_0$. By Lemma~\ref{lem:path-or-cycle}, \(\Gamma_0\) is either a path or a cycle. Moreover, it has one of the following two forms:
\begin{enumerate}
\item \emph{Path form.}
Suppose $\Gamma_0$ is a cycle. Then, there exist ``labels'' 
\[
        c_0,c_m\in\{-1,0,1\},
        \qquad
        c_i\in(-1,1)\setminus\{0\}
        \quad(1\le i<m),
\]
such that, for every \(x=(x_1,\ldots,x_m)\in\mathbb R^m\),
\begin{equation}
\label{eq:102}
        \langle x,M_{\Gamma_0}x\rangle
        =
        \sum_{i=1}^m
        \widehat d(c_{i-1},c_i)x_i^2
        +
        2\sum_{i=1}^{m-1}
        a(c_i)x_ix_{i+1}.
\end{equation}
Thus the only possible difference between $M_{\Gamma_0}$ and an ordinary Jacobi matrix (as in Definition~\ref{def:signed-endpoint-jacobi-path}) is a \(+\frac12\) correction in the first diagonal entry when \(c_0=0\)
and in the last diagonal entry when \(c_m=0\).  When \(m=1\), both
corrections occur in the same entry.

\item \emph{Cycle form.}
Suppose $\Gamma_0$ is a cycle. Then, $m$ is even, and there exist ``labels''
\[
        c_1,\ldots,c_m\in(-1,1)\setminus\{0\},
        \qquad
        c_0:=c_m,
\]
such that, for every \(x=(x_1,\ldots,x_m)\in\mathbb R^m\),
\[
\begin{aligned}
        \langle x,M_{\Gamma_0}x\rangle
        &=
        \sum_{i=1}^m d(c_{i-1},c_i)x_i^2
        +
        2\sum_{i=1}^m a(c_i)x_ix_{i+1},
        \qquad x_{m+1}:=x_1.
\end{aligned}
\]
Here \(\widehat d=d\), since none of the cycle labels is zero.  This
formula also covers \(m=2\): the Alice and Bob edges are then ``parallel'',
and their weights add in the corresponding off-diagonal entry.
\end{enumerate}
In either form, for all \(1\le i\le m\),
\[
        (c_{i-1},c_i)\in E_H \,,
\]
where for a cycle the indices are
read modulo \(m\). Moreover, for every \(1\le i<m\) in the path case and for every \(i\) in the cycle
case, we have
\[
        H(c_i)H(-c_i)=a(c_i)^2 \,.
\]
In the path case, the same equality also holds at \(c_0\) or
\(c_m\) whenever that endpoint label is zero.
\end{lemma}

\begin{proof}
By Lemma~\ref{lem:path-or-cycle}, $\Gamma_0$ is either a path (possibly a single isolated vertex) or a cycle, with alternating Alice and Bob edges.

By Lemma~\ref{cor:no-dangling-scalar}, if
\(\nu=(u,v)\) satisfies \(0<|u|<1\), then \(\nu\) has a (unique) Alice
neighbour.  Consequently, if no Alice edge is incident to \(\nu\), then
\[
        u\in\{-1,0,1\}.
\]
Similarly, if no Bob edge is incident to \(\nu\), then $v\in\{-1,0,1\}$.
We shall also use the identities
\[
        d(-v,-u)=d(u,v),
        \qquad
        \rho_0(-t)=\rho_0(t).
\]
They imply $\widehat d(-v,-u)=\widehat d(u,v)$, and also that that the equality set $E_H$ is invariant under the map $\tau(u,v) =(-v,-u)$:
\[
        (u,v)\in E_H
        \Longleftrightarrow
        \tau(u,v)\in E_H.
\]

\noindent \textbf{The path case.} Suppose first that \(\Gamma_0\) is a path.  If \(m=1\), write its only
vertex as \(\nu_1=(u_1,v_1)\).  Both edge types (Alice and Bob) are missing, so, by the earlier observation, it must be
\[
        u_1,v_1\in\{-1,0,1\}.
\]
Let \(c_0=u_1\) and \(c_1=v_1\). Then, by the definition of the diagonal entries of \(M_\Gamma\) in
Lemma~\ref{lem:95}, we have
\[
        M_{\Gamma_0}
        =
        [\widehat d(c_0,c_1)] \,,
\]
which immediately implies identity~\eqref{eq:102}.
Since $(c_0,c_1)=(u_1,v_1)$ is a vertex of $\Gamma$, we also have $(c_0,c_1)\in E_H$, by definition of the graph.
Moreover, if either \(c_0\) or \(c_1\) is zero, then
Lemma~\ref{lem:zero-slack} gives
\[
        H(0)=\frac12,
        \qquad
        H(0)^2=a(0)^2.
\]
This proves all the desired statements when \(m=1\).

Now assume \(m\ge2\). Order the path vertices as
\[
        \nu_1,\ldots,\nu_m,
        \qquad
        \nu_i=(u_i,v_i),
\]
and let the $i$-th edge \(e_i\) join \(\nu_i\) to \(\nu_{i+1}\).
If \(e_1\) is a Bob edge, set
\[
        c_0:=u_1,\qquad c_1:=v_1.
\]
We define labels \(c_2,\ldots,c_m\) recursively. In words, recall that across a Bob edge the second coordinate changes sign, and across
an Alice edge the first coordinate changes sign; the other coordinate
at the new vertex defines the next \(c_i\), up to a sign. Formally, suppose first that \(i\) is odd.  Then \(e_i\) is a
Bob edge, so $v_{i+1}=-v_i=-c_i$.
We define
\[
        c_{i+1}:=-u_{i+1},
\]
which gives
\[
        (u_{i+1},v_{i+1})=(-c_{i+1},-c_i).
\]
If \(i\) is even, then \(e_i\) is an Alice edge, so $u_{i+1}=-u_i=c_i$.
We define
\[
        c_{i+1}:=v_{i+1},
\]
which gives
\[
        (u_{i+1},v_{i+1})=(c_i,c_{i+1}).
\]
Thus, compactly, the recursion can be written as
\begin{equation}
\label{eq:rec_1}
        (u_i,v_i)=
        \begin{cases}
        (c_{i-1},c_i),&i\ \text{odd},\\
        (-c_i,-c_{i-1}),&i\ \text{even}.
        \end{cases}
\end{equation}
If \(e_1\) is an Alice
edge, set
\[
        c_0:=-v_1,\qquad c_1:=-u_1.
\]
The same recursion, with the parities reversed, gives
\begin{equation}
\label{eq:rec_2}
        (u_i,v_i)=
        \begin{cases}
        (-c_i,-c_{i-1}),&i\ \text{odd},\\
        (c_{i-1},c_i),&i\ \text{even}.
        \end{cases}
\end{equation}

Now, by the Definition of the graph $\Gamma$ (Definition~\ref{def:graph}) an Alice edge $\left((u,v) , (-u,v')\right)$ has weight $a(u)$, and moreover $0<|u|<1$. Similarly, the weight of a Bob edge $\left((u,v) , (u',-v)\right)$ is $a(v)$, and $0<|v|<1$. Crucially, note that, from our definition of the labels $c_i$, we have that the weight of edge $e_i$ in $\Gamma_0$ is precisely $a(c_i)$ (where the latter also uses the fact that $a$ is an even function), and moreover $0<|c_i|<1$, for all \(1\le i<m\). 

The definition of \(M_\Gamma\) in Lemma~\ref{lem:95} then gives
\begin{equation}
\label{eq:ac_path}
        (M_{\Gamma_0})_{i,i+1}
        =
        (M_{\Gamma_0})_{i+1,i}
        =
        a(c_i).
\end{equation}
At each endpoint of the path, exactly one of the two edge types is
missing. By our earlier observation at the start of the proof, this forces the corresponding coordinate to belong to \(\{-1,0,1\}\). At the first vertex in the path, this coordinate is, up to sign, $c_0$, and at the last vertex it is $c_m$. Thus, we have 
\[
        c_0,c_m\in\{-1,0,1\}.
\]

By definition, the diagonal entry in $M_{\Gamma_0}$ corresponding to vertex \((u_i,v_i)\) is \(\widehat d(u_i,v_i)\).  The relations \eqref{eq:rec_1} and \eqref{eq:rec_2}
and the \(\tau\)-invariance of \(\widehat d\) then give
\begin{equation}
\label{eq:d_path}
        (M_{\Gamma_0})_{ii}
        =
        \widehat d(c_{i-1},c_i).
\end{equation}
Together, \eqref{eq:ac_path} and \eqref{eq:d_path} give identity~\eqref{eq:102}.

Since every vertex belongs to \(E_H\), the same relations and the
\(\tau\)-invariance of \(E_H\) give
\[
        (c_{i-1},c_i)\in E_H.
\]
Moreover, since, up to a sign, $c_i$ is the ``label'' of edge $e_i$ for $1 \leq i <m$, Lemma~\ref{lem:zero-slack} also gives
\[
        H(c_i)H(-c_i)=a(c_i)^2.
\]
If an endpoint label, $c_0$ or $c_m$ happens to be zero, then Lemma~\ref{lem:zero-slack} also implies that
\(H(0)=\frac12\), since the label is the coordinate of a vertex $\nu$ in $\Gamma$, which, by definition, must have positive mass $m_\nu >0$. This finishes the path case.

\noindent \textbf{The cycle case.} Finally, suppose that \(\Gamma_0\) is a cycle. Let the vertices on the cycle be
\[
        \nu_1,\ldots,\nu_m,
\]
and let \(e_i\) be the edge joining \(\nu_i\) to \(\nu_{i+1}\). Here, we denote $\nu_{m+1}=\nu_1$.

Now, momentarily disregard the closing edge \(e_m\). The remaining vertices and
edges form a path. Define labels \(c_0,\ldots,c_m\) recursively
exactly as we did in the earlier path case. These satisfy either \eqref{eq:rec_1} or \eqref{eq:rec_2} depending on whether the first edge is a Bob or Alice edge.

Since $\Gamma_0$ is a cycle, $m$ is even. Moreover, by Lemma~\ref{lem:path-or-cycle}, Alice and Bob edges alternate around the cycle. Thus, if \(e_1\) is a Bob edge, then \(e_m\) is an Alice edge, and \eqref{eq:rec_1} gives
\[
        \nu_1=(c_0,c_1),
        \qquad
        \nu_m=(-c_m,-c_{m-1}).
\]
Because \(e_m\) is an Alice edge, the Alice coordinates of its endpoints
are opposite.  Hence \(c_0=c_m\).  If \(e_1\) is an Alice edge, then
\(e_m\) is a Bob edge, and the same argument using \eqref{eq:rec_2} gives
\[
        c_0=c_m.
\]

For \(1\le i<m\), the label \(c_i\), up to sign, is the ``label'' of
edge \(e_i\) (i.e.\ its weight in $\Gamma$ is $a(c_i)$) as in the path case.  The label \(c_m=c_0\), up to sign,
is similarly the label of the closing edge \(e_m\).  Thus every \(c_i\)
labels an actual Alice or Bob edge, and Definition~\ref{def:graph}
implies
\[
        0<|c_i|<1
        \qquad(1\le i\le m).
\]
In particular, \(\rho_0(c_i)=0\) for every \(i\), so
\[
        \widehat d(c_{i-1},c_i)=d(c_{i-1},c_i).
\]

The calculation from the path case gives the diagonal terms and the
contributions of \(e_1,\ldots,e_{m-1}\).  Adding the contribution of the
closing edge \(e_m\) gives
\[
\begin{aligned}
        \langle x,M_{\Gamma_0}x\rangle
        &=
        \sum_{i=1}^m d(c_{i-1},c_i)x_i^2
        +
        2\sum_{i=1}^m a(c_i)x_ix_{i+1},
        \qquad x_{m+1}:=x_1.
\end{aligned}
\]

As in the path case, each vertex \(\nu_i\) is represented either by
\((c_{i-1},c_i)\) or by
\((-c_i,-c_{i-1})\).  Since every vertex belongs to \(E_H\) and \(E_H\)
is invariant under \((u,v)\mapsto(-v,-u)\), it follows that
\[
        (c_{i-1},c_i)\in E_H
        \qquad(1\le i\le m).
\]
Moreover, because every \(c_i\) labels an edge, Lemma~\ref{lem:zero-slack} also gives, for all $1\le i\le m$, 
\[
        H(c_i)H(-c_i)=a(c_i)^2 \,.
\]

Note that when \(m=2\) the edges \(e_1\) and \(e_2\) are the ``parallel'' Alice and
Bob edges between the same two vertices.  The cyclic edge sum includes
both contributions, as required.

\end{proof}

\section{No finite-dimensional strategy attains $q_*$}
\label{sec:no-finite-dimensional-attainment}

Throughout the next two lemmas, \(\Gamma\) denotes the graph from
Definition~\ref{def:graph} associated with a hypothetical
finite-dimensional strategy attaining \(q_*\).

\begin{lemma}[No finite path component attains $q_*$]
\label{lem:no-finite-calibrated-path}
Any finite path component $\Gamma_0$ of \(\Gamma\) is such that $\lambda_{max}(M_{\Gamma_0}) < q_*$.
\end{lemma}

\begin{proof}
Let \(\Gamma_0\) be a finite path component. By Lemma~\ref{lem:component-jacobi}, after ordering the $m$ vertices in the path, there are labels \(c_0,\ldots,c_m\) such that
\[
        c_0,c_m\in\{-1,0,1\},
        \qquad
        c_i\in(-1,1)\setminus\{0\}
        \quad(1\le i<m),
\]
and
\[
        (M_{\Gamma_0})_{ii}=\widehat d(c_{i-1},c_i),
        \qquad
        (M_{\Gamma_0})_{i,i+1}=(M_{\Gamma_0})_{i+1,i}=a(c_i),
\]
where
\[
        \widehat d(u,v)
        =
        d(u,v)+\rho_0(u)+\rho_0(v),
        \qquad
        \rho_0(t)=\frac12\mathbf 1_{\{t=0\}}.
\]
There are three cases.

\medskip
\noindent
\textbf{Case 1: neither endpoint label is zero.}

In this case
\[
        c_0,c_m\in\{\pm1\}.
\]
Therefore \(\rho_0(c_0)=\rho_0(c_m)=0\), and \(M_{\Gamma_0}\) is exactly the
ordinary Jacobi matrix \(J_m(c)\), as defined in Definition~\ref{def:signed-endpoint-jacobi-path}. Hence
Corollary~\ref{cor:no-finite-path} gives
\[
        \lambda_{\max}(M_{\Gamma_0})<q_*.
\]

\medskip
\noindent
\textbf{Case 2: exactly one endpoint label is zero.}

If the zero occurs at the right endpoint, reverse the vertex order, i.e.\ define a new sequence $(c'_0,\ldots,c'_m)$ where
\[
        c'_i:=-c_{m-i},
        \qquad 0\le i\le m.
\]
Let \(M'\) be the matrix associated with this new sequence:
\[
        M'_{ii}
        :=
        \widehat d(c'_{i-1},c'_i),
        \qquad
        M'_{i,i+1}=M'_{i+1,i}:=a(c'_i).
\] 
This differs from $M_{\Gamma_0}$ only by a
permutation similarity, because
\[
        \widehat d(-v,-u)=\widehat d(u,v),
        \qquad
        a(-t)=a(t).
\]
Thus, $M'$ has the same eigenvalues as $M_{\Gamma_0}$. We may therefore assume without loss of generality that the original sequence has
\[
        c_0=0,
        \qquad
        c_m=\eta\in\{\pm1\}.
\]

Now, reflect the path through its zero endpoint and define the length-\(2m\) sequence
\[
        \widetilde c
        =
        (-c_m,-c_{m-1},\ldots,-c_1,0,
          c_1,\ldots,c_{m-1},c_m).
\]
Its two endpoint labels are \(-\eta\) and \(\eta\), so
\(\widetilde c\) is an ordinary finite Jacobi path, as in Definition~\ref{def:signed-endpoint-jacobi-path}. Let
\[
        \widetilde J:=J_{2m}(\widetilde c) \,,
\]
where $J_{2m}(\widetilde c)$ is the ordinary Jacobi matrix corresponding to the path $\widetilde c$, as in Definition~\ref{def:signed-endpoint-jacobi-path}. Define an isometry \(U:\mathbb R^m\to\mathbb R^{2m}\) by
\[
        U(x_1,\ldots,x_m)
        =
        \frac1{\sqrt2}
        (x_m,x_{m-1},\ldots,x_1,
         x_1,\ldots,x_{m-1},x_m).
\]
Thus \(Ux\) has the same amplitudes on the two reflected halves and
\[
        \|Ux\|=\|x\|.
\] 

We will now expand the quadratic form $ \langle Ux,\widetilde J\,Ux\rangle$ directly.  Note that if we index the sequence $\widetilde c$ by $(0,\ldots,2m)$, then the middle label is $\widetilde c_m=0$.
Moreover, the definition of \(U\) gives
\[
        (Ux)_m=(Ux)_{m+1}=\frac{x_1}{\sqrt2}.
\]
Thus, we have
\[
\begin{aligned}
        \langle Ux,\widetilde J\,Ux\rangle
        &=
        \frac12\sum_{i=1}^m
        \Bigl(
                d(-c_i,-c_{i-1})
                +
                d(c_{i-1},c_i)
        \Bigr)x_i^2\\
        &\quad+
        \sum_{i=1}^{m-1}
        \Bigl(
                a(-c_i)+a(c_i)
        \Bigr)x_ix_{i+1}\\
        &\quad+
        2a(0)\frac{x_1}{\sqrt2}\frac{x_1}{\sqrt2}.
\end{aligned}
\]
The last line is the contribution of the edge between vertices \(m\)
and \(m+1\): this edge has label \(\widetilde c_m=0\), and the
coordinates of \(Ux\) at those vertices are both \(x_1/\sqrt2\). Using
\[
        d(-v,-u)=d(u,v),
        \qquad
        a(-t)=a(t),
        \qquad
        a(0)=\frac12,
\]
the previous expression simplifies to
\[
\begin{aligned}
        \langle Ux,\widetilde J\,Ux\rangle
        &=
        \sum_{i=1}^m d(c_{i-1},c_i)x_i^2
        +
        2\sum_{i=1}^{m-1}a(c_i)x_ix_{i+1}
        +
        \frac12x_1^2.
\end{aligned}
\]
Since \(c_0=0\) and \(c_1,\ldots,c_m\ne0\),
\[
        \widehat d(c_{i-1},c_i)
        =
        d(c_{i-1},c_i)
        +
        \frac12\mathbf 1_{\{i=1\}}.
\]
Therefore,
\[
\begin{aligned}
        \langle x, M_{\Gamma_0} x\rangle
        &=
        \sum_{i=1}^m d(c_{i-1},c_i)x_i^2
        +
        2\sum_{i=1}^{m-1}a(c_i)x_ix_{i+1}
        +
        \frac12x_1^2\\
        &=
        \langle Ux,\widetilde J\,Ux\rangle.
\end{aligned}
\]
Since the latter holds for all $x\in\mathbb R^m$, and $\|U x\| = \|x\|$, we have 
$$ \lambda_{\max}(M_{\Gamma_0}) \leq \lambda_{\max}(\widetilde J) \,. $$
Note that the matrices above are symmetric, so maximizing the quadratic form over real vectors gives the same result as maximizing over complex vectors.

But now \(\widetilde J\) is a Jacobi matrix, so
Corollary~\ref{cor:no-finite-path} gives
\[
        \lambda_{\max}(\widetilde J)<q_*.
\]
Hence,
\[
        \lambda_{\max}(M_{\Gamma_0})<q_*.
\]

\medskip
\noindent
\textbf{Case 3: both endpoint labels are zero.}

Suppose $c_0=c_m=0$.
Lemma~\ref{lem:component-jacobi} gives, for \(1\le i\le m\),
\[
        (c_{i-1},c_i)\in E_H \quad \text{ and } \quad H(c_i)H(-c_i)=a(c_i)^2.
\]
Let $S:=\{c_0,\ldots,c_m\}$.
All points of \(S\) lie in \((-1,1)\).  The product equalities above
ensure that \(H\) is differentiable at every point of \(S\), so
Lemma~\ref{lem:single-valued} applies to each equality pair
\((c_{i-1},c_i)\).  Define, for \(t\in S\),
\[
        F(t):=\frac12-H'(t).
\]
Then, Lemma~\ref{lem:single-valued} gives
\begin{equation}
\label{eq:zero-zero-path-dynamics}
        c_{i-1}=F(c_i)
        \qquad(1\le i\le m).
\end{equation}
In particular, \(F\) is well-defined on \(S\); relation
\eqref{eq:zero-zero-path-dynamics}, together with \(c_m=c_0\), also
shows that \(F(S)\subseteq S\).
Starting from \(c_m=0\) and applying
\eqref{eq:zero-zero-path-dynamics} repeatedly gives
\[
\begin{aligned}
        c_{m-1}&=F(0),\\
        c_{m-2}&=F^2(0),\\
        &\ \vdots\\
        c_0&=F^m(0).
\end{aligned}
\]
Since \(c_0=0\), it follows that
\[
        F^m(0)=0.
\]
Thus the labels form the finite orbit
\[
        0=c_m\longmapsto c_{m-1}\longmapsto\cdots
        \longmapsto c_1\longmapsto c_0=0
\]
under \(F\).

Concavity of the certificate \(H\) implies that \(H'\) is nonincreasing at the points
where it exists.  Hence \(F=\frac12-H'\) is nondecreasing on this finite
orbit.  Moreover, every label \(x\) on the orbit satisfies \(F^m(x)=x\).  A
nondecreasing map cannot have any non-constant finite periodic orbit. Thus, it must be $F(x) = x$ for all $x$ on the orbit. 
In particular, $F(0)=0$.
The displayed recurrence then gives successively
\[
        c_{m-1}=0,\quad c_{m-2}=0,\quad\ldots,\quad c_0=0.
\]
If \(m\ge2\), this contradicts the fact that every non-endpoint label
\(c_i\) is nonzero.  If \(m=1\), the component consists of the isolated
vertex \((0,0)\), and
\[
        M_{\Gamma_0}=[\,d(0,0)+1\,]=[\,0\,] \,,
\]
where recall the definition of $d$ from \eqref{eq:a-d}.
Thus, the top eigenvalue of $M_{\Gamma_0}$ is \(0<q_*\), since
Lemma~\ref{lem:inf-quarter-witness} gives \(q_*>\frac14\).

\smallskip
Overall, the three cases all satisfy
\[
        \lambda_{\max}(M_{\Gamma_0})<q_* \,,
\]
as desired.
\end{proof}

\begin{lemma}[No finite cycle component attains $q_*$]
\label{lem:no-finite-cycle}
Every finite cycle component \(\Gamma_0\) satisfies
\[
        \lambda_{\max}(M_{\Gamma_0})\le\frac14<q_*.
\]
In particular, no such component can have top eigenvalue \(q_*\).
\end{lemma}

\begin{proof}
Use the cyclic labels from Lemma~\ref{lem:component-jacobi}.  Every label
belongs to \((-1,1)\setminus\{0\}\), so we actually have $\widehat d = d$ on the diagonal of
\(M_\Gamma\). Moreover,
for every \(i\), with indices read mod $m$,
\[
        (c_{i-1},c_i)\in E_H,
        \qquad
        H(c_i)H(-c_i)=a(c_i)^2.
\]
Lemma~\ref{lem:single-valued} therefore gives
\[
        c_{i-1}=F(c_i),
        \qquad
        F(t):=\frac12-H'(t).
\]
Concavity implies that \(F\) is nondecreasing on these labels, and a nondecreasing
map has no non-constant finite periodic orbit.  Hence it must be
\[
        c_1=\cdots=c_m=:c.
\]

Then, we have
\[
        \langle x,M_{\Gamma_0}x\rangle
        =
        d(c,c)\sum_{i=1}^m x_i^2
        +
        2a(c)\sum_{i=1}^m x_i x_{i+1},
        \qquad x_{m+1}=x_1.
\]
For every real \(x\),
\[
        2\sum_{i=1}^m x_i x_{i+1}
        \le
        \sum_{i=1}^m(x_i^2+x_{i+1}^2)
        =
        2\sum_{i=1}^m x_i^2.
\]
(where note that this also holds for $m=2$, in which the two terms in the edge sum correspond to ``parallel'' Alice and Bob edges). So, we have
\[
        \lambda_{\max}(M_{\Gamma_0})
        = d(c,c) + 2a(c)= 
        c^2-1+ \sqrt{1-c^2}.
\]
Setting \(y=\sqrt{1-c^2}\in[0,1]\), we obtain
\[
        c^2-1+\sqrt{1-c^2}
        =
        y-y^2
        =
        \frac14-\left(y-\frac12\right)^2
        \le\frac14.
\]
Finally, Lemma~\ref{lem:inf-quarter-witness} gives
\(q_*>\frac14\), proving the claim.
\end{proof}

\begin{theorem}[Finite-dimensional nonattainment]
\label{thm:finite-nonattainment}
For every finite \(d\),
\[
        \omega_d(I_{3322})<q_*.
\]
where $\omega_d(I_{3322})$ is the supremum over strategies of dimension $d$.
\end{theorem}

\begin{proof}
First, we show that any finite-dimensional strategy achieves a value strictly less than $q_*$. Suppose for a contradiction that a finite-dimensional strategy attains \(q_*\).
Let \(\Gamma\) be its associated graph as in
Definition~\ref{def:graph}.  The graph is finite.  By
Lemma~\ref{lem:95},
\[
        M_\Gamma\le q_*I,
        \qquad
        M_\Gamma\xi=q_*\xi,
        \quad \text{where } \,
        \xi_\nu=\sqrt{m_\nu}>0.
\]
Decompose $M_{\Gamma}$ and $\xi$ according to the connected-components of $\Gamma$:
\[
        M_\Gamma=\bigoplus_j M_{\Gamma_j},
        \qquad
        \xi=\bigoplus_j\xi^{(j)}.
\]
For every component, \(\xi^{(j)}\ne0\) and
\[
        M_{\Gamma_j}\xi^{(j)}
        =
        q_*\xi^{(j)}.
\]
Since also \(M_{\Gamma_j}\le q_*I\), it follows that
\[
        \lambda_{\max}(M_{\Gamma_j})=q_*.
\]
By Lemma~\ref{lem:component-jacobi}, every finite component is a path or
a cycle.  The path case contradicts
Lemma~\ref{lem:no-finite-calibrated-path}, and the cycle case contradicts
Lemma~\ref{lem:no-finite-cycle}.  Hence no finite-dimensional strategy
attains \(q_*\).

Finally, we show that for any fixed $d$, even the \emph{supremum} over finite-dimensional strategies is strictly less than $q_*$. 
After padding smaller local spaces to \(\mathbb C^d\), the sets of
states and of finite families of binary projectors
are compact, and the Bell value is continuous. Thus, the value
\(\omega_d(I_{3322})\) is attained, i.e.\ the supremum is a maximum. But, we have just shown that any strategy of dimension $d$ achieves a value stricly less than $q_*$. Therefore,
\[
        \omega_d(I_{3322})<q_*.
\]
\end{proof}

\section{Infinite-dimensional attainment}
\label{sec:infinite-dimensional-attainment}

In this section we prove that the value \(q_*\) is attained on separable infinite-dimensional Hilbert spaces. For an outline of the proof, see Section~\ref{sec:infinite-overview} of the technical overview.

\subsection{Two preliminary facts about $q_*$}

Recall that we defined $q_*$ as follows:
\[
        q_*:=\sup_{m\ge1}Q_m.
\]
where 
$$Q_m=\sup_{\substack{c_1,\ldots,c_{m-1}\in[-1,1]\\ c_0,c_m \in \{\pm 1\}}} \lambda_{\max}\left(J_m(c)\right)\,.$$ 
Here $J_m(c)$ is the Jacobi matrix from Definition~\ref{def:signed-endpoint-jacobi-path}. 

It will be convenient in later arguments to focus only on Jacobi paths with $c_0 = 1$ and $c_m = -1$, rather than arbitrary $c_0,c_m \in \{\pm1\}$. The next lemma shows that this does not change the supremum.

\begin{lemma}
\label{lem:inf-open-reduction}
Let
\[
 \tilde{q}_*
 :=
 \sup_{m\geq 1}
 \sup_{\substack{c_0=1,\ c_m=-1\\
                  c_1,\ldots,c_{m-1}\in[-1,1]}}
 \lambda_{\max}\left(J_m(c)\right).
\]
Then
\[
        \tilde{q}_*=q_*.
\]
\end{lemma}

\begin{proof}
The inequality $\tilde{q}_*\leq q_*$ is immediate.  Conversely, let
\(c=(c_0,\ldots,c_m)\) be any path with arbitrary endpoints in $\{\pm1\}$. If necessary,
prepend \(+1\) when \(c_0=-1\), and append \(-1\) when \(c_m=+1\).  The
result is a \(+1\) to \(-1\) path $c'$.  Since, $a(1)=a(-1)=0$, and $d(1,-1)=-1$, the new Jacobi matrix for $b'$ is the direct
sum of \(J_m(b)\) and at most two copies of the scalar matrix \([-1]\). In particular, its largest eigenvalue is at least
\(\lambda_{\max}J_m(b)\). Taking suprema gives
\(q_*\leq \tilde{q}_*\).
\end{proof}

We will make use of the inequality \(q_*>1/4\), which was already observed numerically by Pál and Vértesi. We prove this via an explicit witness.

\begin{lemma}
\label{lem:inf-quarter-witness}
        $q_*>\frac14$.
\end{lemma}

\begin{proof}
Let $t=\frac{\sqrt3}{2}$ and $r=\frac23$, and consider the Jacobi path of length \(46\) given by
\[
\begin{gathered}
        c_0=1,
        \qquad
        c_1=\cdots=c_{21}=t,
        \qquad
        c_{22}=r,
        \qquad
        c_{23}=0,\\
        c_{24}=-r,
        \qquad
        c_{25}=\cdots=c_{45}=-t,
        \qquad
        c_{46}=-1.
\end{gathered}
\]
Note the symmetry: \(c_{46-i}=-c_i\).
Let \(R\) denote the ``coordinate reversal'' map,
\[
        Re_j=e_{47-j}.
\]
Since \(c_{46-i}=-c_i\), the identities
\[
        d(-v,-u)=d(u,v),
        \qquad
        a(-t)=a(t)
\]
imply that \(J_{46}(c)\) commutes with \(R\).  Hence the
reversal-symmetric subspace
\[
        S = \{x\in\mathbb R^{46}:Rx=x\}
\]
is invariant under \(J_{46}(c)\).  We study the restriction of
\(J_{46}(c)\) to $S$. Consider the orthonormal basis of $S$
\[
        f_j=\frac{e_j+e_{47-j}}{\sqrt2},
        \qquad 1\le j\le23.
\]
The restriction \(B\) of \(J_{46}(c)\) to $S$ is tridiagonal in this basis. Its diagonal entries are
\[
\begin{aligned}
        B_{11}&=d(1,t),\\
        B_{jj}&=d(t,t) &&(2\leq j\leq21),\\
        B_{22,22}&=d(t,r),\\
        B_{23,23}&=d(r,0)+a(0)=\frac{r-1}{2},
\end{aligned}
\]
and its off-diagonal entries are
\[
        B_{j,j+1}=a(t)=\frac14
        \quad(1\leq j\leq21),
        \qquad
        B_{22,23}=a(r)=\frac{\sqrt5}{6}.
\]
We show that \(A:=\frac14I-B\) is not positive semidefinite.  
To test whether \(A\) is positive semidefinite, we compute its
\(LDL^{\mathsf T}\) factorization
\[
        A=LDL^{\mathsf T},
\]
where \(L\) is unit lower triangular and
\(D=\operatorname{diag}(p_1,\ldots,p_{23})\). The numbers \(p_j\) are referred to as ``pivots''. Because \(A\) is
tridiagonal, the pivots are given recursively by
\[
        p_1=A_{11},
        \qquad
        p_j=A_{jj}-\frac{A_{j,j-1}^2}{p_{j-1}}.
\]
We will find that the first \(22\) pivots are positive but the last one is negative.  Since \(A\) and \(D\) have the same numbers of positive and negative eigenvalues, this will show that \(A\) is not positive semidefinite.  The pivots can be computed to be
\[
\begin{aligned}
 p_1&=\frac{3-\sqrt3}{4},\\
 p_j&=\frac12-\frac{1/16}{p_{j-1}}
       &&(2\leq j\leq21),\\
 p_{22}&=\frac{19-7\sqrt3}{12}-\frac{1/16}{p_{21}},\\
 p_{23}&=\frac5{12}-\frac{5/36}{p_{22}}.
\end{aligned}
\]
With \(z=2-\sqrt3\), induction in the second line gives
\[
        p_j
        =
        \frac14\left(1+\frac{z}{1+(j-1)z}\right),
        \qquad 1\leq j\leq21.
\]
In particular, \(p_1,\ldots,p_{21}\) are positive.  Consequently,
\[
        p_{22}
        =
        \frac{5(1697-737\sqrt3)}{6312}>0,
\]
where positivity follows already from \(\sqrt3<2\).  Moreover,
\[
        p_{22}-\frac13
        =
        \frac{6381-3685\sqrt3}{6312}<0,
\]
because
\[
        3\cdot 3685^2-6381^2=20514>0.
\]
It follows that
\[
        p_{23}
        =
        \frac5{12}-\frac{5/36}{p_{22}}<0.
\]
Thus \(A\) is not positive semidefinite, so
\[
        \lambda_{\max}(B)>\frac14.
\]
Since \(B\) is a restriction of \(J_{46}(c)\), we obtain
\[
        q_*
        \geq\lambda_{\max}\left(J_{46}(c)\right)
        \geq\lambda_{\max}(B)
        >\frac14.
\]
\end{proof}

\subsection{Consequences of the certificate}

For the rest of the section, we fix a concave certificate function
\[
        H:[-1,1]\longrightarrow[0,\infty)
\]
supplied by Theorem~\ref{thm:scalar-certificate}. Recall that this satisfies:
\begin{align}
 H(v)+d(u,v)+H(-u)&\leq q_*
 &&(u,v\in[-1,1]),
 \label{eq:inf-certificate}\\
 H(t)H(-t)&\geq a(t)^2
 &&(t\in[-1,1]).
 \label{eq:inf-product}
\end{align}

The following elementary inequalities will later help us show that the the limit of the sequence of paths that we consider later has nonincreasing labels from left to right.
\begin{lemma}
\label{lem:inf-upward-row-sums}
For every \(u,v\in[-1,1]\),
\begin{equation}
        d(u,v)+a(v)\leq\frac14,
        \qquad
        d(u,v)+a(u)\leq\frac14.
\label{eq:inf-one-sided-row}
\end{equation}
If \(u\leq v\), then
\begin{equation}
        d(u,v)+a(u)+a(v)\leq\frac14.
\label{eq:inf-two-sided-row}
\end{equation}
\end{lemma}

\begin{proof}
We have 
$$d(u,v) + a(v) = uv + \frac{u-v}{2} -1 + \frac12\sqrt{1-v^2}\,.$$
For fixed \(v\), the latter is affine in \(u\),
with coefficient \(v+1/2\).  If \(v\geq-1/2\), its maximum over \(u \in [-1,1]\) is attained at \(u=1\), and
\[
        d(1,v)+a(v)
        =
        \frac{v-1+\sqrt{1-v^2}}2
        \leq
        \frac{\sqrt2-1}{2}
        <\frac14.
\]
If \(v\leq-1/2\), the maximum is attained at \(u=-1\), and
\[
        d(-1,v)+a(v)
        =
        -\frac32(v+1)+\frac12\sqrt{1-v^2}
        \leq
        \frac{\sqrt{10}-3}{2}
        <\frac14.
\]
This proves the first inequality in \eqref{eq:inf-one-sided-row}.  The
second follows from
\[
        d(u,v)=d(-v,-u),
        \qquad
        a(t)=a(-t).
\]
For the third, let
\[
        L(u,v):=d(u,v)+a(u)+a(v)
\]
Assume \(-1\leq u\leq v\leq1\). Suppose that \(L\) attained its maximum at an interior point $(u,v)$ such that
\[
        -1<u<v<1.
\]
Since \(F\) is differentiable there, both of its partial derivatives
would have to vanish.  The equations
\[
        \frac{\partial L}{\partial u}=0,
        \qquad
        \frac{\partial L}{\partial v}=0
\]
are equivalent to
\[
        \frac{u}{\sqrt{1-u^2}}=2v+1,
        \qquad
        \frac{v}{\sqrt{1-v^2}}=2u-1.
\]
This is impossible: the left side of the first equation is smaller than
the left side of the second, whereas its right side is larger.  Hence the maximum must lie on the boundary of the triangle \(-1\leq u\leq v\leq1\). On \(u=v=t\),
\[
        L(t,t)
        =t^2-1+\sqrt{1-t^2}
        =-y^2+y
        \leq\frac14,
        \qquad y:=\sqrt{1-t^2}.
\]
On \(u=-1\) or \(v=1\), the result follows from
\eqref{eq:inf-one-sided-row}, because \(a(\pm1)=0\).  This proves
\eqref{eq:inf-two-sided-row}.
\end{proof}

\medskip
Let
\[
        D_H
        :=
        \left\{
        t\in(-1,1):
        H(t)H(-t)=a(t)^2
        \right\}.
\]
By Lemma~\ref{lem:differentiability}, \(H\) is differentiable at every
point of \(D_H\).  Moreover, \(D_H\) is invariant under
\(t\mapsto-t\).  We may therefore define
\begin{equation}
\label{eq:F-and-G}
        F(t):=\frac12-H'(t),
        \qquad
        G(t):=-F(-t)=H'(-t)-\frac12,
        \qquad t\in D_H.
\end{equation}
Recall also that we defined the other equality set
\[
        E_H
        :=
        \{(u,v)\in[-1,1]^2:
        H(v)+d(u,v)+H(-u)=q_*\}\,.
\]
Whenever \((u,v)\in E_H\), Lemma~\ref{lem:single-valued} gives
\[
        v\in D_H \Longrightarrow u=F(v),
        \qquad
        u\in D_H \Longrightarrow v=G(u).
\]

\begin{lemma}[Monotonicity of the equality maps]
\label{lem:inf-map-monotonicity}
The maps $F$ and $G$ are nondecreasing on $D_H$.
\end{lemma}

\begin{proof}
Let $s<t$ with $s,t \in D_H$.  Since $H$ is concave,
\[
        H'(s)\ge H'(t),
\]
and therefore
\[
        F(s)\le F(t).
\]
Thus $F$ is nondecreasing.  Then, since $-t<-s$, the the monotonicity
of $F$ also gives
\[
        F(-t)\le F(-s).
\]
Multiplying the latter by $-1$, we obtain
\[
        G(s)=-F(-s)\le -F(-t)=G(t).
\]
Hence $G$ is also nondecreasing.
\end{proof}

\subsection{A monotone generalized eigenvector}
For a label sequence \(c=(c_i)_{i\in\mathbb Z}\) with $c_i \in [-1,1]$, define the operator $J_\infty(c)$ as follows. It acts on $x = (x_i)_{i\in\mathbb Z}$ by
\begin{equation}
 (J_\infty(c)x)_i
 =a(c_{i-1})x_{i-1}
  +d(c_{i-1},c_i)x_i
  +a(c_i)x_{i+1}.
\label{eq:inf-Jacobi-operator}
\end{equation}
Since \(c_i\in[-1,1]\) and \(a,d\) are continuous, the diagonal and
off-diagonal coefficients are bounded uniformly in \(i\).  Then, it is straightforward to see that
\(J_\infty(C)\) defines a bounded self-adjoint operator on
\(\ell^2(\mathbb Z)\) with respect to the usual operator norm. In fact, it is straightforward to see, via a triangle inequality for the operator norm applied to the diagonal and off-diagonal components of $J_{\infty}(c)$, that
\[
        \|J_\infty(C)\|
        \le
        \sup_i|d(C_{i-1},C_i)|
        +
        2\sup_i a(C_i)
        <\infty.
\]
More generally,
for an interval \(\Lambda\subseteq\mathbb Z\), let $J_\Lambda(C)$ denote the
operator on \(\ell^2(\Lambda)\) with the same diagonal entries and only the
``off-diagonal'' terms whose two endpoints lie in $\Lambda$.

\begin{lemma}
\label{lem:inf-Jacobi-upper-bound}
For every sequence \(c\) and every interval \(\Lambda\subseteq\mathbb Z\),
\[
        J_\Lambda(c)\leq q_*I.
\]
In particular, \(J_\infty(c)\leq q_*I\).
\end{lemma}

\begin{proof}
Let \(x\) have finite support in \(\Lambda\) with real entries. For every $i$, \eqref{eq:inf-product} and AM-GM give
\begin{equation}
\label{eq:44}
 2a(c_i)\overline{x_i}x_{i+1}
 \leq 2 \sqrt{H(c_i)H(-c_i)} |x_i||x_{i+1}| \leq
 H(c_i)|x_i|^2+H(-c_i)|x_{i+1}|^2.
\end{equation}
Using this, we have
\begin{align*}
 \langle x,J_\Lambda(c)x\rangle
 &\leq
 \sum_{i\in \Lambda}
 \bigl[d(c_{i-1},c_i)+H(c_i)+H(-c_{i-1})\bigr]|x_i|^2\\
 &\leq q_*\|x\|_2^2
\end{align*}
where the first inequality follows from \eqref{eq:44} after summing all off-diagonal contributions (and noticing that a boundary index might be missing an \(H\)-term, but that term is non-negative); and the second inequality is by \eqref{eq:inf-certificate}. Since finitely supported vectors are dense in \(\ell^2(\Lambda)\) and
\(J_\Lambda(c)\) is bounded, the inequality extends by continuity to every
\(x\in\ell^2(\Lambda)\), which gives the desired operator inequality
\[
        J_\Lambda(c)\le q_*I.
\]
Consider $x$ with real entries is enough to get the above since $J_{\Lambda}(c)$ is a real symmetric matrix.
\end{proof}

As outlined in Section~\ref{sec:infinite-overview} of the technical overview, the first main step of the proof of infinite-dimensional attainment is to construct a nonzero sequence $\phi\in\ell^\infty(\mathbb Z)$ satisfying
\[
        J_\infty(c)\phi=q_*\phi \,.
\]
We refer to such a $\phi$ as a generalized
eigenvector. This need not belong to \(\ell^2(\mathbb Z)\), only to $\ell^\infty(\mathbb Z)$, so it will not yet be a genuine eigenvector of \(J_\infty(c)\) as an operator on
\(\ell^2(\mathbb Z)\).

\begin{lemma}
\label{lem:inf-bounded-ground-state}
There are labels \(c_i\in[-1,1]\) and a sequence
\(\phi=(\phi_i)_{i\in\mathbb Z}\) such that
\[
        0\leq\phi_i\leq1,
        \qquad
        \phi_0=1,
\]
and
\begin{equation}
        J_\infty(c)\phi=q_*\phi
\label{eq:inf-generalized-eigen-equation}
\end{equation}
coordinatewise on \(\mathbb Z\).
\end{lemma}

\begin{proof}
By Lemma~\ref{lem:inf-open-reduction}, we can choose finite paths
\[
        c^{(k)}
        =
        \bigl(c_0^{(k)},c_1^{(k)},\ldots,c_{n_k}^{(k)}\bigr),
\]
with
\[
        c_0^{(k)}=1,
        \qquad
        c_{n_k}^{(k)}=-1,
        \qquad
        c_i^{(k)}\in[-1,1]
        \quad(1\le i<n_k),
\]
such that
\[
        q_k
        :=
        \lambda_{\max}\left(J_{n_k}(c^{(k)})\right)
        \rightarrow q_*.
\] 
For each \(k\), the matrix \(J_{n_k}(c^{(k)})\) has nonnegative
off-diagonal entries.  Its largest eigenvalue \(q_k\) therefore admits
an eigenvector \(u^{(k)}\) with nonnegative coordinates:
\[
        J_{n_k}(c^{(k)})u^{(k)}=q_k u^{(k)},
        \qquad
        u_i^{(k)}\ge0.
\]
Since \(q_k\to q_*\) and Lemma~\ref{lem:inf-quarter-witness} gives
\(q_*>\frac14\), after discarding finitely many terms we may assume
\[
        q_k>0
        \qquad\text{for every }k.
\]

Pad the path $c^{(k)}$ by \(k\) additional \(+1\) labels on the
left and \(k\) additional \(-1\) labels on the right.  Since
\[
        a(\pm1)=0,
        \qquad
        d(1,1)=d(-1,-1)=0,
\]
the Jacobi matrix corresponding to the padded path is
\[
        0_k\oplus J_{n_k}(c^{(k)})\oplus0_k.
\]
Extend \(u^{(k)}\) by zeros on the two new blocks. Choose a coordinate
\(s_k\) at which this padded vector is maximal, and translate \(s_k\) to the
index \(0\). In other words, relabel the coordinates relative to the chosen maximal coordinate:
assign it index \(0\), and use negative and positive indices for the
coordinates to its left and right, respectively. Divide the vector by its maximal coordinate. Denote by \(\phi^{(k)}\) the resulting vector. Then
\[
        0\leq\phi_i^{(k)}\leq1,
        \qquad
        \phi_0^{(k)}=1.
\]
Going forward, we use \(c^{(k)}\) to denote the corresponding padded and shifted path. Notice that padding and shifting do not remove any of the original eigenvalues, and \(\phi^{(k)}\) remains an
eigenvector of the Jacobi matrix associated
with the padded and shifted path \(c^{(k)}\).

Now, the distance from coordinate \(0\) to either endpoint is at
least \(k\). Hence, for every fixed \(i\in\mathbb Z\), the
coordinates \(i-1,i,i+1\) occur in path \(c^{(k)}\) for all sufficiently
large \(k\). So, the eigenvalue equation at \(i\) is
\begin{equation}
\label{eq:evalue-eq}
        a(c_{i-1}^{(k)})\phi_{i-1}^{(k)}
        +d(c_{i-1}^{(k)},c_i^{(k)})\phi_i^{(k)}
        +a(c_i^{(k)})\phi_{i+1}^{(k)}
        =
        q_k\phi_i^{(k)}.
\end{equation}
for all sufficiently large $k$. Since $c_i^{(k)} \in [-1,1]$ and $\phi_i^{(k)} \in [0,1]$,
and \(\mathbb Z\) is countable, a nested
subsequence argument gives a single subsequence along which
\[
        c_i^{(k)}\rightarrow c_i,
        \qquad
        \phi_i^{(k)}\rightarrow\phi_i
\]
for every \(i \in\mathbb Z\), for some $c_i$ and $\phi$, with $0 \leq \phi_i \leq 1$ and $\phi_0 =1$.

Passing to the limit in \eqref{eq:evalue-eq},
and using the continuity of \(a\) and \(d\) together with
\(q_k\to q_*\), gives
\[
        a(c_{i-1})\phi_{i-1}
        +d(c_{i-1},c_i)\phi_i
        +a(c_i)\phi_{i+1}
        =
        q_*\phi_i.
\]
Thus \(J_\infty(c)\phi=q_*\phi\) coordinatewise, as desired.
\end{proof}

For the rest of the section, we fix a path $c$ and a generalized eigenvector $\phi$ of $J_\infty(c)$ as supplied by Lemma~\ref{lem:inf-bounded-ground-state}.

Define a graph on \(\mathbb Z\) by joining \(i\) to \(i+1\) whenever
\(a(c_i)>0\), and let $\Lambda$ be the connected component containing \(0\).
If two indices are joined and one has a positive \(\phi\)-coordinate,
then so does the other. Indeed, at an index \(j\) with \(\phi_j=0\),
the coordinate equation \(J_\infty(c)\phi=q_*\phi\) reduces to
\[
        a(c_{j-1})\phi_{j-1}
        +a(c_j)\phi_{j+1}=0.
\]
Since both terms are nonnegative, if \(j\) were joined to \(j-1\), i.e.\ $a(c_{j-1})>0$, and
\(\phi_{j-1}>0\), then the first term would be
strictly positive, contradicting the equality.  The same argument
applies to \(j+1\). Now, since \(\phi_0=1\), it follows that
\(\phi_i>0\) for every \(i\in \Lambda\).  Moreover, \(\Lambda\) is an interval, and,
by construction, \(a(c_i)>0\) for every edge \(\{i,i+1\}\) contained in
\(\Lambda\).
\begin{lemma}
\label{lem:inf-calibrated-component}
The interval \(\Lambda\) is infinite. Moreover, at every vertex \(i\in \Lambda\),
\begin{equation}
 H(c_i)+d(c_{i-1},c_i)+H(-c_{i-1})=q_*\,,
\label{eq:inf-scalar-calibration}
\end{equation}
and at every internal edge \(\{i,i+1\}\subseteq \Lambda\),
\begin{equation}
        H(c_i)H(-c_i)=a(c_i)^2.
\label{eq:inf-product-calibration}
\end{equation}
\end{lemma}

\begin{proof}
Suppose for a contradiction that \(\Lambda=\{r,\ldots,s\}\) was finite.  Since \(\Lambda\) is a
connected component, neither boundary edge is present, so
\[
        a(c_{r-1})=a(c_s)=0.
\]
Because \(a(t)=0\) only for \(t=\pm1\), the boundary labels satisfy
\[
        c_{r-1},c_s\in\{\pm1\}.
\]
In the coordinate equation \(J_\infty(c)\phi=q_*\phi\) at \(r\), the
term \(a(c_{r-1})\phi_{r-1}\) vanishes because \(a(c_{r-1})=0\).
Likewise, in the equation at \(s\), the term
\(a(c_s)\phi_{s+1}\) vanishes.  Thus the equations for the coordinates in \(\Lambda\) do not involve any
coordinates outside \(\Lambda\).  Let $m =s-r+1$,
and define
\[
        b=(b_0,\ldots,b_m)
        :=
        (c_{r-1},c_r,\ldots,c_s).
\]
Then, from the above, we have \(J_m(b)\)
\[
        J_m(b)\bigl(\phi|_\Lambda\bigr)
        =
        q_*\bigl(\phi|_\Lambda\bigr).
\]
where $\phi|_\Lambda$ is the restriction of $\phi$ to coordinates in $\Lambda$.

Thus $\phi|_\Lambda$ is an eigenvector of the Jacobi matrix $J_m(b)$ with eigenvalue $q_*$. This contradicts Corollary~\ref{cor:no-finite-path}, which says that the largest eigenvalue of any Jacobi matrix of a finite path is strictly less than $q_*$. Hence \(\Lambda\) is infinite.

For the second part of the lemma statement, let
\[
        A:=q_*I-J_\Lambda(c)\geq0.
\]
Similarly to above, restricting
the equation \(J_\infty(c)\phi=q_*\phi\) to \(\Lambda\) gives
\[
        A(\phi|_\Lambda)=0
\]
coordinatewise. For a finitely supported real function
\(f=(f_i)_{i\in \Lambda}\) on \(\Lambda\) to be chosen later, define the vector $x$ with entries in $\Lambda$ by
\[
        x_i=\phi_i f_i,
        \qquad i\in \Lambda.
\]
The coordinatewise equation
\(A(\phi|_\Lambda)=0\) reads
\[
\bigl(q_*-d(c_{i-1},c_i)\bigr)\phi_i
=
a(c_{i-1})\phi_{i-1}+a(c_i)\phi_{i+1},
\]
Multiplying by
\(\phi_i f_i^2\), summing over \(i\), and grouping terms by edges gives
\[
\sum_{i\in \Lambda}
\bigl(q_*-d(c_{i-1},c_i)\bigr)\phi_i^2f_i^2
=
\sum_{\{i,i+1\}\subseteq \Lambda}
a(c_i)\phi_i\phi_{i+1}(f_i^2+f_{i+1}^2).
\]
where we are using the fact that if $j$ is a left endpoint, then $a(c_{j-1}) = 0$, and if $j$ is a right endpoint, then $a(c_j) = 0$. Now, we have
\[
\langle x,Ax\rangle
=
\sum_{i\in \Lambda}
\bigl(q_*-d(c_{i-1},c_i)\bigr)\phi_i^2f_i^2
-
2\sum_{\{i,i+1\}\subseteq \Lambda}
a(c_i)\phi_i\phi_{i+1}f_if_{i+1} \,,
\]
and using the previous identity to rewrite the first sum, we obtain
\begin{align}
\langle x,Ax\rangle
&=
\sum_{\{i,i+1\}\subseteq \Lambda}
a(c_i)\phi_i\phi_{i+1}
\bigl(f_i^2+f_{i+1}^2-2f_if_{i+1}\bigr) \nonumber \\
&=
\sum_{\{i,i+1\}\subseteq \Lambda}
a(c_i)\phi_i\phi_{i+1}(f_i-f_{i+1})^2.
\label{eq:inf-ground-state-transform}
\end{align}
The usefulness of \eqref{eq:inf-ground-state-transform} is that
\(\langle x,Ax\rangle\) can be made arbitrarily small by choosing \(f\)
to vary sufficiently slowly.  We now obtain a complementary lower bound
that expresses the same quantity as a sum of nonnegative local terms.
Comparing the two bounds will force each of these terms to vanish. Define
\[
 s_i
 :=q_*-d(c_{i-1},c_i)-H(c_i)-H(-c_{i-1})\geq0.
\]
Expanding the quadratic form, substituting the definition of \(s_i\),
and grouping the \(H\)-terms by edges gives
\begin{align}
\langle x,Ax\rangle
&=
\sum_{i\in\Lambda}
\bigl(q_*-d(c_{i-1},c_i)\bigr)x_i^2
-
2\sum_{\{i,i+1\}\subseteq\Lambda}
a(c_i)x_ix_{i+1}
\notag\\
&=
\sum_{i\in\Lambda}s_i x_i^2
+
\sum_{\{i,i+1\}\subseteq\Lambda}
\Bigl[
H(c_i)x_i^2+H(-c_i)x_{i+1}^2
-2a(c_i)x_ix_{i+1}
\Bigr]
+E_\Lambda(x)
\notag\\
&\geq
\sum_{i\in\Lambda}s_i x_i^2
+
\sum_{\{i,i+1\}\subseteq\Lambda}
\Bigl[
H(c_i)x_i^2+H(-c_i)x_{i+1}^2
-2a(c_i)x_ix_{i+1}
\Bigr],
\label{eq:inf-slack-lower-bound}
\end{align}
where \(E_\Lambda(x)\geq0\) consists of the unpaired \(H\)-terms at any
finite endpoints of \(\Lambda\).  Each bracket in the final edge sum is
nonnegative by \eqref{eq:inf-product} and AM--GM.

We first prove that \eqref{eq:inf-scalar-calibration} holds with equality at
every vertex \(i\in\Lambda\), that is,
\[
H(c_i)+d(c_{i-1},c_i)+H(-c_{i-1})=q_*.
\]
Fix \(i\in\Lambda\), and choose a finite interval
\(K\subseteq\Lambda\) containing \(i\) and any finite endpoint of
\(\Lambda\).  For \(N\geq1\), define
\[
f_j^{(N)}
:=
\max\left\{
1-\frac{\operatorname{dist}(j,K)}{N},\,0
\right\},
\qquad
x_j^{(N)}:=\phi_jf_j^{(N)}.
\]
where $\operatorname{dist}(j,K)$ is the minimum distance of $j$ from any index in $K$. Then \(f^{(N)}\) is finitely supported, equals \(1\) on \(K\), and
\[
\bigl|f_j^{(N)}-f_{j+1}^{(N)}\bigr|\leq\frac1N.
\]
Moreover, \(f^{(N)}\) changes across at most \(2N\) edges. Since
\[
0\leq a(c_j)\phi_j\phi_{j+1}\leq\frac12,
\]
the identity \eqref{eq:inf-ground-state-transform} gives
\begin{equation}
\label{eq:87}
0\leq
\langle x^{(N)},Ax^{(N)}\rangle = \left\langle\phi f^{(N)},A\phi f^{(N)}\right\rangle
\leq
\frac12(2N)\frac1{N^2}
=
\frac1N 
\longrightarrow0.
\end{equation}
On the other hand, all the terms on the right-hand side of
\eqref{eq:inf-slack-lower-bound} are nonnegative.  Since
\(f_i^{(N)}=1\), that bound implies
\[
0\leq s_i\phi_i^2
\leq
\langle x^{(N)},Ax^{(N)}\rangle
\longrightarrow0.
\]
Hence, since \(s_i\phi_i^2=0\) is independent of $N$, it must be \(s_i\phi_i^2=0\). Since \(\phi_i>0\), we have \(s_i=0\), which,
by the definition of \(s_i\), is precisely
\[
H(c_i)+d(c_{i-1},c_i)+H(-c_{i-1})=q_*.
\]

We next prove that \eqref{eq:inf-product-calibration} holds for every $\{i, i+1\} \subseteq \Lambda$. Fix an edge
\(\{i,i+1\}\subseteq\Lambda\), and choose \(K\) containing both \(i\) and \(i+1\) and any finite endpoint of $\Lambda$. Then
\[
f_i^{(N)}=f_{i+1}^{(N)}=1.
\]
Since every term on the right-hand side of
\eqref{eq:inf-slack-lower-bound} is nonnegative, we obtain
\[
\begin{aligned}
0
&\leq
H(c_i)\phi_i^2+H(-c_i)\phi_{i+1}^2
-2a(c_i)\phi_i\phi_{i+1}\\
&\leq
\langle x^{(N)},Ax^{(N)}\rangle.
\end{aligned}
\]
The last quantity tends to zero, so it must be
\[
H(c_i)\phi_i^2+H(-c_i)\phi_{i+1}^2
-2a(c_i)\phi_i\phi_{i+1}=0.
\]
Now, consider the matrix
\[
M_i=
\begin{pmatrix}
H(c_i)&-a(c_i)\\
-a(c_i)&H(-c_i)
\end{pmatrix}.
\]
It is positive semidefinite because it is symmetric, its diagonal entries are non-negative and, by
\eqref{eq:inf-product},
\[
\det(M_i)
=
H(c_i)H(-c_i)-a(c_i)^2
\geq0.
\]
Thus, the sum and product of its eigenvalues is non-negative, which means both eigenvalues must be non-negative. Now, the preceding equality can be written as
\[
\begin{pmatrix}\phi_i&\phi_{i+1}\end{pmatrix}
M_i
\begin{pmatrix}\phi_i\\ \phi_{i+1}\end{pmatrix}
=0.
\]
Since \(M_i\) is positive semidefinite, this implies that
\((\phi_i,\phi_{i+1})^{\mathsf T}\) belongs to its kernel. This vector
is nonzero because \(\phi_i,\phi_{i+1}>0\), so \(M_i\) has a zero eigenvalue.
Hence, its determinant is zero, which means
\[
H(c_i)H(-c_i)=a(c_i)^2\,.
\]
This proves \eqref{eq:inf-product-calibration}.

\end{proof}

\begin{lemma}
\label{lem:inf-component-monotone}
The labels \(c_i\) for which
\(\{i,i+1\}\subseteq\Lambda\) are nonincreasing from left to right.  In addition:
\begin{itemize}
\item if \(\Lambda=[r,\infty)\), then \(c_{r-1}=+1\);
\item if \(\Lambda=(-\infty,s]\), then \(c_s=-1\).
\end{itemize}
\end{lemma}

\begin{proof}
Fix \(i\) such that \(i,i+1,i+2\in\Lambda\).
\eqref{eq:inf-scalar-calibration} at vertex \(i+1\) gives
\[
        (c_i,c_{i+1})\in E_H,
\]
while \eqref{eq:inf-product-calibration} gives
\(c_{i+1}\in D_H\).  Lemma~\ref{lem:single-valued} and the definition of
\(F\) from \eqref{eq:F-and-G} therefore imply
\[
        c_i=F(c_{i+1}).
\]
By Lemma~\ref{lem:inf-map-monotonicity}, \(F\) is nondecreasing, so any inequality between consecutive terms
propagates to the left.  For example,
\[
c_{i+1}\leq c_i
\quad\Longrightarrow\quad
c_i=F(c_{i+1})\leq F(c_i)=c_{i-1},
\]
and the same argument applies with all inequalities reversed.
Consequently, the direction of the inequalities cannot change, so the
sequence \((c_i)\) is monotone.

It remains to determine the direction.  Suppose for a contradiction that the labels
are nondecreasing from left to right. If \(0\) is an interior point of \(\Lambda\), then evaluating
\eqref{eq:inf-generalized-eigen-equation} at coordinate \(0\), using
\eqref{eq:inf-Jacobi-operator} and \(\phi_0=1\), gives
\[
q_*
=
d(c_{-1},c_0)
+a(c_{-1})\phi_{-1}
+a(c_0)\phi_1.
\]
Since \(a\geq0\) and
\(0\leq\phi_{-1},\phi_1\leq1\), and using Lemma~\ref{lem:inf-upward-row-sums}, we have
\[
q_*
\leq
d(c_{-1},c_0)+a(c_{-1})+a(c_0) \leq \frac14 \,.
\]
If \(0\) is an endpoint of \(\Lambda\), the coupling to the neighboring coordinate outside
\(\Lambda\) vanishes, and the corresponding bound in
\eqref{eq:inf-one-sided-row} again gives \(q_*\leq\frac14\). This
contradicts Lemma~\ref{lem:inf-quarter-witness}, which says that $q_* >\frac14$. Therefore, the labels
must be nonincreasing from left to right.

Suppose now that \(I=[r,\infty)\). Then, we have
\(c_{r-1}\in\{\pm1\}\), whereas \(c_r,c_{r+1}\in(-1,1)\). We will establish that $c_{r-1} = +1$. By \eqref{eq:inf-scalar-calibration} at the vertices \(r\) and \(r+1\),
\[
(c_{r-1},c_r)\in E_H,
\qquad
(c_r,c_{r+1})\in E_H.
\]
Moreover, \eqref{eq:inf-product-calibration} gives
\(c_r,c_{r+1}\in D_H\).  Therefore,
Lemma~\ref{lem:single-valued} and the definition of \(F\) from \eqref{eq:F-and-G} imply
\[
c_{r-1}=F(c_r),
\qquad
c_r=F(c_{r+1}).
\]
Since \(c_r\geq c_{r+1}\) and \(F\) is nondecreasing, we have $c_{r-1}\geq c_r$. Now, \(c_{r-1} \in \{\pm1 \}\) but $c_r \in (-1,1)$. So it must be \(c_{r-1}=+1\).

If \(I=(-\infty,s]\), then
\(c_s\in\{\pm1\}\), while \(c_{s-2},c_{s-1}\in(-1,1)\).
Using \eqref{eq:inf-scalar-calibration} at \(s-1\) and \(s\),
together with \eqref{eq:inf-product-calibration},
Lemma~\ref{lem:single-valued}, and the definition of \(G\), gives
\[
c_{s-1}=G(c_{s-2}),
\qquad
c_s=G(c_{s-1}).
\]
Since \(c_{s-2}\geq c_{s-1}\) and \(G\) is nondecreasing,
\[
c_{s-1}=G(c_{s-2})\geq G(c_{s-1})=c_s.
\]
Because \(c_{s-1}\in(-1,1)\) and \(c_s\in\{\pm1\}\), this forces
\(c_s=-1\).
\end{proof}

\subsection{Monotone paths approaching $q_*$}
The preceding subsection produced an infinite sequence $c$, as in Lemma~\ref{lem:inf-bounded-ground-state}, that is guaranteed to be nonincreasing on \(\Lambda\) (the interval defined just before Lemma~\ref{lem:inf-calibrated-component}) together with a generalized eigenvector
\(\phi\) of $J_\infty(c)$ with eigenvalue \(q_*\). The original finite paths used to construct \(c\) were not known to be monotone, however, and \(\phi\)
need not belong to \(\ell^2(\mathbb{Z})\). Proposition~\ref{prop:inf-monotone-approximants} below constructs monotone finite paths
\[
1=b_0^{(N)}\geq b_1^{(N)}
\geq\cdots\geq b_{m_N}^{(N)}=-1
\]
whose largest eigenvalues approach \(q_*\).  Thus these paths are
simultaneously finite, nonincreasing, and optimal in the limit. Their monotonicity will be essential below because it gives $\operatorname{TV}\bigl(b^{(N)}\bigr)=2$. We will later show how to build a genuine eigenvector of $J_{\infty}(c)$ with eigenvalue $q_*$ from such paths.

\begin{proposition}[Monotone finite paths approaching \(q_*\)]
\label{prop:inf-monotone-approximants}
There are finite paths
\[
        1=b^{(N)}_0
        \geq b^{(N)}_1
        \geq\cdots
        \geq b^{(N)}_{m_N}=-1
\]
such that
\[
        \lambda_{\max}\left(J_{m_N}(b^{(N)})\right)\rightarrow q_*.
\]
\end{proposition}

\begin{proof}
Choose any finite interval $K \subseteq \Lambda$ containing $0$ and any endpoint of $\Lambda$. Let \(f^{(N)}\) be the same vectors used in the proof of
Lemma~\ref{lem:inf-calibrated-component}, defined with respect to this $K$. Define the vectors $x^{(N)}$ by
\[
x_i^{(N)}:=\phi_i f_i^{(N)},
\qquad i\in\Lambda.
\]
Applying \eqref{eq:inf-ground-state-transform} with
\(f=f^{(N)}\) gives
\[
\begin{aligned}
0
&\leq
q_*\|x^{(N)}\|_2^2
-
\langle x^{(N)},J_\Lambda(c)x^{(N)}\rangle\\
&=
\sum_{\{i,i+1\}\subseteq\Lambda}
a(c_i)\phi_i\phi_{i+1}
\bigl(f_i^{(N)}-f_{i+1}^{(N)}\bigr)^2
\leq\frac1N.
\end{aligned}
\]
as in \eqref{eq:87}.
Hence
\[
q_*\|x^{(N)}\|_2^2
-
\langle x^{(N)},J_\Lambda(c)x^{(N)}\rangle
\rightarrow0.
\]
Since $f^{N}_0 = 1$, we have $\|x^{(N)}\|_2\geq\phi_0=1$. So, we can safely divide both sides by $\|x^{(N)}\|^2_2$ and deduce that
\begin{equation}
 \frac{\langle x^{(N)},J_\Lambda(C)x^{(N)}\rangle}
      {\|x^{(N)}\|_2^2}
 \rightarrow q_*.
\label{eq:inf-cutoff-Rayleigh}
\end{equation}

Let \([L_N,R_N]\) be the support interval of \(f^{(N)}\).  Define the
finite path
\[
 b^{(N)}
 =
 \bigl(1,c_{L_N-1},c_{L_N},\ldots,c_{R_N},-1\bigr).
\]
Set $m_N=R_N-L_N+3$,
which is the length of this path and the dimension of its Jacobi matrix.
At the left end, there are two cases.  If \(L_N\) is not the left
endpoint of \(\Lambda\), then \(\Lambda\) continues to \(L_N-1\), and the
monotonicity already proved gives
\[
c_{L_N-1}\geq c_{L_N}.
\]
If \(L_N\) is the left endpoint of \(\Lambda\), then
Lemma~\ref{lem:inf-component-monotone} gives
\(c_{L_N-1}=1\).  Thus, in either case,
\[
1\geq c_{L_N-1}\geq c_{L_N}.
\] The analogous
statement with $-1$ holds at the right endpoint. Therefore,
\[
        1\geq c_{L_N-1}\geq c_{L_N}\geq\cdots
        \geq c_{R_N}\geq-1,
\]
and so \(b^{(N)}\) is monotone.

Let
\[
        y^{(N)}
        =
        \bigl(0,x^{(N)}_{L_N},\ldots,x^{(N)}_{R_N},0\bigr) \,.
\]
Since the two added coordinates are zero, a direct
inspection of the diagonal and nearest-neighbor terms gives
\[
 \frac{\langle y^{(N)},J_{m_N}(b^{(N)})y^{(N)}\rangle}
      {\|y^{(N)}\|_2^2}
 =
 \frac{\langle x^{(N)},J_\Lambda(C)x^{(N)}\rangle}
      {\|x^{(N)}\|_2^2}.
\]
The LHS is at most
\(\lambda_{\max}\left(J_{m_N}(b^{(N)})\right)\), while \(\lambda_{\max}\left(J_{m_N}(b^{(N)})\right)\) is at most \(q_*\)
by definition. Then, \eqref{eq:inf-cutoff-Rayleigh} proves the
claim.
\end{proof}

Now, to obtain a nonzero \(\ell^2\)-limit from these finite paths, we need a
uniform lower bound on at least one coordinate of their normalized
top eigenvectors. The key to
obtaining this bound is uniform control of the paths' ``total variation''. For a path \(c=(c_0,\ldots,c_n)\), write
\[
        \operatorname{TV}(c)
        :=\sum_{i=1}^n|c_i-c_{i-1}|.
\]
Every monotone path from \(1\) to \(-1\) has total variation \(2\).
The next elementary bound converts this fact into a uniform lower bound
on the largest coordinate of an \(\ell^2\)-normalized nonnegative top
eigenvector of the corresponding Jacobi matrix.

\begin{lemma}
\label{lem:inf-variation-localization}
For every finite Jacobi path
\(c=(c_0,\ldots,c_n)\), and every
\(x\in\mathbb R^n\),
\begin{equation}
 \langle x,J_n(c)x\rangle
 \leq
 \frac14\|x\|_2^2
 +\frac12\operatorname{TV}(c)\|x\|_\infty^2.
\label{eq:inf-variation-bound}
\end{equation}
Consequently, if \(x\) is an $\ell^2$-normalized eigenvector corresponding to some eigenvalue \(q>1/4\), then
\begin{equation}
        \|x\|_\infty^2
        \geq
        \frac{2(q-1/4)}{\operatorname{TV}(c)}.
\label{eq:inf-max-coordinate-bound}
\end{equation}
\end{lemma}

\begin{proof}
For \(u,v\in[-1,1]\),
\begin{align*}
 d(u,v)
 -\left(\frac{u^2-1}{2}+\frac{v^2-1}{2}\right)
 &=-\frac{(u-v)^2}{2}+\frac{u-v}{2}\\
 &\leq\frac{|u-v|}{2}.
\end{align*}
Applying this bound (after rearranging) to each diagonal coefficient
\(d(c_{i-1},c_i)\), and using
\[
\frac12\sum_{i=1}^n
|c_i-c_{i-1}|x_i^2
\leq
\frac12\operatorname{TV}(c)\|x\|_\infty^2,
\]
while leaving the off-diagonal terms unchanged, gives
\begin{align}
\langle x,J_n(c)x\rangle
&\leq
\sum_{i=1}^n
\left(\frac{c_{i-1}^2-1}{2}
      +\frac{c_i^2-1}{2}\right)x_i^2
+2\sum_{i=1}^{n-1}a(c_i)x_ix_{i+1} \nonumber\\
&\quad+
\frac12\operatorname{TV}(c)\|x\|_\infty^2. \label{eq:98}
\end{align}
For \(i\in\{1,\ldots,n-1\}\), the terms involving
\(c_i\) in the first two sums depend only on \(x_i\) and \(x_{i+1}\).
Writing \(t=c_i\), their total contribution is
\[
\frac{t^2-1}{2}\bigl(x_i^2+x_{i+1}^2\bigr)
+2a(t)x_ix_{i+1}.
\]
The latter quadratic form is represented by the symmetric matrix
\[
B_t
:=
\begin{pmatrix}
\dfrac{t^2-1}{2} & a(t)\\[2mm]
a(t) & \dfrac{t^2-1}{2}
\end{pmatrix}.
\]
Writing \(z=\sqrt{1-t^2}\), we have
\[
\frac{t^2-1}{2}=-\frac{z^2}{2},
\qquad
a(t)=\frac z2,
\]
and therefore
\[
B_t
=
\begin{pmatrix}
-\dfrac{z^2}{2} & \dfrac z2\\[2mm]
\dfrac z2 & -\dfrac{z^2}{2}
\end{pmatrix}.
\]
The eigenvalues of this matrix are
\[
-\frac{z^2}{2}\pm\frac z2.
\]
Thus its largest eigenvalue satisfies
\[
-\frac{z^2}{2}+\frac z2
=
\frac18-\frac12\left(z-\frac12\right)^2
\leq\frac18.
\]
Hence,
\[
\frac{t^2-1}{2}\bigl(x_i^2+x_{i+1}^2\bigr)
+2a(t)x_ix_{i+1}
\leq
\frac18\bigl(x_i^2+x_{i+1}^2\bigr).
\]
Summing over \(i\in\{1,\ldots,n-1\}\) bounds the first two sums on the RHS of \eqref{eq:98} by
\[
        \frac18(x_1^2+x_n^2)
        +\frac14\sum_{i=2}^{n-1}x_i^2
        \leq\frac14\|x\|_2^2.
\]
The contributions corresponding to terms involving $c_0$ and $c_n$ are zero because \(c_0^2=c_n^2=1\).
This proves \eqref{eq:inf-variation-bound}.
\eqref{eq:inf-max-coordinate-bound} follows by rearranging it.
\end{proof}

We are now ready to use these monotone paths from Lemma~\ref{prop:inf-monotone-approximants} to construct a genuine eigenvector in \(\ell^2(\mathbb Z)\) with eigenvalue \(q_*\). For each finite path, we will choose an \(\ell^2\)-normalized nonnegative eigenvector corresponding to its largest eigenvalue, and reindex it so that its largest coordinate occurs at \(0\).  The preceding estimate bounds this coordinate uniformly away from zero, so the eventual
coordinatewise limit cannot be the zero vector.

\begin{proposition}
\label{prop:inf-l2-eigenvector}
There exist a sequence
\[
c=(c_i)_{i\in\mathbb Z}\in[-1,1]^{\mathbb Z}
\]
and a nonzero nonnegative vector
\(\lambda\in\ell^2(\mathbb Z)\) such that
\begin{equation}
        J_\infty(c)\lambda=q_*\lambda.
\label{eq:inf-l2-eigen-equation}
\end{equation}
\end{proposition}

\begin{proof}
Let \(b^{(k)}\) be the monotone paths given by
Proposition~\ref{prop:inf-monotone-approximants}, let
\[
        q_k:=\lambda_{\max}\left(J_{m_k}(b^{(k)})\right)\rightarrow q_*,
\]
and choose normalized nonnegative top eigenvectors \(v^{(k)}\).
Since \(\operatorname{TV}(b^{(k)})=2\),
Lemma~\ref{lem:inf-variation-localization} gives
\[
        \|v^{(k)}\|_\infty^2\geq q_k-\frac14.
\]
Since $q_*>\frac14$ by Lemma~\ref{lem:inf-quarter-witness}, there is an \(\eta>0\) such that
\begin{equation}
        \|v^{(k)}\|_\infty\geq\eta
\label{eq:inf-uniform-maximum}
\end{equation}
for all sufficiently large \(k\).

Discard finitely many terms so that both this estimate and \(q_k>0\)
hold for every remaining \(k\).
Pad \(b^{(k)}\) by \(k\) repeated \(+1\) labels on the left and \(k\)
repeated \(-1\) labels on the right.  As before, this adds two zero
blocks and does not change \(q_k>0\).  Extend \(v^{(k)}\) by zeros, choose
a maximal coordinate, and translate it to \(0\).  Denote the padded and shifted sequence by
\(\widetilde c^{(k)}\) and the normalized vector by \(\widetilde v^{(k)}\). After shifting, regard \(\widetilde v^{(k)}\) as an element of
\(\ell^2(\mathbb Z)\) by setting
\(\widetilde v_i^{(k)}=0\) outside the coordinate interval of the padded
path. The
distance from \(0\) to both endpoints tends to infinity, and
\[
\widetilde v_0^{(k)}\geq\eta,
\qquad
\sum_{i\in\mathbb Z}\bigl|\widetilde v_i^{(k)}\bigr|^2=1.
\]
By a nested-subsequence argument, we may pass to a single subsequence
such that, for every fixed \(i\in\mathbb Z\),
\[
\widetilde c_i^{(k)}\rightarrow c_i,
\qquad
\widetilde v_i^{(k)}\rightarrow\lambda_i.
\]
Equation~\eqref{eq:inf-uniform-maximum} gives
\(\lambda_0\geq\eta>0\).  Moreover, Fatou's lemma gives
\[
\sum_{i\in\mathbb Z}|\lambda_i|^2
\leq
\liminf_{k\to\infty}
\sum_{i\in\mathbb Z}\bigl|\widetilde v_i^{(k)}\bigr|^2
=
1.
\]
Hence \(0\neq\lambda\in\ell^2(\mathbb Z)\).

For any fixed \(i\), the eigenvalue equation is valid at
\(i\) for all large enough \(k\), since the endpoints of the padded paths tend to infinity. Passing to the limit gives
\eqref{eq:inf-l2-eigen-equation}. 
\end{proof}

\subsection{The infinite-dimensional strategy}

We now have all we need to construct an infinite-dimensional strategy attaining $q_*$. Normalize the eigenvector from Proposition~\ref{prop:inf-l2-eigenvector}
so that \(\|\lambda\|_2=1\). The strategy is the natural infinite-dimensional extension of the canonical strategy from Section~\ref{sec:path-to-strategy}.
\begin{definition}
\label{def:inf-staggered-strategy}
Let $c=(c_i)_{i\in\mathbb Z}\in[-1,1]^{\mathbb Z}$
and \(0\neq\lambda\in\ell^2(\mathbb Z)\) be as in
Proposition~\ref{prop:inf-l2-eigenvector}, normalized so that
\(\|\lambda\|_2=1\). Let
\[
        \mathcal H_A=\mathcal H_B=\ell^2(\mathbb Z),
\]
with standard basis \((e_i)_{i\in\mathbb Z}\), and let
\[
        \ket{\psi}=\sum_{i\in\mathbb Z}\lambda_i e_i\otimes e_i.
\]
For \(s_i:=\sqrt{1-c_i^2}\), define the diagonal operators
\[
 Z_Ae_i=
 \begin{cases}
        c_{i-1}e_i,&i\text{ odd},\\
        -c_i e_i,&i\text{ even},
 \end{cases}
 \qquad
 Z_Be_i=
 \begin{cases}
        c_i e_i,&i\text{ odd},\\
        -c_{i-1}e_i,&i\text{ even}.
 \end{cases}
\]
Define \(X_A\) on the even pairs by
\[
        X_Ae_i=s_i e_{i+1},
        \qquad
        X_Ae_{i+1}=s_i e_i
        \qquad(i\text{ even}),
\]
and \(X_B\) on the odd pairs by
\[
        X_Be_i=s_i e_{i+1},
        \qquad
        X_Be_{i+1}=s_i e_i
        \qquad(i\text{ odd}).
\]
Let \(W_A\) and \(W_B\) be such that
\[
\begin{aligned}
 W_Ae_i&=e_{i+1},&\quad W_Ae_{i+1}&=e_i &&(i\text{ odd}),\\
 W_Be_i&=e_{i+1},&\quad W_Be_{i+1}&=e_i &&(i\text{ even}).
\end{aligned}
\]
Finally, set
\[
 A_1=\frac{I+Z_A-X_A}{2},
 \qquad
 A_2=\frac{I+Z_A+X_A}{2},
 \qquad
 A_3=\frac{I+W_A}{2},
\]
and
\[
 B_1=\frac{I+Z_B-X_B}{2},
 \qquad
 B_2=\frac{I+Z_B+X_B}{2},
 \qquad
 B_3=\frac{I+W_B}{2}.
\]
\end{definition}

\begin{lemma}
\label{lem:inf-staggered-projections}
The six operators in Definition~\ref{def:inf-staggered-strategy} are
orthogonal projections.
\end{lemma}

\begin{proof}
All the displayed operators are bounded and self-adjoint.  On an even
pair \((i,i+1)\),
\[
        Z_A=
        \begin{pmatrix}-c_i&0\\0&c_i\end{pmatrix},
        \qquad
        X_A=
        \begin{pmatrix}0&s_i\\s_i&0\end{pmatrix}.
\]
Therefore
\[
        Z_AX_A+X_AZ_A=0,
        \qquad
        Z_A^2+X_A^2=I.
\]
The same calculation on the odd pairs gives the analogous identities
for \(Z_B,X_B\).  Consequently \(Z_A\pm X_A\) and \(Z_B\pm X_B\) are
self-adjoint unitaries.  Also \(W_A^2=W_B^2=I\).  Each of the six
operators is therefore of the form \((I+R)/2\) for a self-adjoint
unitary \(R\), and hence is an orthogonal projection.
\end{proof}

\begin{theorem}[Infinite-dimensional attainment]
\label{thm:inf-attainment}
For the strategy in Definition~\ref{def:inf-staggered-strategy},
\begin{equation}
        \langle\psi| I_{3322}|\psi\rangle
        =
        \langle\lambda, J_\infty(c)\lambda\rangle = q_* \,.
\label{eq:inf-Bell-Jacobi-identity}
\end{equation}
\end{theorem}

\begin{proof}
We have
\[
 I_{3322}
 =Z_A\otimes Z_B
 +\frac12Z_A\otimes I
 -\frac12I\otimes Z_B
 +\frac12X_A\otimes W_B
 +\frac12W_A\otimes X_B
 -I\otimes I.
\]
All sums below converge absolutely.  Indeed,
\[
        \sum_i\lambda_i^2=1,
        \qquad
        \sum_i|\lambda_i\lambda_{i+1}|
        \leq\frac12\sum_i(\lambda_i^2+\lambda_{i+1}^2)=1.
\]
Let \(z_i^A,z_i^B\) be the diagonal entries of \(Z_A,Z_B\).  The
diagonal part of the expectation is
\[
 \sum_i\lambda_i^2
 \left(z_i^Az_i^B+\frac12z_i^A-\frac12z_i^B-1\right).
\]
If \(i\) is odd, \((z_i^A,z_i^B)=(c_{i-1},c_i)\); if \(i\) is even,
\((z_i^A,z_i^B)=(-c_i,-c_{i-1})\).  In both cases, using
\(d(-v,-u)=d(u,v)\), the bracket equals
\(d(c_{i-1},c_i)\).  Thus the diagonal contribution is
\[
        \sum_i d(c_{i-1},c_i)\lambda_i^2.
\]

For any two of the operators \(R\) and \(S\) defined above, directly
expanding
\[
\ket{\psi}=\sum_i\lambda_i e_i\otimes e_i
\]
gives
\[
\langle\psi,R\otimes S\,\psi\rangle
=
\sum_{i,j}\lambda_i\lambda_j R_{ij}S_{ij},
\]
where \(R_{ij}=\langle e_i,Re_j\rangle\) and similarly for \(S_{ij}\).
If \(i\) is even, the edge \((i,i+1)\) contributes
\[
 \frac12\bigl(2s_i\lambda_i\lambda_{i+1}\bigr)
 =2a(c_i)\lambda_i\lambda_{i+1}
\]
through \(X_A\otimes W_B\).  If \(i\) is odd, it contributes the same
quantity through \(W_A\otimes X_B\).  There are no other off-diagonal
terms.  Hence
\[
 \langle\psi | I_{3322} | \psi\rangle
 =\sum_i d(c_{i-1},c_i)\lambda_i^2
  +2\sum_i a(c_i)\lambda_i\lambda_{i+1} = \langle\lambda,J_\infty(C)\lambda\rangle \,.
\]
The latter is $q_*$ by the choice of $\lambda$ and $c$ from Proposition~\ref{prop:inf-l2-eigenvector}.
\end{proof}

\begin{corollary}[Infinite-dimensional value]
\label{cor:infinite-value}
If \(\omega_\infty(I_{3322})\) denotes the supremum over strategies on
separable Hilbert spaces, then
\[
        \omega_\infty(I_{3322})=q_*,
\]
and the supremum is attained.
\end{corollary}

\begin{proof}
Theorem~\ref{thm:inf-attainment} gives the
attainment, and Theorem~\ref{thm:universal-upper} gives the upper bound.
\end{proof}

\section{The final theorem}
\begin{theorem}
Let $\omega_d(I_{3322})$ be the supremum of the value of $I_{3322}$ over strategies of dimension $d$. Then, $$\omega_d(I_{3322})<q_*$$ 
for every finite $d$. Let $\omega_\infty(I_{3322})$ be the supremum over separable infinite-dimensional strategies. Then, 
$$\omega_\infty(I_{3322})=q_*\,,$$ 
and the supremum is attained.
\end{theorem}

\begin{proof}
Theorem~\ref{thm:finite-nonattainment} gives $\omega_d(I_{3322})<q_*$ for every finite $d$. Corollary~\ref{cor:infinite-value} gives that $\omega_\infty(I_{3322})=q_*$ and that the supremum is attained.
\end{proof}

\bibliographystyle{alpha}
\bibliography{references}

@article{froissart1981constructive,
  title={Constructive generalization of Bell’s inequalities},
  author={Froissart, Marcel},
  journal={Il Nuovo Cimento B (1971-1996)},
  volume={64},
  number={2},
  pages={241--251},
  year={1981},
  publisher={Springer}
}

@article{PV10,
  title={Maximal violation of a bipartite three-setting, two-outcome Bell inequality using infinite-dimensional quantum systems},
  author={P{\'a}l, K{\'a}roly F and V{\'e}rtesi, Tam{\'a}s},
  journal={Physical Review A—Atomic, Molecular, and Optical Physics},
  volume={82},
  number={2},
  pages={022116},
  year={2010},
  publisher={APS}
}

@article{coladangelo2020inherently,
  title={An inherently infinite-dimensional quantum correlation},
  author={Coladangelo, Andrea and Stark, Jalex},
  journal={Nature communications},
  volume={11},
  number={1},
  pages={3335},
  year={2020},
  publisher={Nature Publishing Group UK London}
}

@article{dykema2019non,
  title={Non-closure of the set of quantum correlations via graphs},
  author={Dykema, Ken and Paulsen, Vern I and Prakash, Jitendra},
  journal={Communications in Mathematical Physics},
  volume={365},
  number={3},
  pages={1125--1142},
  year={2019},
  publisher={Springer}
}

@inproceedings{slofstra2019set,
  title={The set of quantum correlations is not closed},
  author={Slofstra, William},
  booktitle={Forum of Mathematics, Pi},
  volume={7},
  pages={e1},
  year={2019},
  organization={Cambridge University Press}
}

@article{coladangelo2020two,
  title={A two-player dimension witness based on embezzlement, and an elementary proof of the non-closure of the set of quantum correlations},
  author={Coladangelo, Andrea},
  journal={Quantum},
  volume={4},
  pages={282},
  year={2020},
  publisher={Verein zur F{\"o}rderung des Open Access Publizierens in den Quantenwissenschaften}
}

@article{BarthMartinWilkinson1967,
  title={Handbook series linear algebra},
  author={Barth, W and Martin, RS and Wilkinson, JH},
  journal={Numerische Mathematik},
  volume={9},
  pages={386--393},
  year={1967}
}

@article{nazarov1997hunt,
  title={The hunt for a Bellman function: applications to estimates of singular integral operators and to other classical problems in harmonic analysis},
  author={Nazarov, F and Treil, S},
  journal={St Petersburg Mathematical Journal},
  volume={8},
  number={5},
  pages={721--824},
  year={1997},
  publisher={Providence, RI, USA: American Mathematical Society, c1992-}
}

@inproceedings{nazarov2001bellman,
  title={Bellman function in stochastic control and harmonic analysis},
  author={Nazarov, Fedor and Treil, Sergei and Volberg, Alexander},
  booktitle={Systems, Approximation, Singular Integral Operators, and Related Topics: International Workshop on Operator Theory and Applications, IWOTA 2000},
  pages={393--423},
  year={2001},
  organization={Springer}
}

@article{bellman1952theory,
  title={On the theory of dynamic programming},
  author={Bellman, Richard},
  journal={Proceedings of the national Academy of Sciences},
  volume={38},
  number={8},
  pages={716--719},
  year={1952}
}

@article{scholz2008tsirelson,
  title={Tsirelson's problem},
  author={Scholz, Volkher B and Werner, Reinhard F},
  journal={arXiv preprint arXiv:0812.4305},
  year={2008}
}

@article{pauwels26,
  title={The quantum supremum of the I3322 Bell inequality is not attained in finite dimension},
  author={Pauwels, Jef},
  journal={arXiv preprint arXiv:2608.29734},
  year={2026}
}

@article{douglas2026i3322,
  author    = {Seth Douglas},
  title     = {The I3322 quantum value: certified window, spatial attainment, finite-dimensional nonattainment, and the Cqs/Cq separation},
  year      = {2026},
  publisher = {Zenodo},
  journal   = {10.5281/zenodo.21843326}
}

@article{douglas26arxiv,
  title={The I3322 quantum value is attained spatially but not in finite dimension},
  author={Douglas, Seth},
  journal={arXiv preprint arXiv:2609.05555},
  year={2026}
}

@article{mghirbi2026i3322,
  author    = {Nidhal Mghirbi},
  title     = {Critical Bellman–Jacobi dynamics for reversible Bell functionals},
  year      = {2026},
  publisher = {Zenodo},
  journal   = {10.5281/zenodo.21910040}
}

@article{beigi2021separation,
  title={Separation of quantum, spatial quantum, and approximate quantum correlations},
  author={Beigi, Salman},
  journal={Quantum},
  volume={5},
  pages={389},
  year={2021},
  publisher={Verein zur F{\"o}rderung des Open Access Publizierens in den Quantenwissenschaften}
}

\appendix

\section{Proof of Lemma~\ref{lem:structure}}
\label{sec:app}
We prove Lemma~\ref{lem:structure}. Define
\[
        R_+:=Z+X,
        \qquad
        R_-:=Z-X.
\]
The identities in the lemma statement imply
\[
        R_\pm^2
        =
        Z^2+X^2\pm(ZX+XZ)
        =
        I.
\]
Thus, \(R_+\) and \(R_-\) are self-adjoint unitaries, and
\[
        P_+:=\frac{I+R_+}{2},
        \qquad
        P_-:=\frac{I+R_-}{2}
\]
are orthogonal projections.  Conversely,
\[
        Z=P_+ + P_- -I,
        \qquad
        X= P_+ - P_-.
\]

By Jordan's lemma (applied to the projections $P_+$ and $P_-$), the Hilbert space has an orthogonal decomposition
into subspaces of dimension at most two that are invariant under both
\(P_+\) and \(P_-\).  The same subspaces are therefore invariant under
\(Z\) and \(X\).  If a two-dimensional subspace contains a
one-dimensional subspace that is invariant under both \(P_+\) and \(P_-\), then, since the
projections are self-adjoint, its orthogonal complement is also
invariant. We may therefore split that subspace into two
one-dimensional blocks.  Hence, we may assume that every remaining two-dimensional block contains no one-dimensional subspace that is invariant under both \(Z\) and~\(X\).

We now determine the possible forms of these blocks.

\noindent \emph{One-dimensional blocks.}
On a one-dimensional block, \(Z\) and \(X\) act as multiplication by
real numbers \(z\) and \(x\).  The two identities become
\[
        2zx=0,
        \qquad
        z^2+x^2=1.
\]
Thus either
\[
        z\in\{\pm1\},
        \qquad
        x=0,
\]
or
\[
        z=0,
        \qquad
        x\in\{\pm1\}.
\]
These are precisely the two types of one-dimensional blocks in the
statement.

\medskip
\noindent
\emph{Two-dimensional blocks.}
Let \(\mathcal K\) be one of the remaining two-dimensional blocks. The restriction \(Z|_{\mathcal K}\) cannot be the zero operator.
Indeed, if \(Z|_{\mathcal K}=0\), then we could choose an eigenvector of
the self-adjoint operator \(X|_{\mathcal K}\) and its span would be
invariant under both \(Z\) and \(X\). Since \(Z|_{\mathcal K}\) is self-adjoint and nonzero, it has a nonzero
eigenvalue \(t\).  Choose a corresponding unit eigenvector
\(e_\alpha\), so that
\[
        Ze_\alpha=t e_\alpha,
        \qquad t\neq0.
\]
We have
\[
\begin{aligned}
        \|Xe_\alpha\|^2
        &=
        \langle e_\alpha,X^2e_\alpha\rangle\\
        &=
        \langle e_\alpha,(I-Z^2)e_\alpha\rangle\\
        &=
        1-t^2.
\end{aligned}
\]
Hence \(|t|\leq1\).  We cannot have \(|t|=1\), since then
\(Xe_\alpha=0\), so \(\operatorname{span}\{e_\alpha\}\) would be invariant under both
\(Z\) and \(X\).  Thus
\[
        0<|t|<1.
\]

Set
\[
        s:=\sqrt{1-t^2}>0
\]
and define
\[
        e_{\bar\alpha}:=\frac{1}{s}Xe_\alpha.
\]
This is a unit vector.  Moreover, the anticommutation relation gives
\[
        Ze_{\bar\alpha}
        =
        \frac{1}{s}ZXe_\alpha
        =
        -\frac{1}{s}XZe_\alpha
        =
        -t e_{\bar\alpha}.
\]
Since \(Z\) is self-adjoint and \(t\neq -t\), the vectors
\(e_\alpha\) and \(e_{\bar\alpha}\) are orthogonal.  They therefore form
an orthonormal basis of \(\mathcal K\).  By construction,
\[
        Xe_\alpha=s e_{\bar\alpha},
\]
while
\[
        Xe_{\bar\alpha}
        =
        \frac{1}{s}X^2e_\alpha
        =
        \frac{1-t^2}{s}e_\alpha
        =
        s e_\alpha.
\]
Consequently, with respect to the basis
\((e_\alpha,e_{\bar\alpha})\),
\[
        Z|_{\mathcal K}
        =
        \begin{pmatrix}
                t&0\\
                0&-t
        \end{pmatrix},
        \qquad
        X|_{\mathcal K}
        =
        \begin{pmatrix}
                0&\sqrt{1-t^2}\\
                \sqrt{1-t^2}&0
        \end{pmatrix}.
\]
Setting \(z_\alpha=t\) gives the required two-dimensional block.  By
interchanging \(e_\alpha\) and \(e_{\bar\alpha}\) if necessary, one may
in fact arrange that \(z_\alpha\in(0,1)\). Taking the union of the orthonormal bases chosen on all the Jordan
blocks gives the desired orthonormal basis of the full Hilbert space.

\end{document}